\documentclass[english]{article}
\usepackage[hyperindex,breaklinks,colorlinks,citecolor=blue,urlcolor=blue]{hyperref}
\usepackage[T1]{fontenc}
\usepackage[latin1]{inputenc}
\usepackage{geometry}
\usepackage[longnamesfirst]{natbib}
\usepackage{float}
\usepackage{mathrsfs}
\usepackage{amsmath}
\numberwithin{equation}{section}
\usepackage{amssymb}
\usepackage{graphicx}
\usepackage{caption}
\usepackage{subcaption}
\usepackage{bbm}
\usepackage{amsthm}
\usepackage{amsfonts}
\usepackage{epsfig}
\usepackage{bm}
\usepackage{enumerate}
\usepackage[all]{xy}

\makeatletter

\floatstyle{ruled}
\newfloat{algorithm}{tbp}{loa}
\providecommand{\algorithmname}{Algorithm}
\floatname{algorithm}{\protect\algorithmname}

\newtheorem{theorem}{Theorem}[section]
\newtheorem{lem}{Lemma}[section]
\newtheorem{rem}{Remark}[section]
\newtheorem{prop}{Proposition}[section]

\newtheorem{ass}{Assumption}[section]

\newcounter{hypA}

\usepackage{babel}\date{}

\usepackage{babel}

\makeatother

\usepackage{babel}

\newcommand{\Y}{\vct{Y}}

\newcommand{\T}{\mathcal{T}}
\newcommand{\olTset}{\ol{\mathcal{T}}}
\newcommand{\LL}{\mathcal{L}}
\newcommand{\B}{\mathcal{B}}

\newcommand{\D}{\bm{D}}

\newcommand{\E}{\mathbb{E}}

\newcommand{\PP}{\mathbb{P}}

\DeclareMathOperator*{\argmin}{arg\,min}

\newcommand{\ol}{\overline}
\newcommand{\ul}{\underline}

\newcommand{\R}{\mathbb{R}}
\newcommand{\C}{\mathbb{C}}
\newcommand{\Z}{\mathbb{Z}}

\newcommand{\N}{\mathbb{N}}

\newcommand{\argmax}{\mathrm{arg\,max}}

\newcommand{\vct}[1]{\bm{#1}}
\newcommand{\mtx}[1]{\bm{#1}}
\newcommand{\dtv}{d_{\mathrm{TV}}}
\newcommand{\dW}{d_{\mathrm{W}}}

\newcommand{\hM}{\widehat{M}}

\newcommand{\grad}{\nabla}

\newcommand{\defby}{\mathrel{\mathop:}=}

\newcommand{\Yok}{\vct{Y}_{0:k}}

\newcommand{\BR}{\mathcal{B}_R}

\newcommand{\musm}{\mu^{\mathrm{sm}}}
\newcommand{\musmG}{\mu^{\mathrm{sm}}_{\mathcal{G}}}
\newcommand{\musmGBR}{\mu^{\mathrm{sm}}_{\mathcal{G}|\BR}}
\newcommand{\tmusm}{\tilde{\mu}^{\mathrm{sm}}}
\newcommand{\mufi}{\mu^{\mathrm{fi}}}
\newcommand{\mufiG}{\mu^{\mathrm{fi}}_{\mathcal{G}}}
\newcommand{\etafiG}{\eta^{\mathrm{fi}}_{\mathcal{G}}}
\newcommand{\tetafiG}{\tilde{\eta}^{\mathrm{fi}}_{\mathcal{G}}}
\newcommand{\lsm}{l^{\mathrm{sm}}}
\newcommand{\lsmG}{l^{\mathrm{sm}}_{\mathcal{G}}}
\newcommand{\umean}{\ol{\vct{u}}^{\mathrm{sm}}}
\newcommand{\ufimean}{\ol{\vct{u}}^{\mathrm{fi}}}
\newcommand{\uMAPsm}{\hat{\vct{u}}^{\mathrm{sm}}_{\mathrm{MAP}}}
\newcommand{\uMAPfi}{\hat{\vct{u}}^{\mathrm{fi}}}
\newcommand{\Cder}{C_{\mathrm{der}}}
\newcommand{\CJ}{C_{\mtx{J}}}
\newcommand{\Bone}{\mathcal{B}_1}

\newcommand{\J}{\mtx{J}}
\newcommand{\olT}{\ol{T}(\u)}
\newcommand{\vZ}{\vct{Z}}
\renewcommand{\l}{\left}
\renewcommand{\r}{\right}
\newcommand{\f}{\frac}

\renewcommand{\phi}{\varphi}
\renewcommand{\epsilon}{\varepsilon}
\renewcommand{\H}{\mtx{H}}

\renewcommand{\u}{\vct{u}}
\newcommand{\uG}{\vct{u}^{\mathcal{G}}}
\renewcommand{\v}{\vct{v}}
\newcommand{\w}{\vct{w}}
\newcommand{\s}{\vct{s}}
\newcommand{\mr}{\mathrm}
\newcommand{\mc}{\mathcal}
\newcommand{\Akfi}{\mtx{A}_k^{\mathrm{fi}}}

\newcommand{\Pcu}[1]{\PP\left(\left.#1 \right|\u\right)}
\newcommand{\Ecu}[1]{\E\left(\left.#1 \right|\u\right)}
\newcommand{\x}{\vct{x}}
\renewcommand{\c}{\vct{c}}
\newcommand{\y}{\vct{y}}
\newcommand{\hk}{\hat{k}}

\newcommand{\mM}{\mtx{M}^{(j_{\max}|\hk)}}
\newcommand{\jmax}{j_{\max}}
\newcommand{\imax}{i_{\max}}
\newcommand{\gsm}{g^{\mr{sm}}}

\newcommand{\CM}{C_{\mtx{M}}^{(l|j_{\max}|\hk)}}

\begin{document}

\begin{center}
{\Large \textbf{Optimization Based Methods for Partially Observed Chaotic Systems}}\\

\vspace{0.5cm}

BY DANIEL PAULIN$^{1}$, AJAY JASRA$^{1}$,  DAN CRISAN$^{2}$ \&  ALEXANDROS BESKOS$^{3}$

{\footnotesize $^{1}$Department of Statistics \& Applied Probability,
National University of Singapore, Singapore, 117546, SG.}
{\footnotesize E-Mail:\,}\texttt{\emph{\footnotesize paulindani@gmail.com, staja@nus.edu.sg}}\\
{\footnotesize $^{2}$Department of Mathematics,
Imperial College London, London, SW7 2AZ, UK.}\\
{\footnotesize E-Mail:\,}\texttt{\emph{\footnotesize d.crisan@ic.ac.uk}}\\
{\footnotesize $^{3}$Department of Statistical Science,
University College London, London, WC1E 6BT, UK.}
{\footnotesize E-Mail:\,}\texttt{\emph{\footnotesize a.beskos@ucl.ac.uk}}
\end{center}

\begin{abstract}
In this paper we consider filtering and smoothing of partially observed chaotic dynamical
systems that are discretely observed, with an additive Gaussian noise in the observation.
These  models are found in a wide variety of real applications and include the Lorenz 96' model.
In the context of a fixed observation interval $T$, observation time step $h$ and Gaussian
observation variance $\sigma_Z^2$, we show under assumptions that the filter and smoother are well approximated by a Gaussian with high probability when $h$ and $\sigma^2_Z h$ are sufficiently small. Based on this result we show that the Maximum-a-posteriori (MAP) estimators are asymptotically optimal in mean square
error as $\sigma^2_Z h$ tends to $0$. Given these results, we provide a batch algorithm for the smoother
and filter, based on Newton's method, to obtain the MAP. In particular, we show that if the initial point is close
enough to the MAP, then Newton's method converges to it at a fast rate. We also provide a method for
computing such an initial point.  These results contribute to the theoretical understanding of widely used 4D-Var data assimilation method. Our approach is illustrated numerically on the Lorenz 96' model with state vector up to 1 million dimensions, with code running in the order of minutes. To our knowledge the results in this paper are the first of their type for this class of models.
\end{abstract}

\section{Introduction}

Filtering and smoothing are amongst the most important problems for several applications, featuring contributions from mathematics, statistics, engineering and many more fields;
see for instance \cite{dan} and the references therein. The basic notion of such models, is the idea of an unobserved stochastic process, that is observed indirectly
by data. The most typical model is perhaps where the unobserved stochastic process is a Markov chain, either in discrete time, or a diffusion process.
In this paper, we are mainly concerned with the scenario when the unobserved dynamics are deterministic and moreover chaotic. The only randomness in the unobserved
system is uncertainty in the initial condition and it is this quantity that we wish to infer, on the basis of discretely and sequentially observed data; we explain the difference between
filtering and smoothing in this context below. This class of problems has slowly become more important in the literature, particularly in the area of data assimilation (\cite{Dataassimilation}).
The model itself has a substantial number of practical applications, including weather prediction, oceanography and oil reservoir simulation, see for instance \cite{kalnay}.

In this paper we consider smoothing and filtering for partially observed deterministic dynamical systems of the general form
\begin{equation}\label{diffeqgeneralform}
\frac{d \u }{d t}=-\mtx{A} \u -\mtx{B}(\u,\u)+\vct{f},
\end{equation} 
where $\u: \R^+\to \R^d$ is a dynamical system in $\R^d$ for some $d\in \Z_+$, $\mtx{A}$ is linear operator in $\R^d$ (i.e. $\mtx{A}$ is a $d\times d$ matrix),  $\vct{f}\in \R^d$ is a constant vector, and $\mtx{B}(\u,\u)$ is a bilinear form corresponding to the nonlinearity (i.e. $\mtx{B}$ is a $d\times d\times d$ array).  We denote the solution of equation \eqref{diffeqgeneralform} with initial condition $\u(0)\defby \v$ for $t\ge 0$ by $\v(t)$. The derivatives of the solution $\v(t)$ at time $t=0$ will be denoted by 
\begin{equation}
\D^i \v\defby \l.\frac{d^i \v(t)}{dt^i}\r|_{t=0} \text{ for }i\in \N,
\end{equation}
in particular, $\D^0\v=\v$ , $\D\v\defby \D^1\v=-\mtx{A} \v -\mtx{B}(\v,\v)+\vct{f}$ (the right hand side of \eqref{diffeqgeneralform}), and $\D^2\v=-\mtx{A} \D^1 \v-\mtx{B}(\D^1 \v,\v)-\mtx{B}(\v,\D^1 \v)$.

In order to ensure the existence of a solution to the equation \eqref{diffeqgeneralform} for every $t\ge 0$, we assume that there are constants $R>0$ and $\delta>0$ such that 
\begin{equation}\label{eqtrappingball}
\l<\D\v,\v\r> \le 0 \text{ for every }\v\in \R^d \text{ with }\|\v\|\in [R,R+\delta].
\end{equation}
We call this the \emph{trapping ball} assumption. Let $\BR\defby \{\v\in \R^d: \|\v\|\le R\}$ be the ball of radius $R$. Using the fact that $\l<\frac{d}{dt} \v(t),\v(t)\r>= \frac{1}{2}\frac{d }{d t}\|\v(t)\|^2$, one can show that the solution to \eqref{diffeqgeneralform} exists for $t\ge 0$ for every $\v\in \BR$, and satisfies that $\v(t)\in \BR$ for $t\ge 0$.

Equation \eqref{diffeqgeneralform} was shown in \citet*{AlonsoStuartLongtime} and \citet*{Dataassimilation} to be applicable to three chaotic dynamical systems, the Lorenz 63' model, the Lorenz 96' model, and the Navier-Stokes equation on the torus; such models have many applications. 
We note that instead of the trapping ball assumption, these papers have considered different assumptions on $\mtx{A}$ and $\mtx{B}(\v,\v)$. As we shall explain in Section \ref{SecPreliminaries}, their assumptions imply \eqref{eqtrappingball},   thus the trapping ball assumption is more general.

We assume that the system is observed at time points $t_j=j h$ for $j=0,1,\ldots$, with observations 
\[
\vct{Y}_j\defby \mtx{H} \u(t_j) + \vct{Z}_j
\]
where $\mtx{H}: \R^d \to \R^{d_o}$ is a linear operator, and $(\vct{Z}_j)_{j\ge 0}$ are i.i.d.~centered random vectors taking values in $\R^{d_o}$ describing the noise. We assume that these vectors have distribution $\eta$ that is Gaussian with i.i.d.~components of variance $\sigma_Z^2$. \footnote{We believe that our results in this paper hold for non-Gaussian noise distributions as well, but proving this would be technically complex.}


The contributions of this article are as follows. In the context of a fixed observation interval $T$, we show under assumptions that the filter and smoother are well approximated by a Gaussian law
when $\sigma^2_Z h$ is sufficiently small. Our next result, using the ideas of the first one, shows that the Maximum-a-posteriori (MAP) estimators (of the filter and smoother) are asymptotically optimal in mean square
error when $\sigma^2_Z h$ tends to $0$. The main practical implication of these mathematical results is that we
can then provide a batch algorithm for the smoother and filter, based on Newton's method, to obtain the MAP. In particular, we prove that if the initial point is close enough to the MAP, then Newton's method converges to it at a fast rate. We also provide a method for computing such an initial point, and prove error bounds for it. Our approach is illustrated numerically on the Lorenz 96' model with state vector up to 1 million dimensions. We believe that the method  of this paper has a wide range of potential applications in meteorology, but we only include one example due to space considerations.

We note that in this paper, we consider finite dimensional models. There is a substantial interest in the statistics literature in recent years in non-parametric inference for infinite dimensional PDE models, see \cite{dashti2017bayesian} for an overview and references, and \cite{Nicklinfinitedimstatmodels} for a comprehensive monograph on the mathematical foundations of infinite dimensional statistical models. This approach can result in MCMC algorithms that are robust with respect to the refinement of the discretisation level, see e.g. \cite{MCMCforfunctions, Cotterrandwalk, nateshstuartthiery, VollmerdimindepMCMC, cui2016dimension, TVGaussian}. There are also other randomisation and optimization based methods that have been recently proposed in the literature, see e.g. \cite{Randomizethenoptimize, RandomizeMAP}.

A key property of these methods is that the prior is defined on the function space, and the discretizations automatically define corresponding prior distributions with desirable statistical properties in a principled manner. This is related to modern Tikhonov-Phillips regularisation methods widely used in applied mathematics, see \cite{benning2018modern} for a comprehensive overview. In the context of infinite dimensional models, MAP estimators are non-trivial to define in a mathematically precise way on the infinite dimensional function space, but several definitions of MAP estimators, various weak consistency results under the small noise limit, and posterior contraction rates have been shown in recent years, see e.g. \cite{Cotterfunctions, DashtiLawStuart, Vollmerpostconsistency, HelinBurgerMAP, PosteriorconsistencyHanneLassasSamuli, monard2017efficient, nickl2017schrodinger, dunlop2016map}. Some other important work on similar models and/or associated filtering/smoothing algorithms include \cite{DiscreteHaydenOlsonTiti, blomker, law2016filter}. These results are very interesting from a mathematical and statistical point of view, however the intuitive meaning of some of the necessary conditions, and their algorithmic implications are difficult to grasp.

In contrast with these works, our results in this paper concern the finite dimensional setting that is the most frequently used one in the data assimilation community. 
By working in finite dimensions, we are able to show consistency results and convergence rates for the MAP estimators under small observation noise / high observation frequency limits under rather weak assumptions (in particular, in Section 4.4 of \cite{ConcentrationProperties} our key assumption on the dynamics was verified in 100 trials when $\mtx{A}$, $\mtx{B}$ and $\vct{f}$ were randomly chosen chosen, and only the first component of the system was observed, and they were always satisfied). 
Moreover, previous work in the literature has not said anything about the computational complexity of actually finding the MAP estimators, which is a non-trivial problem in non-linear setting due to the existence of local maxima for the log-likelihood. In our paper we propose appropriate initial estimators, and show that Newton's method started from them converges to the true MAP with high probability in the small noise/high observation frequency scenario when started from this initial estimator.

It is important to mention that the MAP estimator forms the basis of the 4D-Var method introduced in \cite{le1986variational, talagrand1987variational} that is widely used in weather forecasting. A key methodological innovation of this method is that the gradients of the log-likelihood are computed via the adjoint equations, so that each gradient evaluation takes a similar amount of computation effort as a single run of the model. This has allowed the application of the method on large scale models with up to $d=10^9$ dimensions. See \cite{dimet2005deterministic} for some theoretical results, and \cite{navon2009data, bannister2016review} for an overview of some recent advances. The present paper offers rigorous statistical foundations for this method for the class of non-linear systems defined by \eqref{diffeqgeneralform}.

The structure of the paper is as follows. In Section \ref{SecPreliminaries}, we state some preliminary results for systems of the type \eqref{diffeqgeneralform}. Section \ref{secmainresults} contains our main results: Gaussian approximations, asymptotic optimality of MAP estimators, and approximation of MAP estimators via Newton's method with precision guarantees. In Section \ref{secapplications} we apply our algorithm to the Lorenz 96' model. Section \ref{secpreliminaries} contains some preliminary results, and Section \ref{secproofmain} contains the proofs of our main results. Finally, the Appendix contains the proofs of our preliminary results based on concentration inequalities for empirical processes.

\subsection{Preliminaries}\label{SecPreliminaries}

Some notations and basic properties of systems of the form \eqref{diffeqgeneralform} are now detailed below. The one parameter solution semigroup will be denoted by
$\Psi_t$, thus for a starting point $\v=(v_1,\ldots,v_d)\in \R^d$, the solution of \eqref{diffeqgeneralform} will be denoted by 
$\Psi_t(\v)$, or equivalently, $\v(t)$. \citet*{AlonsoStuartLongtime} and \citet*{Dataassimilation} have assumed that the nonlinearity is energy conserving, i.e. $\l<\mtx{B}(\v,\v),\v\r>=0$ for every $\v\in \R^d$. They also assume that the linear operator $\mtx{A}$ is positive definite, i.e.\@ there is a $\lambda_{\mtx{A}}>0$ such that $\l<\mtx{A}\v,\v\r>\ge \lambda_{\mtx{A}} \l<\v,\v\r>$ for every $\v\in \R^d$. As explained on page 50 of \citet*{Dataassimilation}, \eqref{diffeqgeneralform} together with these assumptions above implies that for every $\v\in \R^d$,
\begin{equation}\label{eqabsorbingset}
\l.\frac{1}{2}\frac{d}{dt} \|\v(t)\|^2 \r|_{t=0}\le \frac{1}{2\lambda_{\mtx{A}}}\|\vct{f}\|^2-\frac{\lambda_{\mtx{A}}}{2}\|\v\|^2.
\end{equation}
From \eqref{eqabsorbingset} one can show that $\mathcal{B}_R$ is an absorbing set for any 
\begin{equation}\label{Rcondeq}
R\ge \frac{\|\vct{f}\|}{\lambda_{\mtx{A}}}, 
\end{equation}
thus all paths enter into this set, and they cannot escape from it once they have reached it. This in turn implies the existence of a global attractor (see e.g.\@ \citet*{Temam}, or Chapter~2 of \citet*{Stuartdynamicalsystems}). Moreover, the trapping ball assumption \eqref{eqtrappingball} holds.

For $t\ge 0$, let $\v(t)$ and $\w(t)$ denote the solutions of \eqref{diffeqgeneralform} started from some points $\v,\w\in \R^d$. Based on \eqref{diffeqgeneralform}, we have that for any two points $\v, \w\in \BR$, any $t\ge 0$,
\[\frac{d}{dt}(\v(t)-\w(t))=-\mtx{A} (\v(t)-\w(t)) -(\mtx{B}(\v(t),\v(t)-\w(t))-\mtx{B}(\w(t)-\v(t),\w(t))),\]
and therefore by Gr\"{o}nwall's lemma, we have that for any $t\ge 0$,
\begin{equation}\label{eqpathdistancebound}\exp( -G t)\|\v-\w\|\le \|\v(t)-\w(t)\|\le \exp( G t)\|\v-\w\|,\end{equation}
for a constant $G\defby \|\mtx{A}\|+2\|\mtx{B}\| R$, where 
\[\|\mtx{A}\|\defby \sup_{\v\in \R^d: \|\v\|=1} \|\mtx{A}\v\| \quad \text{ and } \quad \|\mtx{B}\|\defby \sup_{\v,\w\in \R^d: \|\v\|=1, \|\w\|=1} \|\mtx{B}(\v,\w)\|.\]
For $t\ge 0$, let $\Psi_{t}(\BR)\defby  \{\Psi_{t}(\v): \v\in \BR\}$, then by   \eqref{eqpathdistancebound}, it follows that $\Psi_{t}:\BR\to \Psi_{t}(\BR)$ is a one-to-one mapping, which has an inverse that we denote as $\Psi_{-t}: \Psi_{t}(\BR)\to \BR$.

The main quantities of interest of this paper are the smoothing and filtering distributions corresponding to the conditional distribution of $\u(t_0)$ and $\u(t_k)$, respectively, given the observations $\vct{Y}_{0:k}\defby \l\{\vct{Y}_0,\ldots,\vct{Y}_k\r\}$. 
The densities of these distributions will be denoted by $\musm(\v|\vct{Y}_{0:k})$ and $\mufi(\v|\vct{Y}_{0:k})$. To make our notation more concise, we define the observed part of the dynamics as 
\begin{equation}\label{Phitdefeq}
\Phi_t(\v)\defby \mtx{H}\Psi_{t}(\v),
\end{equation}
for any $t\in \R$ and $\v\in \BR$. Using these notations, the densities of the smoothing and filtering distributions can be expressed as
\begin{align}
\label{eqmusm} &\musm(\v|\vct{Y}_{0:k})=\l[\prod_{i=0}^k \eta\l(\vct{Y}_i -\Phi_{t_i}(\v)\r)\r] \frac{q(\v)}{Z_{k}^{\mathrm{sm}}} \text{ for }\v\in \BR, 
\text{ and }0\text{ for }\v\notin \BR\\
\label{eqmufi}&\mufi(\v|\vct{Y}_{0:k})
=\l[\prod_{i=0}^k\eta\l(\vct{Y}_i-\Phi_{t_i-t_k}(\v) \r)\r] \left|\det(\mtx{J}\Psi_{-t_k}(\v))\right|  \frac{q(\Psi_{-t_k}(\v))}{Z_{k}^{\mathrm{fi}}}  \text{ for }\v\in  \Psi_{t_k}(\BR) ,\\
\nonumber&\text{and }0 \,\text{ for }\,\v\notin \Psi_{t_k}(\BR),
\end{align}
where $\det$ stands for determinant,  and $Z_{k}^{\mathrm{sm}}, Z_{k}^{\mathrm{fi}}$ are normalising constants independent of $\v$.  Since the determinant of the inverse of a matrix is the inverse of its determinant, we have the equivalent formulation
\begin{equation}\label{detalternativeq}
\det(\mtx{J}\Psi_{-t_k}(\v)) =\l(\det(\mtx{J}_{\Psi_{-t_k}(\v)}\Psi_{t_k})\r)^{-1}.
\end{equation}
We assume a prior $q$ on the initial condition that is absolutely continuous with respect to the Lebesgue measure, and zero outside the ball $\BR$ (where the value of $R$ is determined by the trapping ball assumption \eqref{eqtrappingball}).

For $k\ge 1$, we define the $k$th Jacobian of a function $g:\R^{d_1}\to \R^{d_2}$ at point $\v$ as a $k+1$ dimensional array, denoted by $\J^k g(\v)$ or equivalently $\J^k_{\v} g$ , with elements
\[(\J^k g(\v))_{i_1,\ldots,i_{k+1}}\defby \frac{\partial^k}{\partial v_{i_1} \ldots \partial v_{i_k} }
g_{i_{k+1}}(\v), \quad 1\le i_1,\ldots,i_k\le d_1, \, 1\le i_{k+1}\le d_2.\]
We define the norm of this $k$th Jacobian as
\[\|\J^k g(\v)\|\defby \sup_{\v^{(1)}\in \R^{d_1},\ldots,\v^{(k)}\in \R^{d_1},\v^{(k+1)}\in \R^{d_2}: 
\|\v^{(j)}\| \le 1,\, 1\le j\le k+1}
 (\J^k g(\v))[\v^{(1)},\ldots,\v^{(k+1)}],\]
where for a $k+1$ dimensional $d_1\times \ldots \times d_{k+1}$ sized array $\mtx{M}$, we denote 
\begin{align*}\mtx{M}[\v^{(1)},\ldots,\v^{(k+1)}]\defby\sum_{1\le i_1\le d_1,\ldots,1\le i_{k+1}\le d_{k+1}}M_{i_1,\ldots,i_{k+1}}\cdot v_{i_1}^{(1)}\cdot \ldots \cdot v_{i_{k+1}}^{(k+1)}.
\end{align*}

Using  \eqref{diffeqgeneralform} and \eqref{eqtrappingball}, we have that 
\begin{align}\label{eqvmax}&\sup_{\v\in \BR, t\ge 0}\l\|\frac{d \v(t)}{d t}\r\|\le v_{\max}\defby \|\mtx{A}\|R+\|\mtx{B}\|R^2+\|\vct{f}\|, \\
&\label{eqamax}\sup_{\v\in \BR, t\ge 0}\l\|\mtx{J}_{\v(t)} \l(\frac{d \v(t)}{d t}\r)\r\|\le a_{\max}\defby \|\mtx{A}\|+2\|\mtx{B}\|R.
\end{align}
By induction, we can show that for any $i\ge 2$, and any $\v\in \R^d$, we have
\begin{equation}\label{udereq}\D^i \v=-\mtx{A}\cdot\D^{i-1} \v-\sum_{j=0}^{i-1}  {i-1 \choose j} \mtx{B}\l(\D^j \v,\D^{i-1-j} \v\r).
\end{equation}
From this, the following bounds follow (see Section \ref{Secproofpreliminaryresults} of the Appendix for a proof). 
\begin{lem}\label{Dinormbndlemma}
For any $i\ge 0$, $k\ge 1$,  $\v\in \BR$, we have
\begin{align}\label{uderboundeq}
&\l\|\D^i \v\r\|\le C_0 \l(\Cder\r)^i \cdot i!, \\
\label{ugradboundeqgen}&\l\|\mtx{J}_{\v}^k \l(\D^i \v\r)\r\|\le \l(C_{\mtx{J}}^{(k)}\r)^i \cdot i!, \text{ where}\\
&C_0\defby R+\frac{\|\vct{f}\|}{\|\mtx{A}\|},\, \Cder\defby \|\mtx{A}\|+\|\mtx{B}\|R+\frac{\|\mtx{B}\|}{\|\mtx{A}\|} \|\vct{f}\|,\, C_{\mtx{J}}^{(k)}\defby 2^{k}(\Cder+\|\mtx{B}\|),\,\, k\ge 1. \label{C0derJdefeq}
\end{align}
\end{lem}

In some of our arguments we are going to use the multivariate Taylor expansion for vector valued functions.  Let $g:\R^{d_1}\to \R^{d_2}$ be $k+1$ times differentiable for some $k\in \N$. Then using the one dimensional Taylor expansion of the functions $g_i(\vct{a}+t\vct{h})$ in $t$ (where $g_i$ denotes the $i$th component of $g$), one can show that for any $\vct{a}, \vct{h}\in \R^{d_1}$, we have
\begin{equation}g(\vct{a}+\vct{h})=g(\vct{a})+\sum_{1\le j\le k}\frac{1}{j!}\cdot \l(\J^j g(\vct{a})[\vct{h}^j,\cdot]\r)+\vct{R}_{k+1}(\vct{a},\vct{h}),\end{equation}
where $\vct{h}^j\defby(\vct{h},\ldots,\vct{h})$ denotes the $j$ times repetition of $\vct{h}$, and the error term $\vct{R}_{k+1}(\vct{a},\vct{h})$ is of the form
\begin{equation}
\vct{R}_{k+1}(\vct{a},\vct{h})\defby \frac{k+1}{(k+1)!}\cdot \int_{t=0}^{1} (1-t)^k\J^{k+1} g(\vct{a}+t\vct{h})[\vct{h}^{k+1},\cdot] dt,
\end{equation}
whose norm can be bounded using the fact that $\int_{t=0}^{1} (1-t)^k dt=\frac{1}{k+1}$ as
\begin{equation}\label{eqRkp1normbnd}
\|\vct{R}_{k+1}(\vct{a},\vct{h})\|\le \frac{\|\vct{h}\|^{k+1}}{(k+1)!}\cdot \sup_{0\le t\le 1}\|\J^{k+1} g(\vct{a}+t\vct{h})\|.
\end{equation}
In order to be able to use such multivariate Taylor expansions in our setting, the existence and finiteness of $\J^k\Psi_{t}(\vct{v})$ can be shown rigorously in the following way. Firstly, for $0<t<\l(C_{\J}^{(k)}\r)^{-1}$, one has
\[\J^k \Psi_{t}(\vct{v})=\J^k_{\v}\l(\sum_{i=0}^{\infty}\D^i \v \cdot \frac{t^i}{i!}\r),\]
and using the inequality \eqref{ugradboundeqgen}, we can show that
\[\J^k \Psi_{t}(\vct{v})=\sum_{i=0}^{\infty}\J^k_{\v} \l(\D^i \v\r) \cdot \frac{t^i}{i!}\] is convergent and finite. For $t\ge \l(C_{\J}^{(k)}\r)^{-1}$, we can express $\Psi_{t}(\vct{v})$ as a composition $\Psi_{t_1}(\ldots (\Psi_{t_m}(\vct{v})))$ for $t_1+\cdots+t_m=t$, and establish the existence of the partial derivatives by the chain rule.

After establishing the existence of the partial derivatives $\J^k \Psi_{t}(\vct{v})$, we are going to bound their norm in the following lemma (proven in Section \ref{Secproofpreliminaryresults} of the Appendix).
\begin{lem}\label{partderbndlem}
For any $k\ge 1$, let 
\begin{equation}\label{DJkeq}
D_{\J}^{(k)}\defby 2^{k} \l(\|\mtx{A}\|+\|\mtx{B}\|+2\|\mtx{B}\|R\r).
\end{equation}
Then for any $k\ge 1$, $T\ge 0$, we have
\begin{align}\label{MkTdef}
M_k(T)&\defby \sup_{\v\in \BR}\sup_{0\le t\le T}\|\J^k\Psi_{t}(\vct{v})\|\le \exp\l(D_{\J}^{(k)} T\r),\text{ and}\\
\label{tMkTdef}
\hM_k(T)&\defby \sup_{\v\in \BR}\sup_{0\le t\le T}\|\J^k\Phi_{t}(\vct{v})\|\le \|\mtx{H}\|M_k(T)\le \|\mtx{H}\|\exp\l(D_{\J}^{(k)} T\r).
\end{align}
\end{lem}

\section{Main results}\label{secmainresults}
In this section, we present our main results. We start by introducing our assumptions. In Section \ref{secGaussian} we show that the smoother and the filter can be well approximated by Gaussian distributions when $\sigma_Z^2 h$ is sufficiently small. This is followed by Section \ref{secMAP} where based on the Gaussian approximation result we show that the Maximum-a-posteriori (MAP) estimators are asymptotically optimal in mean square error in the $\sigma_Z^2 h\to 0$ limit. We also show that Newton's method can be used for calculating the MAP estimators if the initial point $\vct{x}_0$ can be chosen sufficiently close to the true starting position $\u$. Finally, in Section \ref{secInitial} we propose estimators to use as initial point $\vct{x}_0$ that satisfy this criteria when $\sigma_Z^2 h$ and $h$ is sufficiently small.

We start with an assumption that will be used in these results.

\begin{ass}\label{assgauss1}
Let $T>0$ be fixed, and suppose that $T=kh$, where $k\in \N$. Suppose that $\|\u\|<R$, and that there exist constants $h_{\max}(\u,T)>0$ and $c(\u,T)>0$  such that for every $\v\in \BR$, for every $h\le h_{\max}(\u,T)$ (or equivalently, every $k\ge T/h_{\max}(\u,T)$), we have
\begin{equation}\label{eqassgauss1}
\sum_{i=0}^{k} \|\Phi_{t_i}(\v)-\Phi_{t_i}(\u)\|^2\ge \frac{c(\u,T)}{h} \|\u-\v\|^2.
\end{equation}
\end{ass}
As we shall see in Proposition \ref{Propcheckassumptionsgauss}, this assumption follows from the following assumption on the derivatives (introduced in \citet*{ConcentrationProperties}).
\begin{ass}\label{assder}
Suppose that $\|\u\|<R$, and there is an index $j\in \N$ such that the system of equations in $\v$ defined as
\begin{equation}\label{eqphitideruv} 
\H\D^i \u= \H \D^i \v\text{ for every }0\le i\le j\end{equation}
has a unique solution $\v:=\u$ in $\BR$, and
\begin{equation}\label{eqspanderu}
\mr{span}\l\{\grad \left(\H \D^i \u\right)_k: 0\le i\le j, 1\le k\le d_o\r\}=\R^{d},
\end{equation}
where $\left(\H \D^i \u\right)_k$ refers to coordinate $k$ of the vector $\H \D^i \u\in \R^{d_o}$, and $\grad$ denotes the gradient of the function in $\u$.
\end{ass}

\begin{prop}\label{Propcheckassumptionsgauss}
Assumption \ref{assder} implies Assumption \ref{assgauss1}.
\end{prop}
The proof is given in Section \ref{Secproofpreliminaryresults} of the Appendix. Assumption \ref{assder} was verified for the Lorenz 63' and 96' models in \citet*{ConcentrationProperties} (for certain choices of the observation matrix $\mtx{H}$), thus Assumption \ref{assgauss1} is also valid for these models.

We denote $\mr{int}(\BR)\defby\{\v\in \R^d: \|\v\|<R\}$ the interior of $\BR$. In most of our results, we will make the following assumption about the prior $q$.
\begin{ass}\label{assprior}
The prior distribution $q$ is assumed to be absolutely continuous with respect to the Lebesgue measure on $\R^d$, supported on $\BR$. We assume that $v\to q(\v)$ is strictly positive and continuous on $\BR$, and that it is 3 times continuously differentiable at every interior point of $\BR$.
Let
\begin{align*}
C_{q}^{(i)}:=\sup_{\v\in \mr{int}(\BR)} \|\J^i \log q(\v)\| \text{ for }i=1,2,3.
\end{align*}
We assume that these are finite.
\end{ass}

After some simple algebra, the smoothing distribution for an initial point $\u$ and the filtering distribution for the current position $\u(T)$ can be expressed as 
\begin{align}\label{eqmusmgaussian}
&\musm(\v|\vct{Y}_{0:k})\\
&\nonumber=\exp\l[-\frac{1}{2\sigma_Z^2}\sum_{i=0}^{k} \l(\|\Phi_{t_i}(\v)-\Phi_{t_i}(\u)\|^2 +2\l<\Phi_{t_i}(\v)-\Phi_{t_i}(\u), \vZ_i\r>\r)\r]\cdot q(\v)/C_{k}^{\mr{sm}},
\\
\label{eqmufigaussian}&\mufi(\v|\vct{Y}_{0:k})=1_{[\v\in \Psi_T(\BR)]}\cdot \musm(\Psi_{-T}(\v)|\vct{Y}_{0:k})\cdot \left|\det\left(\J\Psi_{-T}(\v)\right)\right|
\\
&\nonumber=1_{[\v\in \Psi_T(\BR)]}\cdot
\exp\l[-\frac{1}{2\sigma_Z^2}\sum_{i=0}^{k} \l(\|\Phi_{t_i}(\Psi_{-T}(\v))-\Phi_{t_i}(\u)\|^2 +2\l<\Phi_{t_i}(\Psi_{-T}(\v))-\Phi_{t_i}(\u), \vZ_i\r>\r)\r]\\
&\nonumber\cdot q(\Psi_{-T}(\v)) \cdot \left|\det\left(\J\Psi_{-T}(\v)\right)\right|/C_{k}^{\mr{sm}},
\end{align}
where $C_{k}^{\mr{sm}}$ is a normalising constant independent of $\v$ (but depending on $(\vZ_j)_{j\ge 0}$).

In the following sections, we will present our main results for the smoother and the filter. First, in Section \ref{secGaussian} we are going to state Gaussian approximation results, then in Section \ref{secMAP} we state various results about the MAP estimators, and in Section \ref{secInitial} we propose an initial estimator for $\u$ based on the observations $\Y_0,\ldots, \Y_k$, to be used as a starting point for Newton's method.

\subsection{Gaussian approximation}\label{secGaussian}
We define the matrix $\mtx{A}_k\in \R^{d\times d}$ and vector $\vct{B}_k\in \R^d$ as
\begin{align}
\label{eqAkdef}\mtx{A}_k&:=\sum_{i=0}^k \l(\J \Phi_{t_i}(\u)' \J \Phi_{t_i}(\u)+ \J^2 \Phi_{t_i}(\u)[\cdot,\cdot,\vZ_i]\r), \\
\label{eqBkdef}\vct{B}_k&:=\sum_{i=0}^k \J \Phi_{t_i}(\u)' \cdot \vZ_i,
\end{align}
where $\J \Phi_{t_i}$ and $\J^2 \Phi_{t_i}$ denotes the first and second Jacobian of $\Phi_{t_i}$, respectively, and $\J^2 \Phi_{t_i}(\u)[\cdot,\cdot,\vZ_i]$ denotes the $d\times d$ matrix with elements 
\[\l[\J^2 \Phi_{t_i}(\u)[\cdot,\cdot,\vZ_i]\r]_{i_1,i_2}=\sum_{j=1}^{d_o}(\J^2 \Phi_{t_i}(\u))_{i_1,i_2,j}Z_i^{j} \text{ for }1\le i_1,i_2\le d.\]
If $\mtx{A}_k$ is positive definite, then we define the center of the Gaussian approximation of the smoother as
\begin{equation}\label{uGdefeq}
\uG:=\u-\mtx{A}_k^{-1}\vct{B}_k,
\end{equation}
and define the Gaussian approximation of the smoother as
\begin{align}\label{musmGdefeq}
&\musmG(\v|\vct{Y}_{0:k}):=\frac{\det(\mtx{A}_k)^{1/2}}{(2\pi)^{d/2}\cdot \sigma_Z^d}\cdot \exp\l[-\frac{(\v-\uG)'\mtx{A}_k (\v-\uG) }{2\sigma_Z^2}\r].
\end{align}
If $\mtx{A}_k$ is not positive definite, then we define the Gaussian approximation of the smoother $\musmG(\cdot|\vct{Y}_{0:k})$ to be the $d$-dimensional standard normal distribution (an arbitrary choice), and $\uG:=\vct{0}$. If $\mtx{A}_k$ is positive definite, and $\uG\in \BR$, then we define the Gaussian approximation of the filter as
\begin{align}\label{mufiGdefeq}
&\mufiG(\v|\vct{Y}_{0:k}):=\frac{\det(\mtx{A}_k)^{1/2}}{|\det(\J \Psi_T(\uG))|}\cdot \f{1}{(2\pi)^{d/2}\cdot \sigma_Z^d}\\
&\nonumber \cdot \exp\l[-\frac{\l(\v-\Psi_T(\uG)\r)' \l(\l(\J \Psi_T(\uG)\r)^{-1}\r)'\mtx{A}_k \l(\J \Psi_T(\uG)\r)^{-1} \l(\v-\Psi_T(\uG)\r) }{2\sigma_Z^2}\r].
\end{align}
Alternatively, if $\mtx{A}_k$ is not positive definite, or $\uG\notin \BR$, then we define the Gaussian approximation of the smoother $\mufiG(\cdot|\vct{Y}_{0:k})$ to be the $d$-dimensional standard normal distribution.

In order to compare the closeness between the target distributions and their Gaussian approximation, we are going to use two types of distance between distributions. The  \emph{total variation distance} of two distributions $\mu_1,\mu_2$ on $\R^d$ that are absolutely continuous with respect to the Lebesgue measure is defined as
\begin{equation}\label{dtvdefeq}
\dtv(\mu_1,\mu_2):=\frac{1}{2}\int_{x\in \R^d}|\mu_1(x)-\mu_2(x)| dx,
\end{equation}
where  $\mu_1(x)$, and $\mu_2(x)$ denote the densities of the distributions. 

The \emph{Wasserstein distance} (also called 1st Wasserstein distance) of two distributions $\mu_1,\mu_2$ on $\R^d$ (with respect to the Euclidean distance) is defined as
\begin{equation}\label{dWdefeq}
d_{\mr{W}}(\mu_1,\mu_2):=\inf_{\gamma\in \Gamma(\mu_1,\mu_2)} \int_{\vct{x},\vct{y}\in \R^d} \|\vct{x}-\vct{y}\| d \gamma(\vct{x},\vct{y}),
\end{equation}
where $\Gamma(\mu_1,\mu_2)$ is the set of all measures on $\R^d\times \R^d$ with marginals $\mu_1$ and $\mu_2$.

The following two theorems bound the total variation and Wasserstein distances between the smoother, the filter, and their Gaussian approximations. In some of our bounds, the quantity $T+h_{\max}(\u,T)$ appears. For brevity, we denote this as
\begin{equation}
\olT:=T+h_{\max}(\u,T).
\end{equation}
We are also going to use the constant $C_{\|\mtx{A}\|}$ defined as
\begin{equation}\label{CAkdef}
C_{\|\mtx{A}\|}:=\hM_1(T)^2\cdot\olT+\frac{c(\u,T)}{2}.
\end{equation}

\begin{theorem}[Gaussian approximation of the smoother]\label{Gaussianapproxthmsm}
Suppose that Assumptions \ref{assgauss1} and \ref{assprior} hold for the initial point $\u$ and the prior $q$. Then there are constants $C^{(1)}_{\mr{TV}}(\u,T)$, $C^{(2)}_{\mr{TV}}(\u,T)$, $C^{(1)}_{\mr{W}}(\u,T)$, and $C^{(2)}_{\mr{W}}(\u,T)$ independent of $\sigma_Z$, $h$ and $\epsilon$ such that for any $0<\epsilon\le 1$, $\sigma_Z>0$, and $0\le h\le h_{\max}(\u,T)$ satisfying that $\sigma_Z\sqrt{h}\le \frac{1}{2}C_{\mr{TV}}(\u,T,\epsilon)^{-1}$,  we have
\begin{align}\nonumber&\PP\Bigg[\f{c(\u,T)}{2h}\mtx{I}_d\prec \mtx{A}_k \prec \f{C_{\|\mtx{A}\|}}{h}\mtx{I}_d\text{ and  }\dtv\l(\musm(\cdot|\vct{Y}_{0:k}),\musmG(\cdot|\vct{Y}_{0:k})\r)\le C_{\mr{TV}}(\u,T,\epsilon) \sigma_Z\sqrt{h}\\
\label{dtvboundeq} & \quad\text{ and } d_{\mr{W}}\l(\musm(\cdot|\vct{Y}_{0:k}),\musmG(\cdot|\vct{Y}_{0:k})\r)\le C_{\mr{W}}(\u,T,\epsilon) \sigma_Z^2 h\Bigg|\u\Bigg]
\ge 1-\epsilon,\text{ where}
\end{align}
\vspace{-5.5mm}
\begin{align}
\label{eqCTVuTepsdef}
C_{\mr{TV}}(\u,T,\epsilon)&:=C^{(1)}_{\mr{TV}}(\u,T)+C^{(2)}_{\mr{TV}}(\u,T)\l(\log\l(\frac{1}{\epsilon}\r)\r)^2, \text{ and }\\
\label{eqCWuTepsdef}
C_{\mr{W}}(\u,T,\epsilon)&:=C^{(1)}_{\mr{W}}(\u,T)+C^{(2)}_{\mr{W}}(\u,T)\l(\log\l(\frac{1}{\epsilon}\r)\r)^2.
\end{align}
\end{theorem}

\begin{theorem}[Gaussian approximation of the filter]\label{Gaussianapproxthmfi}
Suppose that Assumptions \ref{assgauss1} and \ref{assprior} hold for the initial point $\u$ and the prior $q$. Then there are constants $D^{(1)}_{\mr{TV}}(\u,T)$, $D^{(2)}_{\mr{TV}}(\u,T)$, $D^{(1)}_{\mr{W}}(\u,T)$, and $D^{(2)}_{\mr{W}}(\u,T)$  independent of $\sigma_Z$, $h$ and $\epsilon$ such that for any $0<\epsilon\le 1$, $\sigma_Z>0$, and $0\le h\le h_{\max}(\u,T)$ satisfying that $\sigma_Z\sqrt{h}\le \frac{1}{2}D_{\mr{TV}}(\u,T,\epsilon)^{-1}$,  we have
\begin{align}\label{dtvboundeqfi}\PP\Big[&\f{c(\u,T)}{2h}\mtx{I}_d\prec \mtx{A}_k \prec \f{C_{\|\mtx{A}\|}}{h}\mtx{I}_d\text{ and }\uG\in \BR \text{ and  }\\
\nonumber &\dtv\l(\mufi(\cdot|\vct{Y}_{0:k}),\mufiG(\cdot|\vct{Y}_{0:k})\r)\le D_{\mr{TV}}(\u,T,\epsilon) \sigma_Z\sqrt{h} \text{ and }\\
 \nonumber& d_{\mr{W}}\l(\mufi(\cdot|\vct{Y}_{0:k}),\mufiG(\cdot|\vct{Y}_{0:k})\r)\le D_{\mr{W}}(\u,T,\epsilon) \sigma_Z^2 h\Big|\u\Big]
\ge 1-\epsilon,\text{ where}
\end{align}
\vspace{-6.5mm}
\begin{align}
\label{eqCTVuTepsdeffi}
D_{\mr{TV}}(\u,T,\epsilon)&:=D^{(1)}_{\mr{TV}}(\u,T)+D^{(2)}_{\mr{TV}}(\u,T)\l(\log\l(\frac{1}{\epsilon}\r)\r)^2, \text{ and }\\
\label{eqCWuTepsdeffi}
D_{\mr{W}}(\u,T,\epsilon)&:=D^{(1)}_{\mr{W}}(\u,T)+D^{(2)}_{\mr{W}}(\u,T)\l(\log\l(\frac{1}{\epsilon}\r)\r)^2.
\end{align}
\end{theorem}
Note that the Gaussian approximations $\musmG$ and $\mufiG$ as defined above are not directly computable based on the observations $\Y_{0:k}$, since they involve the true initial position $\u$ in both their mean and covariance matrix. However, we believe that with some additional straightforward calculations one could show that results similar to Theorems \ref{Gaussianapproxthmsm} and \ref{Gaussianapproxthmfi} also hold for the Laplace approximations of the smoothing and filtering distributions (i.e. when the mean and covariance of the normal approximation is replaced by the MAP and the inverse Hessian of the log-likelihood at the MAP, respectively), which are directly computable based on the observations $\Y_{0:k}$.

\subsection{MAP estimators}\label{secMAP}
Let $\umean$ be the mean of the smoothing distribution, and $\ufimean$ be the mean of the filtering distribution, and $\uMAPsm$ be the maximum-a-posteriori of the smoothing distribution, i.e.
\begin{equation}\label{uMAPsmdefeq}\uMAPsm:=\argmax_{\v\in \BR}\musm(\v|\vct{Y}_{0:k}).\end{equation}
In case there are multiple maxima, we choose any of them. For the filter, we will use the push-forward MAP estimator
\begin{equation}\label{uMAPfidefeq}\uMAPfi:=\Psi_T(\uMAPsm).\end{equation}
Based on the Gaussian approximation results, we prove the following two theorems about these  estimators.

\begin{theorem}[Comparison of mean square error of MAP and posterior mean for smoother]\label{GaussianMAPthm}
Suppose that Assumptions \ref{assgauss1} and \ref{assprior} hold for the initial point $\u$ and the prior $q$. Then there is a constant $S^{\mr{sm}}_{\max}(\u,T)>0$ independent of $\sigma_Z$ and $h$ such that for $0<h<h_{\max}(\u,T)$, $\sigma_Z\sqrt{h}\le S^{\mr{sm}}_{\max}(\u,T)$, we have that 
\begin{align}
&\ul{C}^{\mr{sm}}(\u,T)\le \frac{\E\l[\|\umean-\u\|^2|\u\r]}{\sigma_Z^2 h}\le \ol{C}^{\mr{sm}}(\u,T), \text{ and }\\
&\l|\E\l[\|\uMAPsm-\u\|^2|\u\r]-\E\l[\|\umean-\u\|^2|\u\r]\r|\le C^{\mr{sm}}_{\mr{MAP}}(\u,T) (\sigma_Z^2 h)^{\frac{3}{2}},
\end{align}
where the expectations are taken with respect to the random observations, and $\ul{C}^{\mr{sm}}(\u,T)$, $\ol{C}^{\mr{sm}}(\u,T)$, and $C^{\mr{sm}}_{\mr{MAP}}(\u,T)$ are finite positive constants independent of $\sigma_Z$ and $h$.
\end{theorem}

\begin{theorem}[Comparison of mean square error of MAP and posterior mean for filter]\label{GaussianMAPthmfi}
Suppose that Assumptions \ref{assgauss1} and \ref{assprior} hold for the initial point $\u$ and the prior $q$. Then there is a constant $S^{\mr{fi}}_{\max}(\u,T)>0$ independent of $\sigma_Z$ and $h$ such that for $0<h<h_{\max}(\u,T)$, $\sigma_Z\sqrt{h}\le S^{\mr{fi}}_{\max}(\u,T)$, we have that 
\begin{align}
&\ul{C}^{\mr{fi}}(\u,T)\le \frac{\E\l[\|\ufimean-\u(T)\|^2|\u\r]}{\sigma_Z^2 h}\le \ol{C}^{\mr{fi}}(\u,T), \text{ and }\\
&\l|\E\l[\|\uMAPfi-\u(T)\|^2|\u\r]-\E\l[\|\ufimean-\u(T)\|^2|\u\r]\r|\le C^{\mr{fi}}_{\mr{MAP}}(\u,T) (\sigma_Z^2 h)^{\frac{3}{2}},
\end{align}
where the expectations are taken with respect to the random observations, and $\ul{C}^{\mr{fi}}(\u,T)$, $\ol{C}^{\mr{fi}}(\u,T)$, and $C^{\mr{fi}}_{\mr{MAP}}(\u,T)$ are finite positive constants independent of $\sigma_Z$ and $h$.
\end{theorem}
\begin{rem}
Theorems \ref{GaussianMAPthm} and \ref{GaussianMAPthmfi} in particular imply that when $T$ is fixed, and $\sigma_Z\sqrt{h}$ tends to $0$, the ratio between the mean square errors of the posterior mean and MAP estimators conditioned on the initial position $\u$ tends to 1. Since the mean of the posterior distributions, $\umean$ (or $\ufimean$ for the filter), is the estimator $U(\vct{Y}_{0:k})$ that minimises  $\E(\|U(\vct{Y}_{0:k})-\u\|^2)$ (or $\E(\|U(\vct{Y}_{0:k})-\u(T)\|^2)$ for the filter), our results imply that the mean square error of the MAP estimators is close to optimal when $\sigma_Z\sqrt{h}$ is sufficiently small.
\end{rem}

Next we propose a method to compute the MAP estimators. Let $\gsm: \BR\to \R$ be 
\begin{equation}\label{gsmdef}
\gsm(\v):=-\log(q(\v))+\frac{1}{2\sigma_Z^2}\sum_{i=0}^k \|\vct{Y}_i-\Phi_{t_i}(\v)\|^2.
\end{equation}
Then $-\gsm(\v)$ is the log-likelihood of the smoother, except that it does not contain the normalising constant term. 

The following theorem shows that Newton's method can be used to compute 
$\uMAPsm$ to arbitrary precision if it is initiated from a starting point $\vct{x}_0$ that is sufficiently close to the initial position $\u$. The proof is based on the concavity properties of the log-likelihood near $\u$. Based on this, an approximation for the push-forward MAP estimator $\uMAPfi$ can be then computed by moving forward the approximation of $\uMAPsm$ by time $T$ according to the dynamics $\Psi_T$ (this will not increase the error by more than a factor of $\exp(GT)$ according to \eqref{eqpathdistancebound}).

\begin{theorem}[Convergence of Newton's method to the MAP]\label{NewtonMAPthm}Suppose that Assumptions \ref{assgauss1} and \ref{assprior} hold for the initial point $\u$ and the prior $q$.  Then for every $0<\epsilon\le 1$, there exist finite constants $S_{\max}^{\mr{sm}}(\u,T,\epsilon)$,  $N^{\mr{sm}}(\u,T)$ and $D_{\max}^{\mr{sm}}(\u,T)\in (0, N^{\mr{sm}}(\u,T)]$ (defined in \eqref{Smaxsmdefeq} and \eqref{DmaxCNdefeq}) such that the following holds. If
$\sigma_Z\sqrt{h}\le S_{\max}^{\mr{sm}}(\u,T,\epsilon)$, and the initial point $\vct{x}_0\in \BR$ satisfies that $\|\vct{x}_0-\u\|< D_{\max}^{\mr{sm}}(\u,T)$, then the iterates of Newton's method defined recursively as
\begin{equation}\label{smNewtoniteq}\vct{x}_{i+1}:=\vct{x}_{i}-(\grad^2 \gsm(\vct{x}_{i}))^{-1}\cdot \grad \gsm(\vct{x}_{i}) \text{ for }i\in \N
\end{equation}
satisfy that
\begin{align}
\nonumber&\Pcu{\vct{x}_i\text{ are well defined and }\|\vct{x}_{i}-\uMAPsm\|\le N^{\mr{sm}}(\u,T) \l(\frac{\|\vct{x}_0-\u\|}{N^{\mr{sm}}(\u,T)}\r)^{\displaystyle{2^i}} \text{ for every }i\in \N}\\
&\ge 1-\epsilon.\label{MAPNewtonbndeq}
\end{align}
\end{theorem}
\begin{rem}
The bound \eqref{MAPNewtonbndeq} means that number of digits of precision essentially doubles in each iteration. In other words, only a few iterations are needed to approximate the MAP estimator with high precision if $\vct{x}_0$ is sufficiently close to $\u$.
\end{rem}

\subsection{Initial estimator}\label{secInitial}
First we are going to estimate the derivatives $\mtx{H} \D^l \u$ for $l\in \N$ based on observations $\vct{Y}_{0:k}$. For technical reasons, the estimators will depend on $\vct{Y}_{0:\hk}$ for some $0\le \hk\le k$ (which will be chosen depending on $l$). For any $j\in \N$, we  define $\v^{(j|\hk)}\in \R^{\hk+1}$ as
\begin{equation}\label{vjkdefeq}
\v^{(j|\hk)}:=\l\{\l(\frac{i}{\hk}\r)^j\r\}_{0\le i\le \hk}, \text{ with the convention that }0^0:=1.
\end{equation}
For $\jmax\in \N$, we define $\mM \in \R^{(\jmax+1)\times (\hk+1)}$ as a matrix with rows $\v^{(0|\hk)}, \ldots,\v^{(\jmax|\hk)}$.  

We denote by $\mtx{I}_{\jmax+1}$ the identity matrix of dimension $\jmax+1$, and by $\vct{e}^{(l|\jmax)}$ a column vector in $\R^{\jmax+1}$ whose every component is zero except the $l+1$th one which is 1.

For any $l\in \N$, $\jmax\ge l$, $\hk \ge \jmax$, we define the vector
\begin{align}\label{clhkdefeq}
\c^{(l|\jmax|\hk)}&:= \frac{l!}{(\hk h)^l}(\mM)' \l(\mM (\mM)'\r)^{-1}\cdot \vct{e}^{(l|\jmax)}, \text{ then}\\
\label{hatphiljkdefeq}
\hat{\Phi}^{(l|\jmax)}(\Y_{0:\hk})&:=\sum_{i=0}^{\hk}c^{(l|\jmax|\hk)}_i \cdot \Y_i
\end{align}
is an estimator of $\mtx{H} \D^l \u$. The fact that the matrix $\mM (\mM)'$ is invertible follows from the fact that $\v^{(0|\hk)},\ldots, \v^{(j_{\max}|\hk)}$ are linearly independent (since the matrix with rows $\v^{(0|\hk)}, \ldots, \v^{(\hk|\hk)}$ is a so-called Vandermonde matrix whose determinant is non-zero). From \eqref{clhkdefeq}, it follows that the norm of $\c^{(l|\jmax|\hk)}$ can be expressed as
\begin{equation}\label{clhknormeq}
\l\|\c^{(l|\jmax|\hk)}\r\|=\frac{l!}{(\hk h)^l}\sqrt{\l[ \l(\mM (\mM)'\r)^{-1}\r]_{l+1,l+1}}.
\end{equation}
To lighten the notation, for $\jmax\ge l$ and $\hk\ge \jmax$, we will denote 
\begin{equation}
\CM:=\sqrt{\hk\cdot \l[\l(\mM (\mM)'\r)^{-1}\r]_{l+1,l+1}}.
\end{equation}
The next proposition gives an error bound for this estimator, which we will use for choosing the values $\hat{k}$ and $\jmax$ given $l$. 
\begin{prop}\label{propestbiasconc1}
Suppose that $\jmax\ge l$ and $\hk\ge 2\jmax+3$. Then for any $0< \epsilon\le 1$,
\begin{align*}
&\PP\Bigg[\l\|\hat{\Phi}^{(l|\jmax)}(\Y_{0:\hk})-\mtx{H} \D^l \u\r\|\ge 
\CM\cdot l!\cdot g(l,\jmax,\hk) \cdot \sqrt{1+\frac{\log\l(1/\epsilon\r)}{\log(d_o+1)}}\Bigg|\u\Bigg]\le \epsilon,
\end{align*}
where
\begin{equation}\label{gfuncdefeq}
g(l,\jmax,\hk):=
\frac{C_0\|\mtx{H}\|\Cder^{\jmax+1}}{\sqrt{\jmax+3/2}}  \cdot  (\hk h)^{\jmax+1-l}
+(\hk h)^{-l-1/2}\sigma_Z\sqrt{h} \sqrt{2d_o \log(d_o+1)}.
\end{equation}
\end{prop}

The following lemma shows that as $\hk\to \infty$, the constant $\CM$ tends to a limit. \begin{lem}\label{Klimitlem}
Let $\mtx{K}^{(\jmax)}\in \R^{\jmax+1\times \jmax+1}$ be a matrix with elements 
$\mtx{K}^{(\jmax)}_{i,j}:=\frac{1}{i+j-1}$ for $1\le i,j\le \jmax+1$.
Then for any $l\in \N$, $\jmax\ge l$, the matrix $\mtx{K}^{(\jmax)}$ is invertible, and
\begin{equation}\label{CMlimiteq}
\lim_{\hk\to \infty}\CM=\l[\l(\mtx{K}^{(\jmax)}\r)^{-1}\r]_{l+1,l+1}.
\end{equation}
\end{lem}
The proofs of the above two results are included in Section \ref{secInitialproof}. Based on these results, we choose $\hk\in\{2\jmax+3,\ldots, k\}$ such that the function $g(l,\jmax,\hk)$ is minimised. We denote this choice of $\hk$ by $\hk_{\mr{opt}}(l,\jmax)$ (if $g(l,\jmax,\hk)$ takes the same value for several $\hk$, then we choose the smallest of them). By taking the derivative of $g(l,\jmax,\hk)$ in $\hk$, it is easy to see that it has a single minimum among positive real numbers achieved at 
\begin{equation}\label{kmindefeq}
\hk_{\mr{min}}(l,\jmax):=\f{1}{h}\cdot \l(\f{\sigma_Z\sqrt{h} \sqrt{d_o \log(d_o+1) (\jmax+3/2)} (l+1/2)}{(\jmax+1-l)C_0\|\mtx{H}\| \Cder^{\jmax+1}}\r)^{1/(\jmax+3/2)}.
\end{equation}
Based on this, we have 
\begin{align}\label{koptdefeq}
\hk_{\mr{opt}}(l,\jmax)&:=1_{\hk_{\mr{min}}(l,\jmax)\le 2\jmax+3}\cdot (2\jmax+3)+1_{\hk_{\mr{min}}(l,\jmax)\ge k}\cdot k
\\\nonumber
&+1_{2\jmax+3<\hk_{\mr{min}}(l,\jmax)<k}\cdot \argmin_{\hk\in \{\lfloor \hk_{\mr{min}}(l,\jmax)\rfloor,\lceil \hk_{\mr{min}}(l,\jmax)\rceil \}}g(l,\jmax,\hk).
\end{align}
Finally, based on the definition of $\hk_{\mr{opt}}(l,\jmax)$, we choose $\jmax^{\mr{opt}}(l)$ as 
\begin{align}\label{jmaxoptldef}
\jmax^{\mr{opt}}(l):=\argmin_{l\le \jmax\le J_{\max}^{(l)}} \l(C_{\mtx{M}}^{(l|j_{\max}|\hk_{\mr{opt}}(l,\jmax))}\cdot g(l,\jmax,\hk_{\mr{opt}}(l,\jmax))\r),
\end{align}
where $J_{\max}^{(l)}\in \{l,l+1,\ldots, \lfloor (k-3)/2 \rfloor\}$ is a parameter to be tuned. 
We choose the smallest possible $\jmax$ where the minimum is taken. Based on these notations, we define our estimator for $\mtx{H} \D^l \u$ as
\begin{equation}
\label{hatphildefeq}
\hat{\Phi}^{(l)}:=\hat{\Phi}^{(l|\jmax^{\mr{opt}}(l))}(\Y_{0:\hk_{\mr{opt}}(l,\jmax^{\mr{opt}})}).
\end{equation}
The following theorem bounds the error of this estimator.
\begin{theorem}\label{thminitialest}
Suppose that $\u\in \BR$, and $T=kh$. Then for any $l\in \N$, there exist some positive constants 
$h_{\max}^{(l)}$, $s_{\max}^{(l)}(T)$ and  $S_{\max}^{(l)}(T)$ such that 
for any choice of the parameter $J_{\max}^{(l)}\in  \{l,l+1,\ldots, \lfloor (k-3)/2 \rfloor\}$, any $\epsilon>0$, $0<s\le s_{\max}^{(l)}(T)$, $0<h\le \f{h_{\max}^{(l)} \cdot s}{\sqrt{1+\frac{\log\l(1/\epsilon\r)}{\log(d_o+1)}}}$, $0\le \sigma_Z \sqrt{h}\le S_{\max}^{(l)}(T)\cdot \l(\f{s}{\sqrt{1+\frac{\log\l(1/\epsilon\r)}{\log(d_o+1)}}}\r)^{l+3/2}$,
\[\Pcu{\l\|\mtx{H} \D^l \u-\hat{\Phi}^{(l)}\r\|\ge s}\le \epsilon.\]
\end{theorem}

The following theorem proposes a way of estimating $\u$ from estimates for the derivatives $\l(\mtx{H} \D^i \u\r)_{0\le i\le j}$. This will be used as our initial estimator for Newton's method.

\begin{theorem}\label{initialestthm}
Suppose that for some $j\in \N$ there is a function $F: (\R^{d_o})^{j+1}\to \R^d$ independent of $\u$  such that 
\begin{align}\label{Fcond1eq}
&F\l(\mtx{H} \u, \ldots, \mtx{H}\D^j \u\r)=\u,\text{ and }\\
\label{Fcond2eq}&\|F(\x^{(0)},\ldots \x^{(j)})-\u\|\le C_F(\u)\cdot \l(\sum_{i=0}^{j}
\l\|\mtx{H}\D^i \u-\x^{(i)}\r\|^2\r)^{1/2}
\end{align}
for $\sum_{i=0}^{j}\l\|\mtx{H}\D^i \u-\x^{(i)}\r\|^2\le D_F(\u)$, for some positive constants $C_F(\u)$, $D_F(\u)$. Then the estimator $F(\hat{\Phi}^{(0)},\ldots, \hat{\Phi}^{(j)})$ satisfies that if $\sum_{i=0}^{j}\l\|\mtx{H}\D^i \u-\hat{\Phi}^{(i)}\r\|^2\le D_F(\u)$, then 
\begin{equation}\label{Ferrbndeq}
\l\|F\l(\hat{\Phi}^{(0)},\ldots, \hat{\Phi}^{(j)}\r)-\u\r\|\le  C_F(\u)\cdot \l(\sum_{i=0}^{j}\l\|\mtx{H}\D^i \u-\hat{\Phi}^{(i)}\r\|^2\r)^{1/2}.
\end{equation}
In particular, under Assumption \ref{assder} and for $j$ as determined therein,   function $F$ defined as
\begin{align}\label{Fdefminsquare}
F(\x^{(0)},\ldots \x^{(j)}):=\argmin_{\v\in \BR} \sum_{i=0}^j\l\|\mtx{H} \D^i \v-\x^{(i)}\r\|^2
\end{align}
satisfies conditions \eqref{Fcond1eq} and \eqref{Fcond2eq}.
\end{theorem}
Thus, the initial estimator can simply be chosen as
\begin{equation}\label{initialestimatoreq}
\x_0:=\argmin_{\v\in \BR} \sum_{i=0}^j\l\|\mtx{H} \D^i \v-\hat{\Phi}^{(i)}\r\|^2,
\end{equation}
with the above two theorems implying that the estimate gets close to $u$ for decreasing $\sigma_Z \sqrt{h}$. Solving polynomial sum of squares minimisation problems of this type is a well-studied problem in optimisation theory (see \cite{lasserremoments} for a theoretical overview), and several toolboxes are available (see \cite{sostools},  \cite{tutuncu2003solving}). Besides \eqref{Fdefminsquare}, other problem-specific choices of $F$ satisfying conditions \eqref{Fcond1eq} and \eqref{Fcond2eq} can also be used, as we explain in Section \ref{secChoiceofF} for the Lorenz 96' model.

\subsection{Optimisation based smoothing and filtering}
The following algorithm provides an estimator of $\u$ given $\Y_{0:k}$.
We assume that either there is a problem-specific $F$ satisfying conditions \eqref{Fcond1eq} and \eqref{Fcond2eq} for some $j\in \N$, or we suppose that Assumption \ref{assder} is satisfied for the true initial point $\u$, and use $F$ as defined in \eqref{Fdefminsquare}.
\begin{algorithm}[H]
\textbf{Input}: $k\in \N$ (window size parameter), $\Delta_{\min}>0$ (minimum step size parameter), $\Y_{0:k}$ (observations).\\
\textbf{Step 1}: We compute the estimators $\hat{\Phi}^{(0)},\ldots, \hat{\Phi}^{(j)}$ based on  \eqref{hatphildefeq}, and set the initial point as
$x_0:=F\l(\hat{\Phi}^{(0)},\ldots, \hat{\Phi}^{(j)}\r)$.\\
\textbf{Step 2}: We compute the iterates $\x_i$ for $i\ge 1$ based on \eqref{smNewtoniteq} recursively until $\|\x_{i}-\x_{i-1}\|$ becomes smaller than $\Delta_{\min}$, and return $\hat{\u}=\x_n$ for $n:=\min_{i \in \Z_+} \|\x_{i}-\x_{i-1}\|<\Delta_{\min}$.
\caption{Optimisation based smoothing}\label{algsm}
\end{algorithm}

The following algorithm returns an online estimator of $\l(\u(t_i)\r)_{i\ge 0}$ given $\Y_{0:i}$ at time $t_i$.
\begin{algorithm}[H]
\textbf{Input}: $k\in \N$ (window size parameter), $\Delta_{\min}>0$ (minimum step size parameter), $(\Y_{i})_{i\ge 0}$ (observations come consecutively in time).\\
\textbf{Step 1}: For $i< k$, return the estimate $\widehat{\u(t_i)}=0$.\\
\textbf{Step 2}: For $i\ge k$, we first compute the estimate $\hat{\u}^{(i-k)}$ of $\u(t_{i-k})$ based on Algorithm \ref{algsm} applied on $\Y_{i-k:i}$, and then return $\widehat{\u(t_k)}=\Psi_T(\hat{\u}^{(i-k)})$.
\caption{Optimisation based filtering}\label{algfi}
\end{algorithm}
This algorithm can be modified to run Step 2 only at every $K$ step for some $K\in \Z_+$ (i.e. for $i=k+lK$ for $l\in \N$), and propagate forward the estimate of the previous time we ran Step 2 at the intermediate time points. This increases the execution speed at the cost of the loss of some precision (depending on the choice of $K$).

Based on our results in the previous sections, we can see that if the assumptions of the results hold and $\Delta_{\min}$ is chosen sufficiently small then the estimation errors are of $O(\sigma_Z \sqrt{h})$ with high probability for both algorithms (in the second algorithm, for $i\ge k$).

\section{Application to the Lorenz 96' model}\label{secapplications}

The Lorenz 96' model is a $d$ dimensional chaotic dynamical system that was introduced in \cite{lorenz1996predictability}. In its original form, it is written as
\begin{equation}\label{Lorenz96eq}
\frac{d}{dt}u_i=-u_{i-1}u_{i-2}+u_{i-1}u_{i+1}-u_i+f,
\end{equation}
where the indices are understood modulo $d$, and $f$ is the so-called forcing constant. In this paper we are going to fix this as $f=8$ (this is a commonly used value that is experimentally known to cause chaotic behaviour, see \cite{MajdaHarlim} and \cite{MajdaHarlimGershgorin}). As shown on page 16 of \cite{AlonsoStuartLongtime}, this system can be written in the form \eqref{diffeqgeneralform}, and the bilinear form $\mtx{B}(\u,\u)$ satisfies the energy conserving property (i.e.~$\l<\mtx{B}(\v,\v),\v\r>=0$ for any $\v\in \R^d$).

We consider 2 observation scenarios for this model.  In the first scenario, we assume that $d$ is divisible by 6, and choose $\mtx{H}$ such that coordinates $1,2,3$, $7,8,9$, $\ldots$, $d-5,d-4,d-3$ are observed directly, i.e. each observed batch of 3 is followed by a non-observed batch of 3. In this case, the computational speed is fast, and we are able to obtain simulation results for high dimensions.
We consider first a small dimensional case ($d=12$) to show the dependence of the MSE of the MAP estimator on the parameter $k$ (the amount of observations), when the parameters $\sigma_Z$ and $h$ are fixed. After this, we consider a high dimensional case ($d=1000002$), and look at the dependence of the MSE of the MAP estimator on $\sigma_Z\sqrt{h}$.

In the second scenario, we choose $\mtx{H}$ such that we observe the first 3 coordinates directly. We present some simulation results for $d=60$ dimensions for this scenario.

In \cite{ConcentrationProperties}, we have shown that in the second scenario, the system satisfies Assumption \ref{assder} for Lebesgue-almost every initial point $\u\in \R^d$. A simple modification of that argument shows that Assumption \ref{assder} holds for Lebesgue-almost every initial point $\u\in \R^d$ in the first observation scenario too.

In each case, we have set the initial point as $\u=\left(\frac{d+1}{2d},\frac{d+2}{2d},\ldots, 1\right)$ (we have tried different randomly chosen initial points and obtained similar results). Figures \ref{fig:halfobsdepk}, \ref{fig:halfobsdepsigmazsqrth}, and \ref{fig:3obs} show the simulation results when applying Algorithm \ref{algsm} (optimisation based smoother) to each of these cases. Note that Algorithm \ref{algfi} (optimisation based filter) applied to this setting yields very similar results.

In Figure \ref{fig:halfobsdepk}, we can see that the MAP estimators RMSE (root mean square error) does not seem to decrease significantly after a certain amount of observations. This is consistent with the non-concentration of the smoother due to the existence of leaf sets, described in \cite{ConcentrationProperties}. Moreover, by increasing $k$ above 100, we have observed that the Newton's method often failed to improve significantly over the initial estimator, and the RMSE of the estimator became of order $10^{-1}$, significantly worse than for smaller values of $k$. We believe that this is due to the fact that as we increase $k$, while keeping $\sigma_Z$ and $h$ fixed, the normal approximation of the smoother breaks down, and the smoother becomes more and more multimodal. Due to this, we are unable to find the true MAP when starting from the initial estimator, and settle down at another mode. To conclude, for optimal performance, it is important to tune the parameter $k$ of the algorithm.

In Figure \ref{fig:halfobsdepsigmazsqrth}, we present results for a $d=1000002$ dimensional Lorenz 96' system, with half of the coordinates observed. The observation time is $T=10^{-5}$. The circles correspond to data points with $h=10^{-6}$ (so $k=10$), while the triangles correspond to data points with $h=2\cdot 10^{-7}$ (so $k=50$). The plots show that the method works as expected for this high dimensional system, and that the RMSE of the estimator is proportional to $\sigma_Z \sqrt{h}$. 

Finally, in Figure \ref{fig:3obs}, we present results for a $d=60$ dimensional system with the first 3 coordinates observed. The observation time is $T=10^{-3}$. The circles correspond to data points with $h=5 \cdot 10^{-5}$, while the triangles correspond to data points with $h=2.5\cdot 10^{-5}$. We can see that Algorithm \ref{algsm} is able to handle a system which has only a small fraction of its coordinates observed. Note that the calculations are done with numbers having 360 decimal digits of precision. Such high precision is necessary because the interaction between the 3 observed coordinates of the system and some of the non-observed coordinates is weak.

The requirements on the observations noise $\sigma_Z$ and observation time step $h$ for the applicability of our method depend heavily on the parameters of the model (such as the dimension $d$) and on the observation matrix $\mtx{H}$. In the first simulation (Figure \ref{fig:halfobsdepk}), relatively large noise ($\sigma_Z=10^{-3}$) and large time step ($h=10^{-2}$) were possible. For the second simulation (Figure \ref{fig:halfobsdepk}), due to the high dimensionality of the model ($d=1000002$), and the sparse approximation used in the solver, we had to choose smaller time step ($h\le 10^{-6}$) and smaller observation noise ($\sigma_Z\le 10^{-6}$). Finally, in the third simulation (Figure \ref{fig:3obs}), since only a very small fraction of the $d=60$ coordinates is observed, the observation noise has to be very small in order for us to be able to recover the unobserved coordinates ($\sigma_Z\le 10^{-120}$).

In all of the above cases, the error of the MAP estimator is several orders of magnitude less than the error of the initial estimator. When comparing these results  with the simulation results of \cite{law2016filter} using the 3DVAR and Extended Kalman Filter methods for the Lorenz 96' model, it seems that our method improves upon them, since it allows for larger dimensions, and smaller fraction of coordinates observed.

In the following sections, we describe the theoretical and technical details of these simulations. First, in Section \ref{secLorenz96preliminarybounds}, we bound some constants in our theoretical results for the Lorenz 96' model. In Section \ref{secChoiceofF}, we explain the choice of the function $F$ in our initial estimator (see Theorem \ref{initialestthm}) in the two observation scenarios.  In Section \ref{secODETaylor}, we adapt the Taylor expansion method for numerically solving ODEs to our setting. Finally, based on these preliminary results, we give the technical details of the simulations in Section \ref{secsimdetails}.
\begin{figure}[h]
\centering
\begin{subfigure}{0.48\textwidth}
  \centering
  \includegraphics[width=7cm]{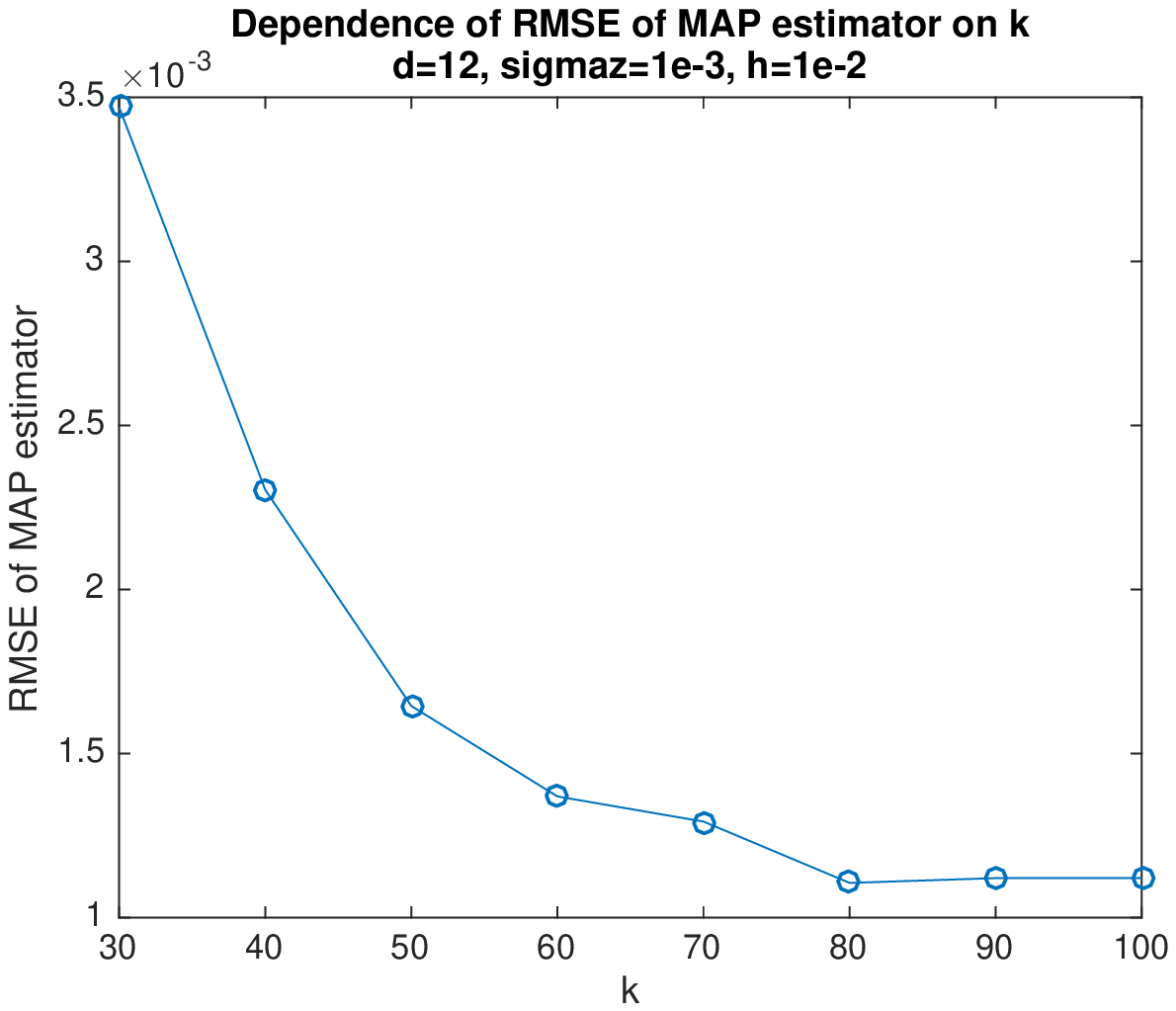}
\end{subfigure}
\begin{subfigure}{0.48\textwidth}
  \centering
  \includegraphics[width=7cm]{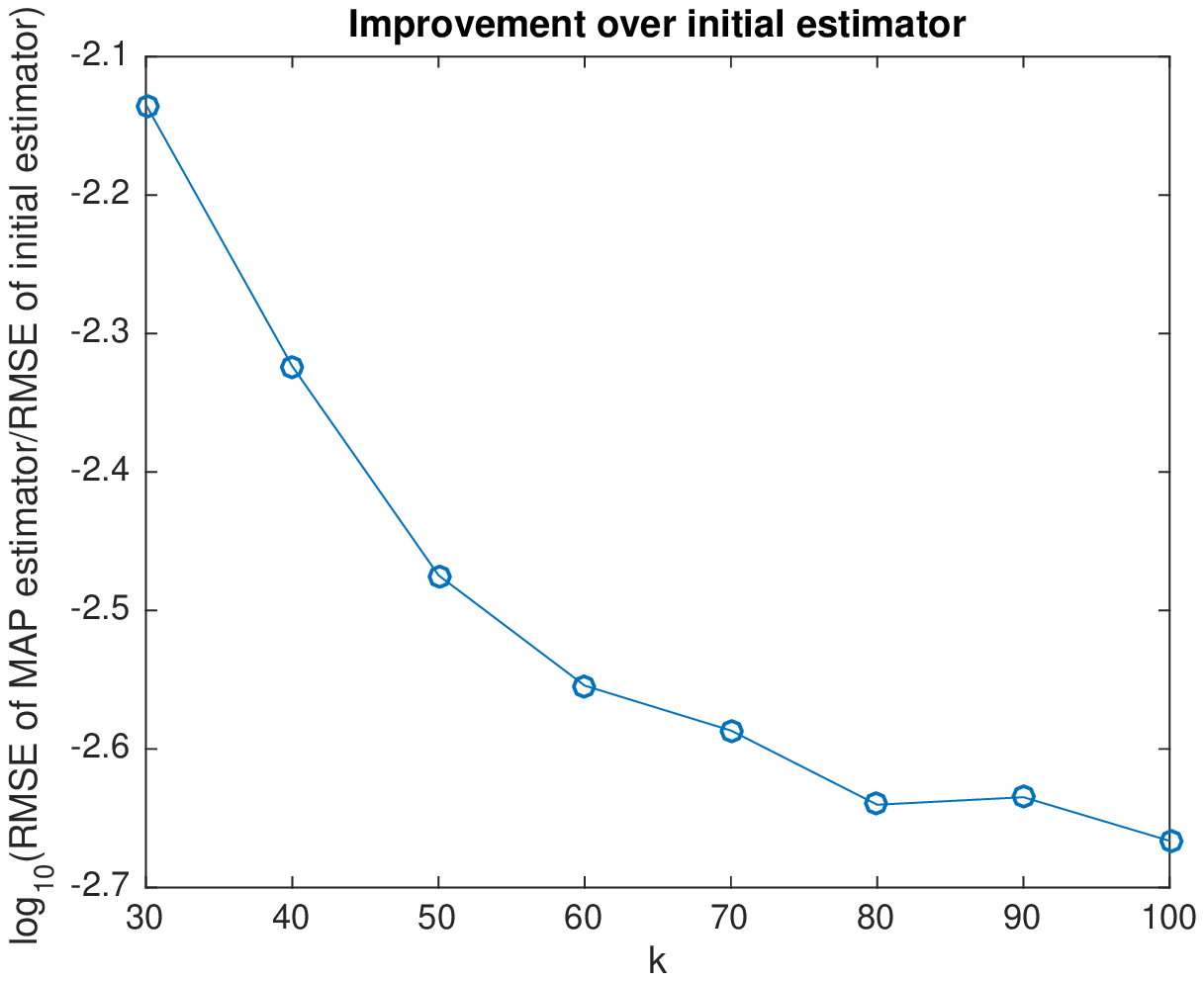}
\end{subfigure}
\caption{Dependence of RMSE of estimator on $k$ for $d=12$}
\label{fig:halfobsdepk}
\end{figure}
\begin{figure}[h]
\centering
\begin{subfigure}{0.48\textwidth}
  \centering
  \includegraphics[width=7cm]{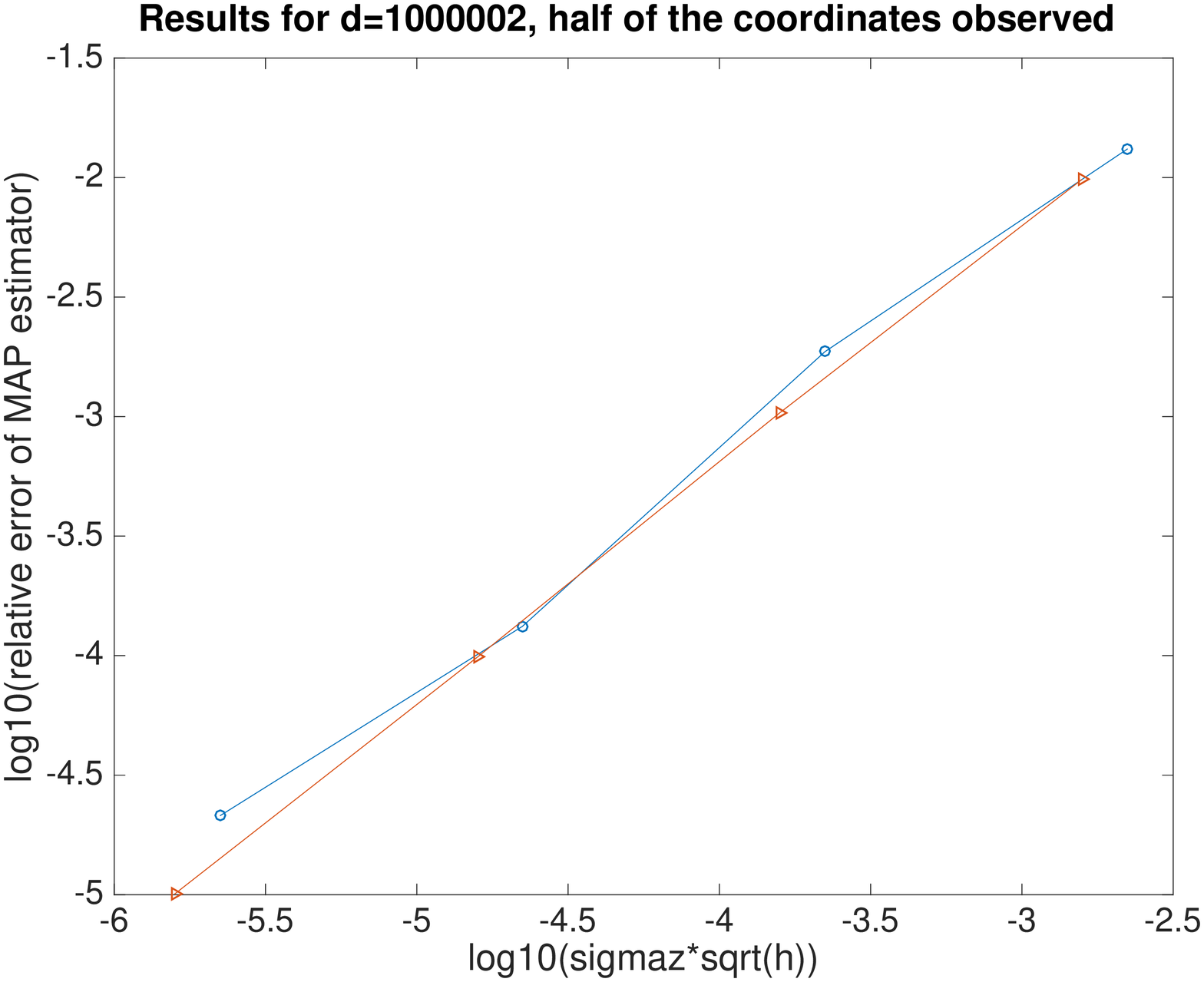}
\end{subfigure}
\begin{subfigure}{0.48\textwidth}
  \centering
  \includegraphics[width=7cm]{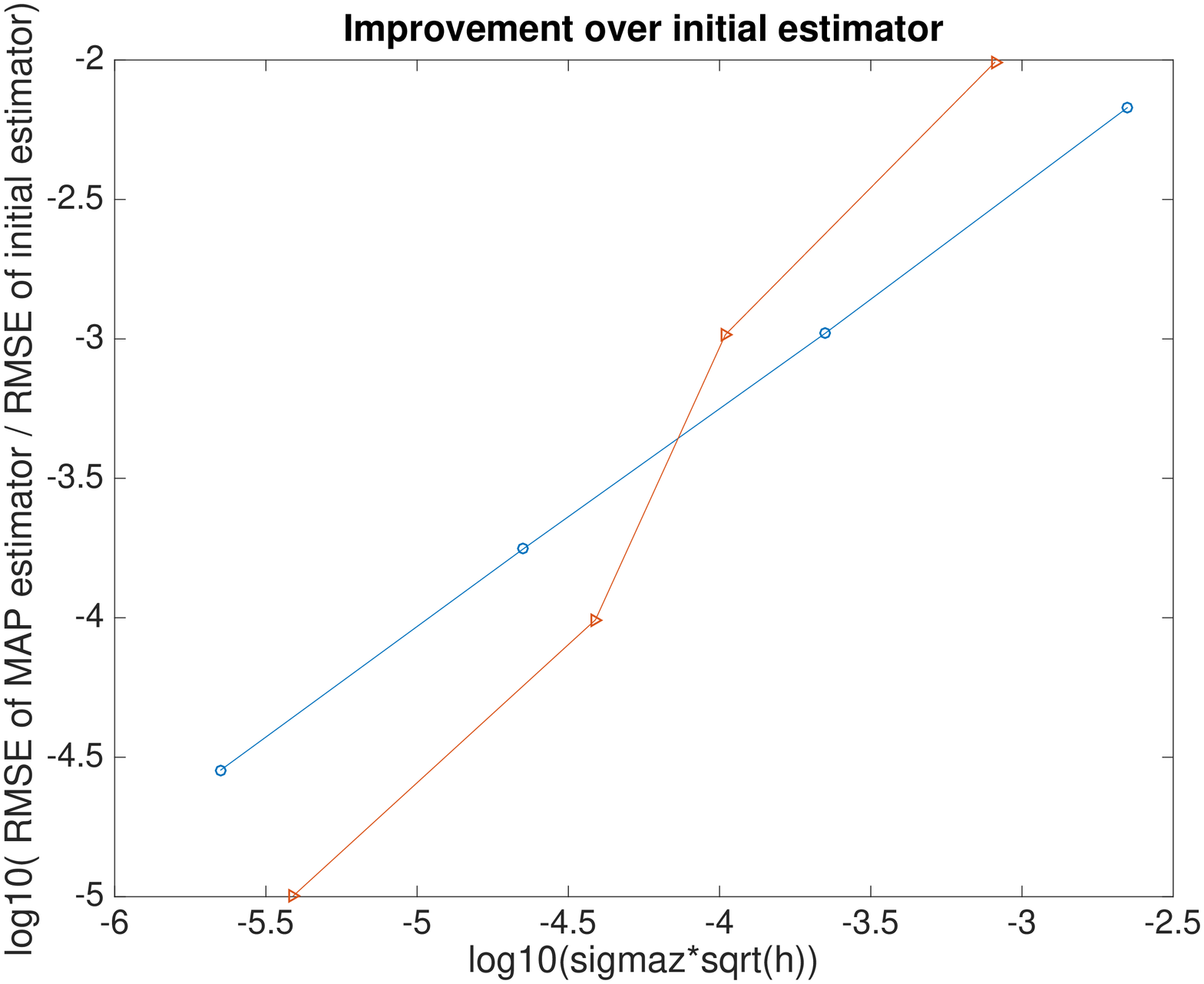}
\end{subfigure}
\caption{Dependence of RMSE of estimator on $\sigma_Z$ and $h$ for $d=1000002$}
\label{fig:halfobsdepsigmazsqrth}
\end{figure}
\begin{figure}[h]
\centering
\begin{subfigure}{.48\textwidth}
  \centering
  \includegraphics[width=7cm]{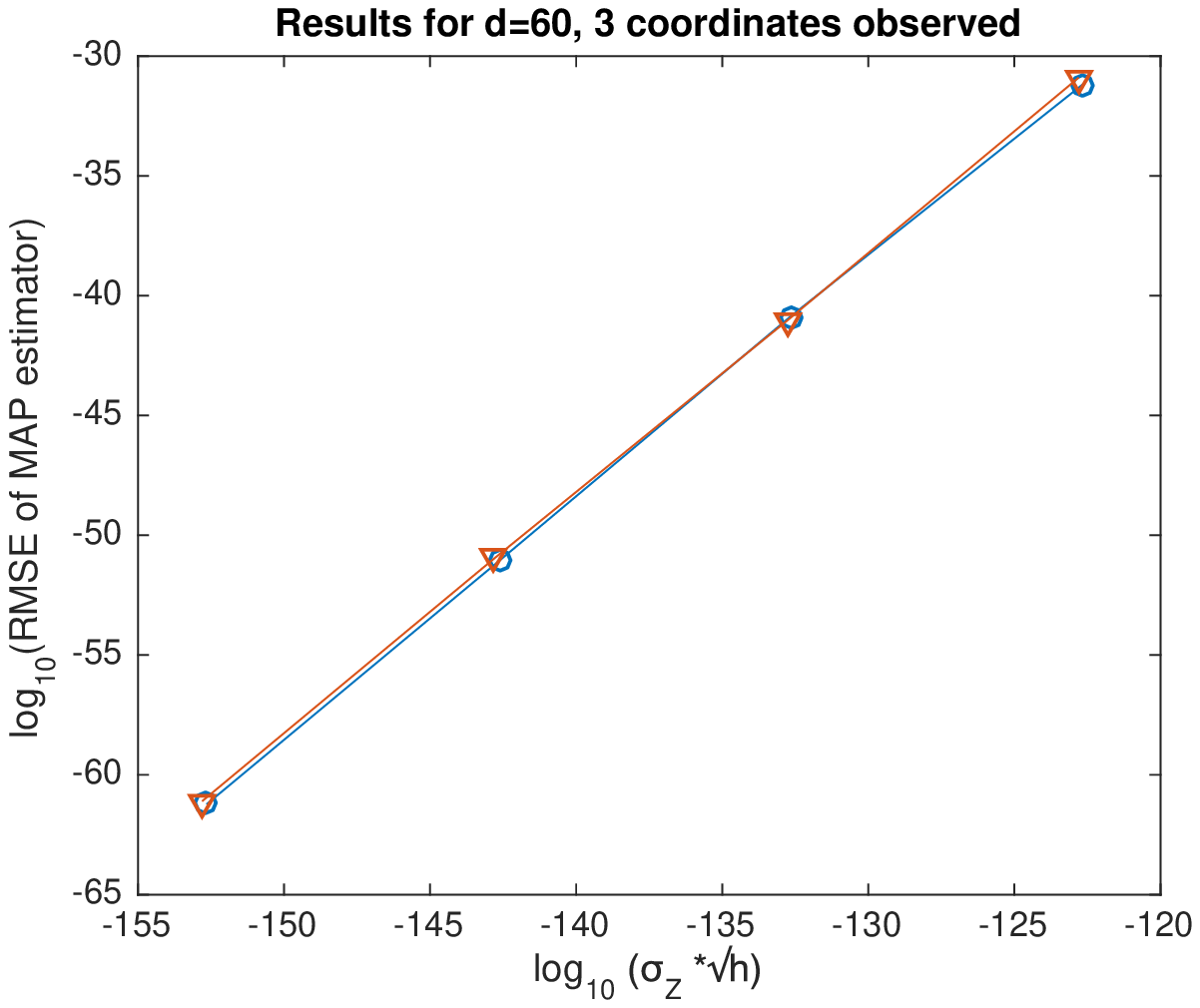}
\end{subfigure}
\begin{subfigure}{.48\textwidth}
  \centering
  \includegraphics[width=7cm]{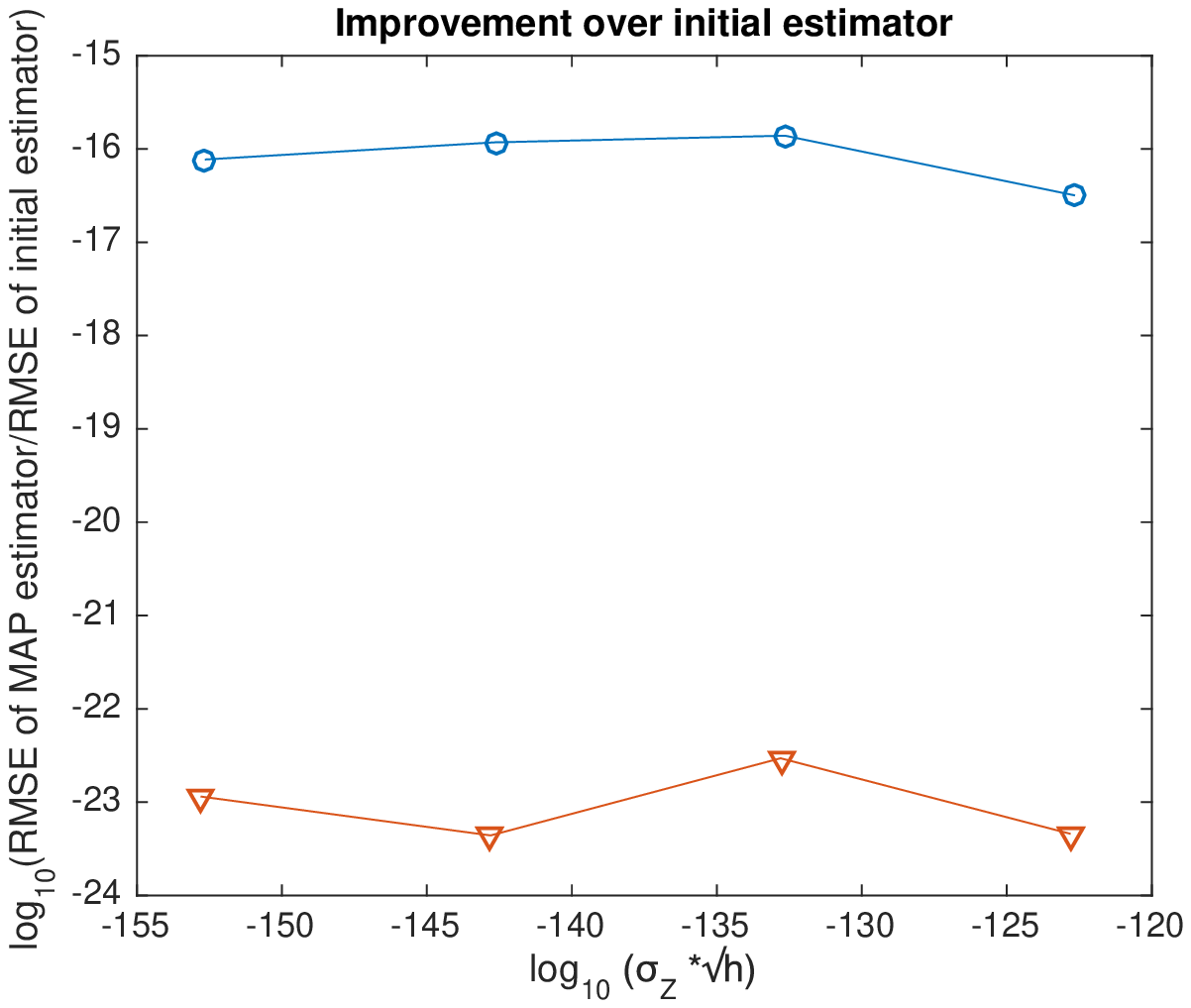}
\end{subfigure}
\caption{Dependence of RMSE of estimator on $\sigma_Z$ and $h$ for $d=60$, 3 coordinates observed}\label{fig:3obs}
\end{figure}

\subsection{Some bounds for Lorenz 96' model}\label{secLorenz96preliminarybounds}
In this section, we will bound some of the constants in Section \ref{SecPreliminaries} for the Lorenz 96' model. Since $\mtx{A}$ is the identity matrix, we have $\|\mtx{A}\|=1$ and satisfies that $\l<\v,\mtx{A}\v\r>\ge \lambda_{\mtx{A}}\v$ for every $\v\in \R^d$ for $\lambda_{\mtx{A}}=1$. The condition that $\l<\mtx{B}(\v,\v),\v\r>=0$ for every $\v\in \R^d$ was verified for the Lorenz 96' model, see Property 3.1 of \cite{law2016filter}. For our choice $f=8$, we have $\|\vct{f}\|=8\sqrt{d}$. 
Thus based on \eqref{Rcondeq}, the trapping ball assumption \eqref{eqtrappingball} is satisfied for the choice $R:=\frac{\|\vct{f}\|}{\lambda_{\mtx{A}}}=8\sqrt{d}$.

For $\mtx{B}$, given any $\u,\v\in \R^d$ such that $\|\u\|,\|\v\|\le 1$, by the arithmetic-mean-root mean square inequality, and the inequality $(ab)^2\le \frac{a^2+b^2}{2}$, we have
\begin{align*}
\|\mtx{B}(\u,\v)\|^2&=\frac{1}{4}\sum_{i=1}^{d}\l(v_{i-1} u_{i+1} + u_{i-1} v_{i+1} -v_{i-2} u_{i-1}-u_{i-2} v_{i-1}\r)^2\\
&\le  (v_{i-1} u_{i+1})^2 + (u_{i-1} v_{i+1})^2 + (v_{i-2} u_{i-1})^2 + (u_{i-2} v_{i-1})^2\le 4,
\end{align*}
thus $\|\mtx{B}\|\le 2$. If $d$ is divisible by 2, then the choice $u_i=v_i=\frac{(-1)^i}{\sqrt{d}}$ shows that this bound is sharp, and $\|\mtx{B}\|=2$.   For simplicity, we have chosen the prior $q$ as the uniform distribution on $\BR$. Based on these, and the definitions \eqref{C0derJdefeq} and \eqref{eqpathdistancebound}, we have
\begin{equation}\label{constbndeq}
C_0\defby 16\sqrt{d},\quad \Cder\le 1+32\sqrt{d}, \quad\text{and}\quad G\le 1+32\sqrt{d}.
\end{equation}

\subsection{Choice of the function $F$ in the initial estimator}\label{secChoiceofF}
In this section, we will construct a computationally simple function $F$ satisfying the conditions of Theorem \ref{initialestthm} for the two observation scenarios.

First, we look at the second scenario, when only the first 3 coordinates are observed. We are going to show that for $j=\l\lceil \frac{d-3}{3}\r\rceil$, it is possible to construct a function $F: (\R^{d_o})^{j+1}\to \R^d$ such that $F$ is computationally simple, Lipschitz in a neighbourhood of $\u$, and $F\l(\mtx{H} \u, \ldots, \mtx{H}\D^j \v\r)=\u$, thus satisfies the conditions of Theorem \ref{initialestthm}. Notice that
\begin{equation}\label{upeq1}
\D u_{i-1}=-u_{i-2}u_{i-3}+u_{i-2}u_{i}-u_{i-1}+f,
\end{equation}
so for $m=0$, we have
\begin{equation}\label{upeq2}
u_i=\D^0 u_{i}=\l(\D u_{i-1}-f+u_{i-1}+u_{i-2}u_{i-3}\r)/u_{i-2}.
\end{equation}
In general, for $m\ge 1$, by differentiating \eqref{upeq1} $m$ times, we obtain that
\begin{equation}\label{upeq3}
\D^{m+1}u_{i-1}=-\D^m u_{i-1}-\sum_{l=0}^{m}\binom{m}{l}\D^l u_{i-2}\cdot \D^{m-l}u_{i-3}+\sum_{l=0}^{m}\binom{m}{l}\D^l u_{i}\cdot \D^{m-l}u_{i-2},
\end{equation}
thus for any $m\ge 1$,
\begin{align}\label{upeq4}
\D^{m}u_{i}=\bigg(\D^{m+1}u_{i-1}+\D^m u_{i-1}+\sum_{l=0}^{m}\binom{m}{l}\D^l u_{i-2}\cdot \D^{m-l}u_{i-3}\\
\nonumber -\sum_{l=0}^{m-1}\binom{m}{l}\D^l u_{i}\cdot \D^{m-l}u_{i-2}\bigg)/u_{i-2}.
\end{align}
Thus for any $m\in \N$, we have a recursion for the $m$th derivative of $u_i$ based on the first $m+1$ derivatives of $u_{i-1}$ and the first $m$ derivatives of $u_{i-2}$ and $u_{i-3}$. Based on this recursion, and the knowledge of the first $j$ derivatives of $u_1$, $u_2$ and $u_3$, we can compute the first $j-1$ derivatives of $u_4$, then the first $j-2$ derivatives of $u_5$, etc. and finally the zeroth derivative of $u_{3+j}$ (i.e. $u_{3+j}$ itself). 

In the other direction, 
\begin{equation}\label{downeq1}
\D u_{i+2}=f-u_{i+2}-u_{i+1}u_i + u_{i+1}u_{i+3},
\end{equation}
therefore for $m=0$, we have
\begin{equation}\label{downeq2}
u_i=\D^0 u_{i}=\l(f-\D u_{i+2}-u_{i+2}+ u_{i+1}u_{i+3}\r)/u_{i+1}.
\end{equation}
By differentiating \eqref{downeq1} $m$ times, we obtain that
\begin{equation}\label{downeq3}
\D^{m+1}u_{i+2}=-\D^m u_{i+2}+\sum_{l=0}^{m}\binom{m}{l}\D^l u_{i+1}\cdot \D^{m-l}u_{i+3}-\sum_{l=0}^{m}\binom{m}{l}\D^l u_{i}\cdot \D^{m-l}u_{i+1},
\end{equation}
thus
\begin{align}\label{downeq4}
\D^m u_{i}=\bigg(-\D^{m+1}u_{i+2}-\D^m u_{i+2}+\sum_{l=0}^{m}\binom{m}{l}\D^l u_{i+1}\cdot \D^{m-l}u_{i+3}\\
\nonumber-\sum_{l=0}^{m-1}\binom{m}{l}\D^l u_{i}\cdot \D^{m-l}u_{i+1}\bigg)/u_{i+1}.
\end{align}
Thus for any $m\in \N$, we have a recursion allowing us to compute the first $m$ derivatives of $u_i$ based on the first $m+1$ derivatives of $u_{i+2}$ and the first $m$ derivatives of $u_{i+1}$ and $u_{i+3}$ (with indices considered modulo $d$). This means that given the first $j$ derivatives of $u_1$, $u_2$ and $u_3$, we can compute the first $j-1$ derivatives of $u_d$ and $u_{d-1}$, then the first $j-2$ derivatives of $u_{d-2}$ and $u_{d-3}$, etc. and finally the zeroth derivatives of $u_{d+2-2j}, u_{d+1-2j}$.

Based on the choice $j:=\l\lceil \frac{d-3}{3}\r\rceil$, these recursions together define a function $F$ for this case. From the recursion formulas, and the boundedness of $\u$ it follows that $F$ is Lipschitz in a neighbourhood of $\l(\mtx{H}\u, \ldots, \mtx{H} \D^j \u\r)$ as long as none of the coordinates of $\u$ is 0 (thus for Lebesgue-almost every $\u\in \BR$).

Now we look at the first observation scenario, i.e.~suppose that $d$ is divisible by 6, and we observe coordinates $(6i+1, 6i+2, 6i+3)_{0\le i\le d/6-1}$. In this case, we choose $j=1$, and define $F\l(\mtx{H} \u, \mtx{H} \D \u\r)$ based on formulas \eqref{upeq2} and \eqref{downeq2}, so that we can express $u_{6i+4}, u_{6i-1}, u_{6i-2}$ based on $u_{6i+1}, u_{6i+2}, u_{6i+3}$ and $\D u_{6i+1}, \D u_{6i+2}, \D u_{6i+3}$ for $0\le i\le d/6 -1$ (with indices counted modulo $d$). Based on the equations \eqref{upeq2} and \eqref{downeq2}, we can see that $F$ defined as above satisfies the conditions of Theorem \ref{initialestthm} as long as none of the coordinates of $u$ is 0 (thus for Lebesgue-almost every $\u\in \BR$). In both scenarios, $F$ is computationally simple.

We note that in the above argument, $F$ might be not defined if some of the components of $\u$ are 0 (and the proof of Assumption \ref{assder} in \cite{ConcentrationProperties} also requires that none of the components are 0). Moreover, due to the above formulas, some numerical instability might arise when some of the components of $\u$ are very small in absolute value. In Section \ref{secinitialestcompzero} of the Appendix, we state a simple modification of the initial estimator of Theorem \ref{initialestthm} based on the above $F$ that is applicable even when some of the components of $\u$ are zero.

\subsection{Numerical solution of chaotic ODEs based on Taylor expansion}\label{secODETaylor}
Let $\v\in \BR$, and $\imax\in \N$. The following lemma provides some simple bounds that allow us to approximate the quantities $\v(t)=\Psi_t(\v)$ for sufficiently small values of $t$ by the sum of the first $\imax$ terms in their Taylor expansion. These bounds will be used to simulate the system \eqref{diffeqgeneralform}, and to implement Newton's method as described in \eqref{smNewtoniteq}.

\begin{lem}\label{Taylorapproxlemma}
For any $\v\in \BR$, we have
\begin{equation}\label{eqvbnd}
\l\|\v(t)-\sum_{i=0}^{\imax}\frac{t^i}{i!} \D^i \v\r\|\le C_0 (\Cder t)^{\imax+1},
\end{equation}
where $\D^0=\v$, and for $i\ge 1$, we have the recursion
\begin{equation}\label{eqvrec}
\D^i \v=-\mtx{A} \D^{i-1} \v-\sum_{j=0}^{i-1} \binom{i-1}{j}\mtx{B}\l(\D^{j} \v,\D^{i-1-j} \v\r).
\end{equation}
\end{lem}
\begin{proof}
The bounds on the error in the Taylor expansion follow from \eqref{uderboundeq}. The recursion equation is just \eqref{udereq}.
\end{proof}

The following proposition shows a simple way of estimating $\v(t)$ for larger values $t$ (the bound \eqref{eqvbnd} is not useful for $t\ge \frac{1}{\Cder}$). For $\v\in \R^d$, we let 
\begin{equation}\label{PBRdefeq}
P_{\BR}(v):=\v\cdot 1_{[\v\in \BR]}+ R\frac{\v}{\|\v\|}\cdot 1_{[\v\notin \BR]}
\end{equation}
be the projection of $\v$ on $\BR$.
\begin{prop}\label{Taylorapproxprop}
Let $\v\in \BR$, and $\Delta< \frac{1}{\Cder}$. Let $\widehat{\v}(0)=\v$, $\widehat{\v}(\Delta):=P_{\BR}\l(\sum_{i=0}^{\imax}\frac{\Delta^i}{i!} \D^i \v\r)$, 
and similarly, given $\widehat{\v}(j \Delta)$, define $\widehat{\v}((j+1)\Delta):=P_{\BR}\l(\sum_{i=0}^{\imax}\frac{\Delta^i}{i!} \D^i (\widehat{\v}(j\Delta))\r).$
Finally, let $\delta:=t-\lfloor \frac{t}{\Delta}\rfloor \Delta$ and
$\widehat{\v}(t):=P_{\BR}\l(\sum_{i=0}^{\imax}\frac{\delta^i}{i!} \D^i( \widehat{\v}(\lfloor \frac{t}{\Delta}\rfloor \Delta))\r).$
Then the error is bounded as 
\[\|\widehat{\v}(t)-\v(t)\|\le (t+\Delta) \exp(Gt)C_0\Cder \cdot \l(\Cder \Delta\r)^{\imax}.\]
\end{prop}
\begin{proof}
Using \eqref{eqvbnd} and the fact that the projection $P_{\BR}$ decreases distances we know that for every $j\in \l\{0, 1,\ldots,\lceil t/\Delta\rceil-1\r\}$, $\|\widehat{\v}((j+1)\Delta)-\Psi_{\Delta}(\widehat{\v}(j\Delta))\|\le C_0 \l(\Cder \Delta\r)^{\imax+1}$. By inequality \eqref{eqpathdistancebound}, it follows that 
\[\|\Psi_{t-j\Delta}(\widehat{\v}(j\Delta))-\Psi_{t-(j+1)\Delta}(\widehat{\v}((j+1)\Delta))\|\le \exp(Gt) C_0 \l(\Cder \Delta\r)^{\imax+1},\]
and using the triangle inequality, we have
\[ \|\v(t)-\widehat{\v}(t)\|\le \sum_{j=0}^{\lceil t/\Delta\rceil-1} \|\Psi_{t-j\Delta}(\widehat{\v}(j\Delta))-\Psi_{t-(j+1)\Delta}(\widehat{\v}((j+1)\Delta))\|,\]
so the claim follows.
\end{proof}

\subsection{Simulation details}\label{secsimdetails}
The algorithms were implemented in Julia, and ran on a computer with a 2.5Ghz Intel Core i5 CPU. In all cases, the convergence of Newton's method up to the required precision (chosen to be much smaller than the RMSE) occurred in typically 3-8 steps. 

In the case of Figure \ref{fig:halfobsdepk} ($d=12$, half of the coordinates observed) the observation time $T$ is much larger than $\frac{1}{\Cder}$, so we have used the method of Proposition \ref{Taylorapproxprop} to simulate from the system. The gradient and Hessian of the function $\gsm$ were approximated numerically based on finite difference formulas (requiring $O(d^2)$ simulations from the ODE).  We have noticed that in this case, the Hessian has elements with significantly large absolute value even far away from the diagonal. The running time of Algorithm \ref{algsm}  was approximately 1 second. The RMSEs were numerically approximated from 20 parallel runs. The parameters $J_{\max}^{(0)}$ and $J_{\max}^{(1)}$ of the initial estimator were chosen as 1.

In the case of Figure \ref{fig:halfobsdepk} ($d=1000002$, half of the coordinates observed), we could not use the same simulation technique as previously (finite difference approximation of the gradient and Hessian of $\gsm$) because of the huge computational and memory requirements. Instead, we have computed the Newton's method iterations described in \eqref{smNewtoniteq} based on preconditioned conjugate gradient solver, with the gradient and the product of the Hessian with a vector were evaluated  based on adjoint methods as described by equations (3.5)-(3.7) and Section 3.2.1 of \cite{paulin20174d} (see also \cite{le2002second}). This means that the Hessians were approximated using products of Jacobian matrices that were stored in sparse format due to the local dependency of the equations \eqref{Lorenz96eq}. This efficient storage has allowed us to run Algorithm \ref{algsm} in approximately 20-40 minutes in the simulations. We made 4 parallel runs to estimate the MSEs. The parameters $J_{\max}^{(0)}$ and $J_{\max}^{(1)}$ of the initial estimator were chosen as 2.

Finally, in the case of Figure \ref{fig:3obs} ($d=60$, first 3 coordinates observed), we used the same method as in the first example (finite difference approximation of the gradient and Hessian of $\gsm$). The running time of Algorithm \ref{algsm} was approximately 1 hour (in part due to the need of using arbitrary precision arithmetics with hundreds of digits of precision). The MSEs were numerically approximated from 2 parallel runs. The parameters $J_{\max}^{(0)}, \ldots, J_{\max}^{(19)}$ of the initial estimator were chosen as 24.

\section{Preliminary results}\label{secpreliminaries}
The proof of our main theorems are based on several preliminary results. Let
\begin{align}\label{lsmdefeq}
\lsm(\v)&:=\sum_{i=0}^{k} \l(\|\Phi_{t_i}(\v)-\Phi_{t_i}(\u)\|^2 +2\l<\Phi_{t_i}(\v)-\Phi_{t_i}(\u), \vZ_i\r>\r), \text{ and }\\
\label{lsmGdefeq}\lsmG(\v)&:=(\v-\u)'\mtx{A}_k(\v-\u)+2\l<\v-\u,\vct{B}_k\r>.
\end{align}
These quantities are related to the log-likelihoods of the smoothing distribution and its Gaussian approximation as
\begin{align}\label{musmlsmeq}
\musm(\v|\vct{Y}_{0:k})&=\frac{q(\v)}{C_k^{\mathrm{sm}}}\exp\l[-\frac{\lsm(\v)}{2\sigma_Z^2}\r], \text{ and if }\mtx{A}_k\succ \mtx{0}, \text{ then }\\
\label{musmGlsmGeq}
\musmG(\v|\vct{Y}_{0:k})&=\frac{\det(\mtx{A}_k)^{1/2}}{(2\pi)^{d/2}\cdot \sigma_Z^d}\cdot \exp\l[-\frac{\vct{B}_k\mtx{A}_k^{-1}\vct{B}_k }{2\sigma_Z^2}\r]\cdot \exp\l[-\frac{\lsmG(\v)}{2\sigma_Z^2}\r],
\end{align}
where $\succ$ denotes the positive definite order (i.e. $\mtx{A}\succ \mtx{B}$ if and only if $\mtx{A}-\mtx{B}$ is positive definite). Similarly, $\succeq$ denotes the positive semidefinite order.

The next three propositions show various bounds on the log-likelihood related quantity $\lsm(\v)$. Their proof is included in Section \ref{Secproofpreliminaryresults} of the Appendix.
\begin{prop}[A lower bound on the tails of $\lsm(\v)$]\label{lsmlowerboundprop}
Suppose that Assumption \ref{assgauss1} holds for $\u$, then for any $0<\epsilon\le 1$, for every $\sigma_Z>0$, $h\le h_{\max}(\u,T)$, we have
\[\PP\l(\l. \lsm(\v)\ge \frac{c(\u,T)}{h} \|\v-\u\|^2- \frac{C_1(\u,T,\epsilon)\sigma_Z}{\sqrt{h}}\cdot \|\v-\u\|\text{ for every }\v\in \BR\r|\u\r)\ge 1-\epsilon,\]
where
\begin{align*}
C_1(\u,T,\epsilon):=
44(\hM_2(T)R+\hM_1(T))\sqrt{\olT(d+1)d_o}+2\sqrt{2\olT d_o \hM_1(T)\log\l(\frac{1}{\epsilon}\r)}.
\end{align*}
\end{prop}
\begin{prop}[A bound on the difference between $\lsm(\v)$ and $\lsmG(\v)$]\label{lsmlsmGdiffprop}
Suppose that Assumption \ref{assgauss1} holds for $\u$, then for any $0<\epsilon\le 1$, $\sigma_Z>0$, $0< h\le h_{\max}(\u,T)$, we have
\[\PP\l(\l. |\lsm(\v) - \lsmG(\v)|\le \|\v-\u\|^3 \cdot \frac{C_2(\u,T)+C_3(\u,T,\epsilon)\sigma_Z\sqrt{h}}{h}\text{ for every }\v\in \BR\r|\u\r)\ge 1-\epsilon,\]
where
\begin{align*}
&C_2(\u,T):=\olT \hM_1(T)\hM_2(T)\text{ and }\\
&C_3(\u,T,\epsilon):=22(\hM_3(T)+\hM_4(T)R)\sqrt{(d+1)d_o\olT}+
\sqrt{\frac{4}{3} \olT \hM_3(T)d_o \log\l(\frac{2}{\epsilon}\r)}.
\end{align*}
\end{prop}
\begin{prop}[A bound on the difference between $\grad\lsm(\v)$ and $\grad\lsmG(\v)$]\label{lsmlsmGgraddiffprop}
Suppose that Assumption \ref{assgauss1} holds for $\u$, then for any $0<\epsilon\le 1$, $\sigma_Z>0$, $0<h\le h_{\max}(\u,T)$, we have
\[\PP\l(\l.\|\grad\lsm(\v) - \grad\lsmG(\v)\|\le \|\v-\u\|^2 \cdot \frac{C_4(\u,T)+C_5(\u,T,\epsilon)\sigma_Z\sqrt{h}}{h}\text{ for every }\v\in \BR\r|\u\r)\ge 1-\epsilon,\]
where
\begin{align*}
C_4(\u,T)&:=4 \olT \hM_1(T)\hM_2(T), \text{ and }\\
C_5(\u,T,\epsilon)&:=66 \l(\hM_3(T)+\hM_4(T) R\r) \sqrt{\olT (2d+1)d_o}
+2\sqrt{\olT \hM_3(T)d_o \log\l(\frac{1}{\epsilon}\r)}.
\end{align*}
\end{prop}

The following lemma is useful for controlling the total variation distance of two distributions that are only known up to normalising constants.
\begin{lem}\label{dtvboundlemma}
Let $f$ and $g$ be two probability distributions which have densities on $\R^d$ with respect to the Lebesgue measure. Then their total variation distance satisfies that for any constant $c>0$,
\[\dtv(f,g)\le \int_{\x\in \R^d}|f(\x)-cg(\x)|d\x.\]
\end{lem}
\begin{proof}
We are going to use the following characterisation of the total variation distance,
\begin{align*}
\dtv(f,g)&=\frac{1}{2}\int_{\x\in \R^d}|f(\x)-g(\x)|d\x\\
&=\int_{\x\in \R^d}(f(\x)-g(\x))_+d\x=\int_{\x\in \R^d}(f(\x)-g(\x))_-d\x.
\end{align*}
Based on this, the result is trivial for $c=1$. If $c>1$, then we have
\[\int_{\x\in \R^d}(f(\x)-g(\x))_-d\x\le \int_{\x\in \R^d}(f(\x)-cg(\x))_-d\x\le \int_{\x\in \R^d}|f(\x)-cg(\x)| d\x,\]
and the $c<1$ case is similar.
\end{proof}

The following lemma is useful for controlling the Wasserstein distance of two distributions.
\begin{lem}\label{dWboundlemma}
Let $f$ and $g$ be two probability distributions which have densities on $\R^d$ with respect to the Lebesgue measure, also denoted by $f(\x)$ and $g(\x)$. Then their Wasserstein distance (defined as in \eqref{dWdefeq}) satisfies that for any $\y\in \R^d$,
\[\dW(f,g)\le \int_{\x\in \R^d}|f(\x)-g(\x)|\cdot \|\x-\y\|d\x.\]
\end{lem}
\begin{proof}
Let $m(\x):=\min(f(\x),g(\x))$,  $\gamma:=\int_{\x\in\R^d}m(\x)d\x$, $\hat{f}(\x):=f(\x)-m(\x)$, $\hat{g}(\x):=g(\x)-m(\x)$. 
Suppose first that $\gamma\ne 0$ and $1-\gamma\ne 0$. Let $\nu^{(1)}_{f,g}$ denote the distribution of a random vector $(\vct{X},\vct{X})$ on $\R^d\times \R^d$ for a random variable $\vct{X}$ with distribution with density $m(\x)/\gamma$ (the two components are equal). Let $\nu^{(2)}_{f,g}$ be a distribution on $\R^d\times \R^d$ with density $\nu^{(2)}_{f,g}(\x_1,\x_2):=\frac{\hat{f}(\x_1)}{1-\gamma}\cdot \frac{\hat{g}(\x_2)}{1-\gamma}$ (the two components are independent).

We define the optimal coupling of $f$ and $g$ as a probability distribution $\nu_{f,g}$ on $\R^d\times \R^d$ as a mixture of $\nu^{(1)}_{f,g}$ and $\nu^{(2)}_{f,g}$, that is, for any Borel-measurable $E\in \R^d\times \R^d$,
\begin{equation}\label{nudefeq}
\nu_{f,g}(E):=\gamma \nu^{(1)}_{f,g}(E)+ (1-\gamma)\nu^{(2)}_{f,g}(E).
\end{equation}
It is easy to check that this distribution has marginals $f$ and $g$, so by \eqref{dWdefeq}, and the fact that $\int_{\x_1\in \R^d}\hat{f}(\x_1) d\x_1= \int_{\x_2\in \R^d}\hat{g}(\x_2) d\x_2=1-\gamma$, we have
\begin{align*}\dW(f,g)&\le \int_{\x_1, \x_2\in \R^d} \|\x_1-\x_2\|d\nu_{f,g}(\x_1,\x_2)\\
&= \int_{\x_1, \x_2\in \R^d} \frac{\hat{f}(\x_1)\hat{g}(\x_2) }{1-\gamma}\cdot \|\x_1-\x_2\| d\x_1 d\x_2\\
&\le  \int_{\x_1, \x_2\in \R^d} \frac{\hat{f}(\x_1)\hat{g}(\x_2) }{1-\gamma}\cdot \l(\|\x_1-\y\|+\|\x_2-\y\|\r) d\x_1 d\x_2\\
&= \int_{\x_1\in \R^d}\hat{f}(\x_1)\|\x_1-\y\|d\x_1 + \int_{\x_2\in \R^d}\hat{g}(\x_2)\|\x_1-\y\|d\x_2 
\\
&=\int_{\x\in \R^d}(\hat{f}(\x)+\hat{g}(\x))\|\x-\y\|d\x,
\end{align*}
and the result follows from the fact that $\hat{f}(\x)+\hat{g}(\x)=|f(\x)-g(\x)|$. Finally, if $\gamma=0$ then we can set $\nu_{f,g}$ as $\nu^{(2)}_{f,g}$, and the same argument works, while if $\gamma=1$, then both sides of the claim are zero (since $f(\x)=g(\x)$ Lebesgue-almost surely).
\end{proof}

The following two lemmas show concentration bounds for the norm of $\vct{B}_k$ and the smallest eigenvalue of $\mtx{A}_k$. The proofs are included in Section \ref{Secproofpreliminaryresults} of the Appendix (they are based on matrix concentration inequalities).
\begin{lem}\label{Bknormboundlemma}
Suppose that Assumption \ref{assgauss1} holds for $\u$, then the random vector $\vct{B}_k$ defined in \eqref{eqBkdef} satisfies that for any $t\ge 0$,
\begin{equation}\label{eqBkbnd1}\PP(\l.\|\vct{B}_k\|\ge t\r|\u)\le (d+1)\exp\l(-\frac{t^2}{(k+1)d_o \hM_1(T)^2\sigma_Z^2}\r).\end{equation}
Thus for any $0<\epsilon\le 1$, we have
\begin{equation}\label{eqBkbnd2}
\PP\l(\l.\|\vct{B}_k\|\ge C_{\vct{B}}(\epsilon) \frac{\sigma_Z}{\sqrt{h}}\r|\u\r)\le \epsilon \text{ for }C_{\vct{B}}(\epsilon):=\sqrt{\log\l(\frac{d+1}{\epsilon}\r)\cdot \hM_1(T)^2  \olT d_o}.
\end{equation}
\end{lem}

\begin{lem}\label{Akeigboundlemma}
Suppose that Assumption \ref{assgauss1} holds for $\u$, then the random matrix $\mtx{A}_k$ defined in \eqref{eqAkdef} satisfies that for any $t\ge 0$,
\begin{align}
\nonumber &\PP\l(\l.\lambda_{\min}(\mtx{A}_k)\le \frac{c(\u,T)}{h}-t \text { or } \|\mtx{A}_k\|\ge \hM_1(T)^2\cdot \frac{\olT}{h}+t\r|\u\r)\\
\label{eqlambdaminAk1bnd} &\le  2d\exp\l(-\frac{t^2 }{(k+1) \sigma_Z^2 \hM_2(T)^2  d_o}\r),
\end{align}
thus for any $0<\epsilon\le 1$, we have
\begin{align}\label{eqlambdaminAk2bnd}
\nonumber&\PP\Bigg(\lambda_{\min}(\mtx{A}_k)> \frac{c(\u,T)- C_{\mtx{A}}(\epsilon) \sigma_Z \sqrt{h} }{h} \\
&\quad\text{ and } \|\mtx{A}_k\|< \hM_1(T)^2\cdot \frac{\olT}{h}+\frac{C_{\mtx{A}}(\epsilon) \sigma_Z}{\sqrt{h}} \Bigg)\le \epsilon \text{ for }\\
\nonumber &C_{\mtx{A}}(\epsilon):=\sqrt{\log\l(\frac{2d}{\epsilon}\r)\cdot \hM_2(T)^2  \olT d_o}.
\end{align}
\end{lem}

The following proposition bounds the derivatives of the log-determinant of the Jacobian.
\begin{prop}\label{proplogdetder}
The function $\log \det \J\Psi_T: \BR\to \R$ is continuously differentiable on $\mr{int}(\BR)$ and its derivative can be bounded as 
\[\sup_{\v\in \mr{int}(\BR)}\|{\grad}\log \det \J\Psi_T(\v)\|\le M_1(T)M_2(T) d.\]
\end{prop}
\begin{proof}
Let $\mathbb{M}^d$ denote the space of $d\times d$ complex matrices. 
By the chain rule, we can write the derivatives of the log-determinant as functions of the derivatives of the determinant, which were shown to exist in \cite{determinantderivative}.  Following their notation, we define the $k$th derivative of $\det$ at a point $\mtx{M}\in \mathbb{M}^d$ as a map $D^k \det \mtx{M}$ from 
$(\mathbb{M}^d)^k$ to $\C$ with value
\[D^k \det \mtx{M}(\mtx{X}^{1},\ldots, \mtx{X}^{k}):=\left.\frac{\partial^k}{\partial t_1\ldots \partial t_k}\right|_{t_1=\ldots=t_k=0} \det\l(\mtx{M}+t_1\mtx{X}^{1}+\ldots+ t_k\mtx{X}^{k}\r). \]
They have defined the norm of $k$th derivative of the determinant as
\[\|D^k \det \mtx{M}\|:=\sup_{\{\mtx{X}^{i}\}_{1\le i\le k}: \|\mtx{X}^{i}\|=1 \text{ for } 1\le i\le k} \l|D^k \det \mtx{M}(\mtx{X}^{1},\ldots, \mtx{X}^{k})\r|.\]
From Theorem 4 of \cite{determinantderivative}, it follows that for any $k\ge 1$, the norm of the $k$th derivative can be bounded as
\begin{equation}\label{detkthderbndeq}
\|D^k \det \mtx{M}\|\le \|\mtx{M}\|^k d^k\cdot |\det \mtx{M}|.
\end{equation}
Based on the chain rule, the norm first derivative of the log-determinant can be bounded as
\begin{align*}
\|\grad\log \det \J\Psi_T(\v)\|=\l|
\frac{\grad \det \J \Psi_T(\v)}{\det \J \Psi_T(\v)}\r|\le 
\frac{\|D\det \J\Psi_T(\v)\| \|\J^2\Psi_T(\v)\|}{|\det \J \Psi_T(\v)|},
\end{align*}
and the result follows by \eqref{MkTdef} and \eqref{detkthderbndeq} for $k=1$ and $\mtx{M}=\J\Psi_T(\v)$.
\end{proof}
\section{Proof of the main results}\label{secproofmain}

\subsection{Gaussian approximation}
In the following two subsections, we are going to prove our Gaussian approximation results for the smoother, and the filter, respectively.

\subsubsection{Gaussian approximation for the smoother}
In this section, we are going to describe the proof of Theorem \ref{Gaussianapproxthmsm}. First, we will show that the result can be obtained by bounding 3 separate terms, then bounds these in 3 lemmas, and finally combine them.

By choosing the constants $C^{(1)}_{\mr{TV}}(\u,T)$ and $C^{(2)}_{\mr{TV}}(\u,T)$ sufficiently large, we can assume that $C_{\mr{TV}}(\u,T,\epsilon)$ satisfies the following bounds,
\begin{align}
\label{CTVb1}C_{\mr{TV}}(\u,T,\epsilon)&\ge \f{2C_{\mtx{A}}(\f{\epsilon}{4})}{c(\u,T)}\\
\label{CTVb2}C_{\mr{TV}}(\u,T,\epsilon)&\ge \f{C_{3}(\u,T,\f{\epsilon}{4})}{C_2(\u,T)}\\
\label{CTVb3}C_{\mr{TV}}(\u,T,\epsilon)&\ge \l(2 C_2(\u,T) \min\l(\f{1}{2C_{q}^{(1)}},R-\|\u\|\r)\r)^{-3/2}\\
\label{CTVb4}C_{\mr{TV}}(\u,T,\epsilon)&\ge 512 \l(\f{C_2(\u,T) C_{\vct{B}}(\f{\epsilon}{4})}{c(\u,T)}\r)^3.\\
\label{CTVb5}C_{\mr{TV}}(\u,T,\epsilon)&\ge 64 \l(\f{C_1(\u,T,\epsilon) C_2(\u,T)}{c(\u,T)}\r)^3.
\end{align}
Based on the assumption that $\sigma_Z\sqrt{h}\le \frac{1}{2}C_{\mr{TV}}(\u,T,\epsilon)^{-1}$,  \eqref{CTVb1}, and  Lemma \ref{Akeigboundlemma}, we have
\begin{align}\label{Aknormboundeq1}
&\Pcu{\f{c(\u,T)}{2h}\mtx{I}_d\prec \mtx{A}_k \prec \f{C_{\|\mtx{A}\|}}{h}\mtx{I}_d}\ge 1-\f{\epsilon}{4},
\end{align}
where $C_{\|\mtx{A}\|}$ is defined as in \eqref{CAkdef}. The event $\lambda_{\min}(\mtx{A}_k)> \f{c(\u,T)}{2h}$ implies in particular that $\mtx{A}_k$ is positive definite. From Lemma \ref{Bknormboundlemma}, we know that
\begin{equation}\label{Bknormboundeq1}
\Pcu{\|\vct{B}_k\|< C_{\vct{B}}\l(\f{\epsilon}{4}\r) \cdot \f{\sigma_Z}{\sqrt{h}}}\ge 1-\f{\epsilon}{4}.
\end{equation}
From Proposition \ref{lsmlowerboundprop}, we know that
\begin{equation}\label{lsmlowerboundpropeq}
\Pcu{\lsm(\v)\ge \f{c(\u,T)}{h} \|\v-\u\|^2- \f{C_1\l(\u,T,\f{\epsilon}{4}\r)\sigma_Z}{\sqrt{h}}\cdot \|\v-\u\|\text{ for every }\v\in \BR}\ge 1-\f{\epsilon}{4}.\end{equation}
Finally, from Proposition \ref{lsmlsmGdiffprop}, it follows that
\begin{equation}\label{lsmlsmGdiffpropeq}
\Pcu{|\lsm(\v) - \lsmG(\v)|\le \|\v-\u\|^3 \cdot \f{C_2(\u,T)+C_3\l(\u,T,\f{\epsilon}{4}\r)\sigma_Z\sqrt{h}}{h}\text{ for every }\v\in \BR}\ge 1-\f{\epsilon}{4}.
\end{equation}

In the rest of the proof, we are going to assume that all four of the events in the equations \eqref{Aknormboundeq1},\eqref{Bknormboundeq1},\eqref{lsmlowerboundpropeq}, and \eqref{lsmlsmGdiffpropeq} hold. From the above bounds, we know that this happens with probability at least $1-\epsilon$.

Let $\vZ$ be a $d$ dimensional standard normal random vector, then from the definition of $\musmG$, it follows that when conditioned on $\vct{B}_k$ and $\vct{A}_k$, 
\begin{equation}\label{Wdefeq}
\vct{W}:=\sigma_Z \cdot \mtx{A}_k^{-1/2} \cdot \vZ+\u-\mtx{A}_k^{-1}\vct{B}_k
\end{equation}
has distribution $\musmG(\cdot |\vct{Y}_{0:k})$. This fact will be used in the proof several times.

Since the normalising constant of the smoother, $C_{k}^{\mathrm{sm}}$ of \eqref{eqmusmgaussian} is not known, it is not easy to bound the total variation distance of the two distributions directly. Lemma \ref{dtvboundlemma} allows us to deal with this problem by rescaling the smoothing distribution suitably. We define the rescaled smoothing distribution (which is not a probability distribution in general) as
\begin{align}\label{tmusmdefeq}
\tmusm(\v|\vct{Y}_{0:k})&:=\f{\det(\mtx{A}_k)^{1/2}}{(2\pi)^{d/2}\cdot \sigma_Z^d}\cdot \exp\l[-\f{\vct{B}_k\mtx{A}_k^{-1}\vct{B}_k }{2\sigma_Z^2}\r]\cdot \f{q(\v)}{q(\u)}\cdot \exp\l[-\f{\lsm(\v)}{2\sigma_Z^2}\r],
\end{align}
which is of similar form as the Gaussian approximation 
\begin{align}
\musmG(\v|\vct{Y}_{0:k})&=\f{\det(\mtx{A}_k)^{1/2}}{(2\pi)^{d/2}\cdot \sigma_Z^d}\cdot \exp\l[-\f{\vct{B}_k\mtx{A}_k^{-1}\vct{B}_k }{2\sigma_Z^2}\r]\cdot \exp\l[-\f{\lsmG(\v)}{2\sigma_Z^2}\r].
\end{align}
Let 
\begin{equation}\label{rhodefeq}
\rho(h,\sigma_Z):=\f{(h\sigma_Z^2)^{1/3}}{2C_2(\u,T)},
\end{equation}
then based on the assumption that $\sigma_Z\sqrt{h}\le C_{\mr{TV}}(\u,T,\epsilon)^{-1}$, \eqref{CTVb2} and \eqref{CTVb3},  it follows that
\begin{equation}\label{rhoineq}
\rho(h,\sigma_Z)\le \min\l[\l(\f{h\sigma_Z^2}{C_2(\u,T)+C_3(\u,T,\f{\epsilon}{4})\sigma_Z\sqrt{h}}\r)^{1/3},\f{1}{2C_{q}^{(1)}},R-\|\u\|\r].
\end{equation}
Let $B_{\rho}:=\{\v\in \R^d: \|\v-\u\|\le \rho(h,\sigma_Z)\}$, and denote by $B_\rho^c$ its complement in $\R^d$. Then by Lemma \ref{dtvboundlemma}, we have
\begin{align}
\label{dtvboundsecondstepeq}&\dtv\l(\musm(\cdot|\vct{Y}_{0:k}),\musmG(\cdot|\vct{Y}_{0:k})\r)
\le \int_{\v\in \R^d} \l|\tmusm(\v|\vct{Y}_{0:k})-\musmG(\v|\vct{Y}_{0:k})\r|d\v\\
\label{dtvbound3partseq}&\le \tmusm(B_\rho^c|\vct{Y}_{0:k})
+\musmG(B_\rho^c|\vct{Y}_{0:k})+\int_{\v\in B_{\rho}} \l|\tmusm(\v|\vct{Y}_{0:k})-\musmG(\v|\vct{Y}_{0:k})\r|d\v.
\end{align}

By Lemma \ref{dWboundlemma}, we can bound the Wasserstein distance of $\musm(\cdot|\vct{Y}_{0:k})$ and $\musmG(\cdot|\vct{Y}_{0:k})$ as
\begin{align}
\nonumber&\dW(\musm(\cdot|\vct{Y}_{0:k}),\musmG(\cdot|\vct{Y}_{0:k}))\le \int_{\v\in \R^d}\|\v-\u\| \cdot |\musm(\v |\vct{Y}_{0:k})-\musmG(\v |\vct{Y}_{0:k})|d\v\\
\label{dWbound3partseq}&\le \int_{\v\in B_{\rho}}\|\v-\u\| \cdot \l|\musm(\v |\vct{Y}_{0:k})-\musmG(\v |\vct{Y}_{0:k})\r|d\v  \\
\nonumber&+ \int_{\v\in B_\rho^c}\|\v-\u\| \musm(\v|\vct{Y}_{0:k}) d\v+\int_{\v\in B_\rho^c}\|\v-\u\| \musmG(\v |\vct{Y}_{0:k})d\v.
\end{align}
In the following six lemmas, we bound the three terms in inequalities \eqref{dtvbound3partseq} and \eqref{dWbound3partseq}. 

\begin{lem}\label{gaussapproxsmTVlem1}
Using the notations and assumptions of this section, we have
\begin{align*}
&\int_{\v\in B_{\rho}} \l|\tmusm(\v|\vct{Y}_{0:k})-\musmG(\v|\vct{Y}_{0:k})\r| d\v\le D_1(\u,T,\epsilon)\sigma_Z\sqrt{h}+D_2(\u,T,\epsilon)\sigma_Z^2 h\text{ for}\\
\nonumber&D_1(\u,T,\epsilon):= \f{2 C_{\vct{B}}\l(\f{\epsilon}{4}\r)}{c(\u,T)}+ \sqrt{\f{2d}{c(\u,T)}}+ C_2(\u,T)\l(6 \l(\f{2d}{c(\u,T)}\r)^{3/2}+2\l(\f{2C_{\mtx{B}}(\f{\epsilon}{4})}{c(\u,T)}\r)^{3}\r),\text{ and}\\
\nonumber&D_2(\u,T,\epsilon):= C_3\l(\u,T,\f{\epsilon}{4}\r) \l(6 \l(\f{2d}{c(\u,T)}\r)^{3/2}+2\l(\f{2C_{\mtx{B}}(\f{\epsilon}{4})}{c(\u,T)}\r)^{3}\r).\end{align*}
\end{lem}
\begin{proof}
Note that by \eqref{rhoineq}, we know that $B_{\rho}\subset \BR$, and 
\begin{align*}&\int_{\v\in B_{\rho}} \l|\tmusm(\v|\vct{Y}_{0:k})-\musmG(\v|\vct{Y}_{0:k})\r|d\v=\int_{\v\in B_{\rho}}\musmG(\v|\vct{Y}_{0:k}) \l|1-\f{\tmusm(\v|\vct{Y}_{0:k})}{\musmG(\v|\vct{Y}_{0:k}) }\r|d\v\\
&=\int_{\v\in B_{\rho}}\musmG(\v|\vct{Y}_{0:k}) \l|1-\exp\l(\log\l(\f{q(\v)}{q(\u)}\r)-\f{(\lsm(\v)-\lsmG(\v))}{2\sigma_Z^2}\r)\r|d\v.
\end{align*}
Now using \eqref{rhoineq}, we can see that  $\sup_{\v\in B_{\rho}} \l|\log\l(\f{q(\v)}{q(\u)}\r)\r|\le C_{q}^{(1)}\cdot \f{1}{2C_{q}^{(1)}}\le \f{1}{2}$, and 
using \eqref{lsmlsmGdiffpropeq}, we have $\l|\f{(\lsm(\v)-\lsmG(\v))}{2\sigma_Z^2}\r|\le \f{1}{2}$. Using the fact that $|1-\exp(x)|\le 2|x|$ for $-1\le x\le 1$, and the bounds $|\log(q(\v)/q(\u))|\le C_{q}^{(1)}\|\v-\u\|$ and \eqref{lsmlsmGdiffpropeq}, we can see that for every $\v\in B_{\rho}$,
\begin{equation}\label{tmusmmusmGdiffbndeq}
\l|1-\f{\tmusm(\v|\vct{Y}_{0:k})}{\musmG(\v|\vct{Y}_{0:k})}\r|\le 2\l(C_{q}^{(1)}\|\v-\u\|+\f{C_2(\u,T)+C_3\l(\u,T,\f{\epsilon}{4}\r)\sigma_Z\sqrt{h}}{h}\cdot \f{\|\v-\u\|^3}{2\sigma_Z^2}\r),
\end{equation}
therefore
\begin{align*}&\int_{\v\in B_{\rho}} \l|\tmusm(\v|\vct{Y}_{0:k})-\musmG(\v|\vct{Y}_{0:k})\r|d\v\\
&\le 2\int_{\v\in \R^d}\musmG(\v|\vct{Y}_{0:k}) \l(C_{q}^{(1)}\|\v-\u\|+\f{C_2(\u,T)+C_3\l(\u,T,\f{\epsilon}{4}\r)\sigma_Z\sqrt{h}}{h}\cdot \f{\|\v-\u\|^3}{2\sigma_Z^2}\r).
\end{align*}
Let $\vZ$ denote a $d$-dimensional standard normal random vector, then it is easy to see that $\E(\|\vZ\|)\le (\E(\|\vZ\|^2))^{1/2}\le d^{1/2}$, and $\E(\|\vZ\|^3)\le 
(\E(\|\vZ\|^4))^{3/4}\le (3d^2)^{3/4}\le 3d^{3/2}$.
Since we have assumed that the events in \eqref{Aknormboundeq1} and \eqref{Bknormboundeq1} hold, we know that 
\begin{equation}\label{AkAkBkineq}\|\mtx{A}_k^{-1/2}\|\le \sqrt{\f{2h}{c(\u,T)}},\text { and }\|\mtx{A}_k^{-1}\vct{B}_k\|\le 
\f{2 C_{\vct{B}}\l(\f{\epsilon}{4}\r)}{c(\u,T)} \cdot \sigma_Z\sqrt{h}.
\end{equation}
Finally, it is not difficult to show that for any $a,b\ge 0$, $(a+b)^3\le 4(a^3+b^3)$. Therefore 
\begin{align}
\nonumber&\int_{\v\in B_{\rho}} \l|\tmusm(\v|\vct{Y}_{0:k})-\musmG(\v|\vct{Y}_{0:k})\r| d\v \\
\nonumber&\le \E\l[\l. C_{q}^{(1)}\|\vct{W}-\u\|+\f{C_2(\u,T)+C_3\l(\u,T,\f{\epsilon}{4}\r)\sigma_Z\sqrt{h}}{h}\cdot \f{\|\vct{W}-\u\|^3}{2\sigma_Z^2}\r|\mtx{A}_k,\vct{B}_k \r]\\
\nonumber&\le C_{q}^{(1)}\l(\|\mtx{A}_k^{-1}\vct{B}_k\|+\sigma_Z \|\mtx{A}_k^{-1/2}\|\sqrt{d}\r)\\
\nonumber&+\f{C_2(\u,T)+C_3\l(\u,T,\f{\epsilon}{4}\r)\sigma_Z\sqrt{h}}{2\sigma_Z^2 h} \cdot 4\cdot \l( 3\|\mtx{A}_k^{-1/2}\|^3 \sigma_Z^3 d^{3/2}+ \|\mtx{A}_k^{-1}\vct{B}_k\|^3\r)\\
\nonumber&\le \l(\f{2 C_{\vct{B}}\l(\f{\epsilon}{4}\r)}{c(\u,T)}+ \sqrt{\f{2d}{c(\u,T)}}\r)\sigma_Z\sqrt{h}\\
\nonumber&+\f{C_2(\u,T)+C_3\l(\u,T,\f{\epsilon}{4}\r)\sigma_Z\sqrt{h}}{2\sigma_Z^2 h} \cdot 4\cdot \l(3\l(\f{2h}{c(\u,T)}\r)^{3/2}\cdot \sigma_Z^3 d^{3/2}+\l(\f{2 C_{\vct{B}}\l(\f{\epsilon}{4}\r)}{c(\u,T)} \cdot \sigma_Z\sqrt{h}\r)^3\r),
\end{align}
thus the result follows.
\end{proof}

\begin{lem}\label{gaussapproxsmTVlem2}
Using the notations and assumptions of this section, we have
\begin{align*}
\musmG(B_\rho^c|\vct{Y}_{0:k})\le  (d+1)\exp\l(-\f{c(\u,T)\cdot (\sigma_Z\sqrt{h})^{-2/3}}{64d C_2(\u,T)^2}\r).
\end{align*}
\end{lem}
\begin{proof}
Note that if $\vZ$ is a $d$ dimensional standard normal random vector, then by Theorem 4.1.1 of \cite{tropp2015introduction}, we have
\begin{equation}\label{eqnormZgetbound1}
\PP(\|\vZ\|\ge t)\le (d+1)\exp\l(-\f{t^2}{2d}\r)\text { for any }t\ge 0.
\end{equation}
Since the random variable $\vct{W}$ defined in \eqref{Wdefeq} is distributed as $\musmG(\cdot |\vct{Y}_{0:k})$ when conditioned on $\mtx{A}_k, \vct{B}_k$, we have
\begin{align*}
&\musmG(B_\rho^c|\vct{Y}_{0:k})=\PP\l(\|\vct{W}-\u\|> \rho(h,\sigma_Z)|\mtx{A}_k, \vct{B}_k\r)\\
&\le \PP\l(\|\vZ\|> \f{\rho(h,\sigma_Z)-\|\mtx{A}_k^{-1}\vct{B}_k\|}{\sigma_Z \cdot \|\mtx{A}_k^{-1/2}\|}\r)\le \PP\l(\|\vZ\|> \f{\rho(h,\sigma_Z)-\f{2 C_{\vct{B}}\l(\f{\epsilon}{4}\r)}{c(\u,T)} \cdot \sigma_Z\sqrt{h}}{\sigma_Z \cdot \sqrt{\f{2h}{c(\u,T)}}}\r).
\end{align*}
Based on the assumption that $\sigma_Z\sqrt{h}\le C_{\mr{TV}}(\u,T,\epsilon)^{-1}$, and \eqref{CTVb4}, we have $\f{2 C_{\vct{B}}\l(\f{\epsilon}{4}\r)}{c(\u,T)} \cdot \sigma_Z\sqrt{h}\le \f{\rho(h,\sigma_Z)}{2}$, and 
\begin{align}\nonumber\musmG(B_\rho^c|\vct{Y}_{0:k})&\le \PP\l(\|\vZ\|\ge \f{(\sigma_Z\sqrt{h})^{-1/3} \cdot \sqrt{c(\u,T)}}{4\sqrt{2}C_2(\u,T)}\r),
\end{align}
and the result follows by \eqref{eqnormZgetbound1}.
\end{proof}

\begin{lem}\label{gaussapproxsmTVlem3}
Using the notations and assumptions of this section, we have
\begin{align*}
&\tmusm(B_\rho^c|\vct{Y}_{0:k})\le D_3(\u,T)\cdot \exp\l(-\f{(\sigma_Z\sqrt{h})^{-2/3}}{D_4(\u,T)}\r),\text{ with}\\
&D_3(\u,T):= C_{\|\mtx{A}\|}^{d/2} \cdot \f{\sqrt{2} \sup_{\v\in \BR}q(\v) }{\sqrt{c(\u,T)}}\cdot (d+1)\text{ and }\\
&D_4(\u,T):=\f{16d \cdot (C_2(\u,T))^2}{c(\u,T)}.
\end{align*}
\end{lem}
\begin{proof}
Let $q_{\max}:=\sup_{\v\in \BR}q(\v)$. By our assumption that the event in \eqref{lsmlowerboundpropeq} holds,  we have
\begin{align*}&\lsm(\v)\ge \f{c(\u,T)}{h} \|\v-\u\|^2- \f{C_1\l(\u,T,\f{\epsilon}{4}\r)\sigma_Z}{\sqrt{h}}\cdot \|\v-\u\|\text{ for every }\v\in \BR, \text{ and thus}\\
&\tmusm(B_\rho^c|\vct{Y}_{0:k})\le 
\f{\det(\mtx{A}_k)^{1/2}}{(2\pi)^{d/2}\cdot \sigma_Z^d}\exp\l[-\f{\vct{B}_k\mtx{A}_k^{-1}\vct{B}_k } {2\sigma_Z^2}\r]\cdot \\
&\int_{\v\in B_\rho^c}\f{q(\v)}{q(\u)}\cdot \exp\l[-\f{\l(\f{c(\u,T)}{h} \|\v-\u\|^2- \f{C_1\l(\u,T,\f{\epsilon}{4}\r)\sigma_Z}{\sqrt{h}}\cdot \|\v-\u\|\r)}{2\sigma_Z^2}\r]d \v
\\
&\le C_{\|\mtx{A}\|}^{d/2} \cdot \f{q_{\max}}{(2\pi)^{d/2}\cdot (\sigma_Z \sqrt{h})^d}\\
&\cdot\int_{\v\in B_\rho^c} \exp\l[-\f{\l(c(\u,T)\|\v-\u\|^2- C_1\l(\u,T,\f{\epsilon}{4}\r)\sigma_Z\sqrt{h}\cdot \|\v-\u\|\r)}{2\sigma_Z^2 h}\r]d \v,
\end{align*}
where in the last step we have used the fact that $\det(\mtx{A}_k)^{1/2}\le \|\mtx{A}_k\|^{d/2}$. 
Based on the assumption that $\sigma_Z\sqrt{h}\le C_{\mr{TV}}(\u,T,\epsilon)^{-1}$, and \eqref{CTVb5}, we have
\[\f{1}{2}c(\u,T)\|\v-\u\|^2\ge - C_1\l(\u,T,\f{\epsilon}{4}\r)\sigma_Z\sqrt{h}\cdot \|\v-\u\|\text{ for }\|\v-\u\|\ge \rho(h,\sigma_Z),\]
therefore by \eqref{eqnormZgetbound1}, we have
\begin{align}
\nonumber&\tmusm(B_\rho^c|\vct{Y}_{0:k})\le C_{\|\mtx{A}\|}^{d/2} \cdot \f{q_{\max}}{(2\pi)^{d/2}\cdot (\sigma_Z \sqrt{h})^d}\cdot\int_{\v\in B_\rho^c} \exp\l[-\f{c(\u,T)\|\v-\u\|^2 }{4\sigma_Z^2 h}\r]d \v\\
\nonumber&\le C_{\|\mtx{A}\|}^{d/2} \cdot \f{q_{\max} \sqrt{2}}{\sqrt{c(\u,T)}}\cdot \PP\l(\|\vct{Z}\|\ge \rho(h,\sigma_Z)\cdot\f{\sqrt{c(u,T)}}{\sqrt{2}\sigma_Z\sqrt{h}}\r),
\end{align}
and the claim of the lemma follows by \eqref{eqnormZgetbound1}.
\end{proof}

From inequality \eqref{dtvbound3partseq} and Lemmas \ref{gaussapproxsmTVlem1}, \ref{gaussapproxsmTVlem2} and \ref{gaussapproxsmTVlem3}, it follows that under the assumptions of this section, we can  set $C_{\mr{TV}}^{(1)}(\u,T)$ and $C_{\mr{TV}}^{(2)}(\u,T)$ sufficiently large such that we have
\begin{equation}\label{dtvboundeqsm1}
\dtv\l(\musm(\cdot|\vct{Y}_{0:k}),\musmG(\cdot|\vct{Y}_{0:k})\r)\le 
\int_{\v\in \R^d} \l|\tmusm(\v|\vct{Y}_{0:k})-\musmG(\v|\vct{Y}_{0:k})\r|d\v\le
C_{\mr{TV}}(\u,T,\epsilon)\sigma_Z\sqrt{h}.
\end{equation}

Now we bound the three terms needed for the Wasserstein distance.
\begin{lem}\label{gaussapproxsmWlem1}
Using the notations and assumptions of this section, we have
\begin{align*}
&\int_{\v\in B_{\rho}}\|\v-\u\| \cdot \l|\musm(\v |\vct{Y}_{0:k})-\musmG(\v |\vct{Y}_{0:k})\r|d\v\le \sigma_Z^2 h \cdot \l(C_{1}^*(\u,T)+C_2^*(\u,T)\l(\log\l(\f{1}{\epsilon}\r)\r)^{2}\r),
\end{align*}
for some finite positive constants $C_{1}^*(\u,T)$, $C_{2}^*(\u,T)$.
\end{lem}
\begin{proof}
Note that for any $\v\in \B_{\rho}$, we have
\begin{align*}\l|\f{\musm(\v|\vct{Y}_{0:k})}{\musmG(\v|\vct{Y}_{0:k})}-1\r|&=\l|\f{\musm(\v|\vct{Y}_{0:k})}{\tmusm(\v|\vct{Y}_{0:k})}\cdot \f{\tmusm(\v|\vct{Y}_{0:k})}{\musmG(\v|\vct{Y}_{0:k})}-1\r|\\
&\le 
\l|\f{\tmusm(\v|\vct{Y}_{0:k})}{\musmG(\v|\vct{Y}_{0:k})}-1\r|+\f{\tmusm(\v|\vct{Y}_{0:k})}{\musmG(\v|\vct{Y}_{0:k})}\cdot \l|\f{\musm(\v|\vct{Y}_{0:k})}{\tmusm(\v|\vct{Y}_{0:k})}-1\r|.
\end{align*}
From \eqref{dtvboundeqsm1}, and the fact that $\tmusm(\cdot|\vct{Y}_{0:k})$ is a rescaled version of  $\tmusm(\cdot|\vct{Y}_{0:k})$, it follows that for $\sigma_Z\sqrt{h}\le \f{1}{2}C_{\mr{TV}}(\u,T,\epsilon)^{-1}$, we have
\begin{equation}\label{Cksmbndeq1}
\l|\f{\musm(\v|\vct{Y}_{0:k})}{\tmusm(\v|\vct{Y}_{0:k})}-1\r|\le 2 C_{\mr{TV}}(\u,T,\epsilon) \sigma_Z\sqrt{h}.
\end{equation}
By \eqref{tmusmmusmGdiffbndeq} and \eqref{rhoineq}, it follows that  $\f{\tmusm(\v|\vct{Y}_{0:k})}{\musmG(\v|\vct{Y}_{0:k})}\le 3$ for every $\v\in \B_{\rho}$. By using these and bounding $\l|\f{\tmusm(\v|\vct{Y}_{0:k})}{\musmG(\v|\vct{Y}_{0:k})}-1\r|$ via \eqref{tmusmmusmGdiffbndeq}, we obtain that
\begin{align*}
&\int_{\v\in B_{\rho}}\|\v-\u\| \cdot \l|\musm(\v |\vct{Y}_{0:k})-\musmG(\v |\vct{Y}_{0:k})\r|d\v\\
&= \int_{\v\in B_{\rho}}\musmG(\v |\vct{Y}_{0:k})\|\v-\u\| \cdot \l|\f{\musm(\v |\vct{Y}_{0:k})}{\musmG(\v |\vct{Y}_{0:k})}-1\r|d\v\\
&\le \int_{\v\in \R^d}\musmG(\v |\vct{Y}_{0:k})\|\v-\u\| \cdot \Bigg(2\bigg(C_{q}^{(1)}\|\v-\u\|+\f{C_2(\u,T)+C_3\l(\u,T,\f{\epsilon}{4}\r)\sigma_Z\sqrt{h}}{h}\cdot \f{\|\v-\u\|^3}{2\sigma_Z^2}\bigg)\\
&+6C_{\mr{TV}}(\u,T,\epsilon) \sigma_Z\sqrt{h}\Bigg) d\v=\E\Bigg(\|\vct{W}-\u\| \cdot 6C_{\mr{TV}}(\u,T,\epsilon) \sigma_Z\sqrt{h} +\|\vct{W}-\u\|^2 \cdot 2C_{q}^{(1)} \\
& +\|\vct{W}-\u\|^4\cdot \f{C_2(\u,T)+C_3\l(\u,T,\f{\epsilon}{4}\r)\sigma_Z\sqrt{h}}{\sigma_Z^2 h} \Bigg|\mtx{A}_k, \vct{B}_k\Bigg)
\end{align*}
Based on \eqref{AkAkBkineq}, we have 
\[\|\vct{W}-\u\|=\|\sigma_Z \cdot \mtx{A}_k^{-1/2} \cdot \vZ-\mtx{A}_k^{-1}\vct{B}_k\|\le \sigma_Z \sqrt{h}\l( \sqrt{\f{2}{c(\u,T)}}\|\vZ\| + \f{2 C_{\vct{B}}\l(\f{\epsilon}{4}\r)}{c(\u,T)} \cdot \r),\]
where $\vct{Z}$ is a $d$ dimensional standard normal random vector. The claimed result now follows using the fact that $\E(\|\vZ\|)\le \sqrt{d}$, $\E(\|\vZ\|^2)\le d$, and $\E(\|\vZ\|^4)\le 3d^2$.
\end{proof}

\begin{lem}\label{gaussapproxsmWlem2}
Using the notations and assumptions of this section, we have
\[\int_{\v\in B_\rho^c}\|\v-\u\| \musm(\v|\vct{Y}_{0:k}) d\v\le 4RD_3(\u,T)\cdot \exp\l(-\f{(\sigma_Z\sqrt{h})^{-2/3}}{D_4(\u,T)}\r).\]
\end{lem}
\begin{proof}
Using \eqref{Cksmbndeq1}, Lemma \ref{gaussapproxsmTVlem3}, and the fact that $\sigma_Z\sqrt{h}\le \f{1}{2}C_{\mr{TV}}(\u,T,\epsilon)^{-1}$, we have
\[\musm(B_\rho^c|\vct{Y}_{0:k})\le 2\tmusm(B_\rho^c|\vct{Y}_{0:k}) \le 2D_3(\u,T)\cdot \exp\l(-\f{(\sigma_Z\sqrt{h})^{-2/3}}{D_4(\u,T)}\r),\]
and the result follows from  $\int_{\v\in B_\rho^c}\|\v-\u\| \musm(\v|\vct{Y}_{0:k}) d\v\le 2R\cdot \musm(B_\rho^c|\vct{Y}_{0:k})$.
\end{proof}

\begin{lem}\label{gaussapproxsmWlem3}
Using the notations and assumptions of this section, we have
\[\int_{\v\in B_{\rho}^c}\|\v-\u\| \musmG(\v |\vct{Y}_{0:k})d\v\le C_3^*(\u,T)\exp\l(-C_4^*(\u,T)\cdot (\sigma_Z\sqrt{h})^{-2/3}\r),\]
for some finite positive constants $C_{3}^*(\u,T)$, $C_{4}^*(\u,T)$.
\end{lem}
\begin{proof}
Let $\vct{W}$ be defined as in \eqref{Wdefeq}, and $\vct{Z}$ be $d$-dimensional standard normal. Using the fact that for a non-negative valued random variable $X$, we have $\E(X)=\int_{t=0}^{\infty}\PP(X\ge t)dt$, it follows that
\begin{align*}
&\int_{\v\in B_{\rho}^c}\|\v-\u\| \musmG(\v |\vct{Y}_{0:k})d\v=\E\l(\l.\|\vct{W}-\u\| 1_{[\|\vct{W}-\u\|\ge \rho(\sigma_Z,h)]}\r|\mtx{A}_k, \vct{B}_k\r)\\
&=(\rho(\sigma_Z,h))\musmG( B_{\rho}^c)+\int_{t=\rho(\sigma_Z,h)}^{\infty} \PP\l(\|\vct{W}-\u\|\ge t\r) dt\\
&\le (\rho(\sigma_Z,h))\musmG( B_{\rho}^c)+\int_{t=\rho(\sigma_Z,h)}^{\infty} \PP\l(\|\vct{Z}\|\ge \f{t}{\sigma_Z \sqrt{h}} \cdot \sqrt{\f{c(\u,T)}{2}}\r) dt
\end{align*}
\begin{align*}
&\le (\rho(\sigma_Z,h))\musmG( B_{\rho}^c)+(d+1)\int_{t=\rho(\sigma_Z,h)}^{\infty} \exp\l(-\f{t^2}{4d \sigma_Z^2 h/c(\u,T)}\r) dt\\
&\le (\rho(\sigma_Z,h)) (d+1)\exp\l(-\f{c(\u,T)\cdot (\sigma_Z\sqrt{h})^{-2/3}}{64d C_2(\u,T)^2}\r) \\
&+ (d+1)\sqrt{\f{4\pi d \sigma_Z^2 h}{c(\u,T)}}\cdot \exp\l(-\f{\rho(\sigma_Z,h)^2}{4d \sigma_Z^2 h/c(\u,T)}\r),
\end{align*}
and the claim of the lemma follows (we have used Lemma \ref{gaussapproxsmTVlem2} in the last step).
\end{proof}

From inequality \eqref{dWbound3partseq} and Lemmas \ref{gaussapproxsmTVlem1}, \ref{gaussapproxsmTVlem2} and \ref{gaussapproxsmTVlem3}, it follows that under the assumptions of this section, for some appropriate choice of constants $C_{\mr{W}}^{(1)}(\u,T)$ and $C_{\mr{W}}^{(2)}(\u,T)$, we have
\begin{equation}\label{dWboundeqsm1}
\dW\l(\musm(\cdot|\vct{Y}_{0:k}),\musmG(\cdot|\vct{Y}_{0:k})\r)\le C_{\mr{W}}(\u,T,\epsilon)\sigma_Z^2h.
\end{equation}

\begin{proof}[Proof of Theorem \ref{Gaussianapproxthmsm}]
The claim of the theorem follows from inequalities \eqref{dtvboundeqsm1} and \eqref{dWboundeqsm1}, and the fact that the assumption that all four of the events in the equations \eqref{Aknormboundeq1},\eqref{Bknormboundeq1},\eqref{lsmlowerboundpropeq}, and \eqref{lsmlsmGdiffpropeq} hold happens with probability at least $1-\epsilon$.
\end{proof}

\subsubsection{Gaussian approximation for the filter}
In this section, we are going to describe the proof of Theorem \ref{Gaussianapproxthmfi}.
We start by some notation. We define the restriction of $\musmG(\cdot |\Yok)$ to $\BR$, denoted by $\musmGBR(\cdot |\Yok)$ as
\begin{equation}\label{musmGBRdefeq}\musmGBR(S|\Yok)=\frac{\musmG(S\cap \BR|\Yok)}{\musmG(\BR|\Yok)} \text{ for any Borel-measurable }S\subset \R^d.
\end{equation}
This is a probability distribution which is supported on $\BR$. We denote its push-forward map by $\Psi_T$ as $\etafiG(\cdot|\Yok)$, i.e. if a random vector $\vct{X}$ is distributed as $\musmGBR(\cdot |\Yok)$, then $\etafiG(\cdot|\Yok)$ denotes the distribution of $ \Psi_T(\vct{X})$.

The proof uses a coupling argument stated in the next two lemmas that allows us to deduce the results based on the Gaussian approximation of the smoother (Theorem \ref{Gaussianapproxthmsm}).

\begin{lem}[Coupling argument for total variation distance bound]\label{lemcouplingTV}
The total variation distance of the filtering distribution and its Gaussian approximation can be bounded as follows,
\[
\dtv(\mufi(\cdot|\Yok),\mufiG(\cdot|\Yok))\le \dtv(\musm(\cdot|\Yok),\musmG(\cdot|\Yok))+\dtv(\etafiG(\cdot|\Yok), \mufiG(\cdot|\Yok)).\]
\end{lem}
\begin{proof}
First, notice that by Proposition 3(f) of \cite{Robertsgeneral}, we have
\begin{align*}
&\dtv(\musm(\cdot|\Yok),\musmGBR(\cdot|\Yok))=\int_{\v\in \BR} (\musm(\v|\Yok)-\musmGBR(\cdot|\Yok))_+\\
&\le  \int_{\v\in \R^d} (\musm(\v|\Yok)-\musmG(\cdot|\Yok))_+=\dtv(\musm(\cdot|\Yok),\musmG(\cdot|\Yok)).
\end{align*}
By Proposition 3(g) of \cite{Robertsgeneral}, there is a coupling $(\vct{X}_1, \vct{X}_2)$ of random vectors such that $\vct{X}_1\sim \musm(\cdot|\Yok)$, $\vct{X}_2\sim \musmGBR(\cdot|\Yok)$, and $\PP(\vct{X}_1\ne \vct{X}_2|\Yok)=\dtv(\musm(\cdot|\Yok),\musmGBR(\cdot|\Yok))$. Given this coupling, we look at the coupling of the transformed random variables $(\Psi_T(\vct{X}_1),\Psi_T(\vct{X}_2))$. This obviously satisfies that $\PP(\Psi_T(\vct{X}_1)\ne \Psi_T(\vct{X}_2)|\Yok)\le \dtv(\musm(\cdot|\Yok),\musmGBR(\cdot|\Yok))$. Moreover, we have $\Psi_T(\vct{X}_1)\sim \mufi(\cdot|\Yok)$ and $\Psi_T(\vct{X}_2)\sim \etafiG(\cdot|\Yok)$, thus
\[\dtv(\mufi(\cdot|\Yok),\etafiG(\cdot|\Yok))\le \dtv(\musm(\cdot|\Yok),\musmGBR(\cdot|\Yok))\le \dtv(\musm(\cdot|\Yok),\musmG(\cdot|\Yok)).\]
The statement of the lemma now follows by the triangle inequality.
\end{proof}

\begin{lem}[Coupling argument for Wasserstein distance bound]\label{lemcouplingW}
The Wasserstein distance of the filtering distribution and its Gaussian approximation can be bounded as follows,
\begin{align}
\label{dWcouplingineq}&\dW(\mufi(\cdot|\Yok),\mufiG(\cdot|\Yok))\\
\nonumber&\le \exp(GT)\cdot \l[\dW(\musm(\cdot|\Yok),\musmG(\cdot|\Yok))+2R \musmG(\BR^c|\Yok))\r]
+\dW(\etafiG(\cdot|\Yok), \mufiG(\cdot|\Yok)).
\end{align}
\end{lem}
\begin{proof}
By Theorem 4.1 of \cite{Villanioptimaloldnew}, there exists a coupling of random variables $(\vct{X}_1,\vct{X}_2)$ (called the optimal coupling) such that $\vct{X}_1\sim \musm(\cdot|\Yok)$, $\vct{X}_2\sim  \musmG(\cdot|\Yok)$, and 
\[\E(\|\vct{X}_1-\vct{X}_2\| |\Yok)=\dW(\musm(\cdot|\Yok),\musmG(\cdot|\Yok)).\]
Let $\hat{\vct{X}}_2:=\vct{X}_2\cdot 1_{[\|\vct{X}_2\|\le R]}+\frac{\vct{X}_2}{\|\vct{X}_2\|}\cdot R\cdot 1_{[\|\vct{X}_2\|>R]}$ denote the projection of $\vct{X}_2$ on the ball $\BR$. Then using the fact that $\vct{X}_1\in \BR$, it is easy to see that $\|\hat{\vct{X}}_2-\vct{X}_1\| \le \|\vct{X}_2-\vct{X}_1\|$, and therefore 
\[\dW(\LL(\hat{\vct{X}}_2|\Yok), \musm(\cdot|\Yok) )\le \dW(\musm(\cdot|\Yok),\musmG(\cdot|\Yok)),\]
where $\LL(\hat{\vct{X}}_2|\Yok)$ denotes the distribution of $\hat{\vct{X}}_2$ conditioned on $\Yok$.

Moreover, by the definitions, for a given $\Y_{0:k}$, it is easy to see we can couple random variables $\hat{\vct{X}}_2$ and $\tilde{\vct{X}}_2\sim \musmGBR(\cdot|\Yok)$ such that they are the same with probability at least $1-\musmG(\BR^c)$. Since the maximum distance between two points in $\BR$ is at most $2R$, it follows that
\[\dW(\LL(\hat{\vct{X}}_2|\Yok),\musmGBR(\cdot|\Yok))\le 2R \musmG(\BR^c).\]
By the triangle inequality, we obtain that
\[\dW(\musm(\cdot|\Yok), \musmGBR(\cdot|\Yok)))\le  \dW(\musm(\cdot|\Yok),\musmG(\cdot|\Yok))+2R \musmG(\BR^c), \]
and by \eqref{eqpathdistancebound}, it follows that
\[\dW(\mufi(\cdot|\Yok), \etafiG(\cdot|\Yok))\le  \exp(GT)\l[\dW(\musm(\cdot|\Yok),\musmG(\cdot|\Yok))+2R \musmG(\BR^c)\r].\]
The claim of the lemma now follows by the triangle inequality.
\end{proof}

As we can see, the above results still require us to bound the total variation and Wasserstein distances between the distributions $\etafiG(\cdot|\Yok)$ and  $\mufiG(\cdot|\Yok)$. Let 
\begin{equation}\label{Akfidef}\Akfi:=((\J\Psi_T(\uG))^{-1})'\cdot \mtx{A}_k\cdot (\J\Psi_T(\uG))^{-1},\end{equation}
then the density of $\mufiG(\cdot|\Yok)$ can be written as
\begin{equation}\label{eqmufiGdef}
\mufiG(\v|\Yok):=\frac{(\det(\mtx{A}_k))^{\frac{1}{2}} }{|\det (\J \Psi_T(\uG))|}\cdot \f{1}{(2\pi)^{d/2}\cdot \sigma_Z^d}\cdot \exp\l[-\frac{\l(\v-\Psi_T(\uG)\r)' \Akfi \l(\v-\Psi_T(\uG)\r) }{2\sigma_Z^2}\r].
\end{equation}
Since the normalising constant is not known for the case of $\etafiG(\cdot|\Yok)$, we define a rescaled version $\tetafiG(\cdot|\Yok)$ with density
\begin{equation}\label{eqtetafiGdef}
\tetafiG(\v|\Yok):=1_{[\v\in \Psi_T(\BR)]}  \cdot \f{(\det(\mtx{A}_k))^{\frac{1}{2}}}{(2\pi)^{d/2}\cdot \sigma_Z^d}\cdot \exp\l[-\frac{\l(\Psi_{-T}(\v)-\uG\r)' \mtx{A}_k \l(\Psi_{-T}(\v)-\uG\r) }{2\sigma_Z^2}\r].
\end{equation}

The following lemma bounds the difference between the logarithms of $\tetafiG(\v|\Yok)$ and $\mufiG(\v|\Yok)$.
\begin{lem}\label{lemtetafiGmufiGdiff}
For any $\v\in \Psi_T(\BR)$, we have
\begin{align*}
&|\log(\tetafiG(\v|\Yok))-\log(\mufiG(\v|\Yok))|\\
&\le \frac{M_2(T) \|\mtx{A}_k\|  \exp(4GT) \|\v-\Psi_T(\uG)\|^3}{2\sigma_Z^2}
+M_1(T)M_2(T) d \|\v-\Psi_T(\uG)\|.
\end{align*}
\end{lem}

\begin{proof}[Proof of Lemma \ref{lemtetafiGmufiGdiff}]
By \eqref{eqmufiGdef} and \eqref{eqtetafiGdef}, we have
\begin{align*}
&\log(\tetafiG(\v|\Yok))-\log(\mufiG(\v|\Yok))=\log |\det(\J \Psi_{-T}(\v))|-\log |\det(\J_{\Psi_{T}(\uG)}\Psi_{-T})|\\
&+\frac{1}{2\sigma_Z^2}\cdot \Big[(\Psi_{-T}(\v)-\uG)'\mtx{A}_k (\Psi_{-T}(\v)-\uG)\\
&\quad\quad\quad\quad-\l((\J\Psi_T(\uG))^{-1}(\v-\Psi_T(\uG))\r)' \mtx{A}_k (\J\Psi_T(\uG))^{-1}(\v-\Psi_T(\uG))\Big].
\end{align*}
The absolute value of the first difference can be bounded by Propositon \ref{proplogdetder} as
\begin{align*}
\l|\log |\det(\J \Psi_{-T}(\v))|-\log |\det(\J_{\Psi_{T}(\uG)}\Psi_{-T})|\r|
\le  \l\|\v-\Psi_{T}(\uG)\r\|\cdot M_1(T)M_2(T)d.
\end{align*}
For any two vectors $\x,\y\in \R^d$, we have
\[|\x' \mtx{A}_k \x - \y' \mtx{A}_k \y| =|\x' \mtx{A}_k \x -\x' \mtx{A}_k \y +\x' \mtx{A}_k \y -\y' \mtx{A}_k \y|\le \|\mtx{A}_k\| \|\x-\y\|(\|\x\|+\|\y\|),\]
so the second difference can be bounded as
\begin{align*}
&\big|(\Psi_{-T}(\v)-\uG)'\mtx{A}_k (\Psi_{-T}(\v)-\uG)\\
&-\l((\J\Psi_T(\uG))^{-1}(\v-\Psi_T(\uG))\r)' \mtx{A}_k (\J\Psi_T(\uG))^{-1}(\v-\Psi_T(\uG))\big|\\
&\le \|\mtx{A}_k\| \cdot \left\|\Psi_{-T}(\v)-\uG-(\J\Psi_T(\uG))^{-1}(\v-\Psi_T(\uG))\right\|\\
&\cdot \left(\l\|\Psi_{-T}(\v)-\uG\r\|+\l\|(\J\Psi_T(\uG))^{-1}(\v-\Psi_T(\uG))\r\|\r).
\end{align*}
Using \eqref{eqpathdistancebound}, we have
\[\l\|\Psi_{-T}(\v)-\uG\r\|+\l\|(\J\Psi_T(\uG))^{-1}(\v-\Psi_T(\uG))\r\|\le 2 \exp(GT) \|\v-\Psi_T(\uG)\|.\]
By \eqref{eqRkp1normbnd}, we have
\begin{align*}&\|\v-\Psi_T(\uG)- \J\Psi_T(\uG) (\Psi_{-T}(\v)-\uG)\|\\
&=\|\Psi_T(\Psi_{-T}(\v))-\Psi_T(\uG)- \J\Psi_T(\uG) (\Psi_{-T}(\v)-\uG)\|
\le \frac{1}{2}M_2(T) \|\Psi_{-T}(\v)-\uG\|^2,\end{align*}
so by \eqref{eqpathdistancebound}, it follows that
\begin{align*}
\l\|\Psi_{-T}(\v)-\uG-(\J\Psi_T(\uG))^{-1}(\v-\Psi_T(\uG))\r\|\le \frac{1}{2}M_2(T) \exp(3GT) \|\v-\Psi_{T}(\uG)\|^2.
\end{align*}
We obtain the claim of the lemma by combining the stated bounds.
\end{proof}

Now we are ready to prove our Gaussian approximation result for the filter.
\begin{proof}[Proof of Theorem \ref{Gaussianapproxthmfi}]
We suppose that $D^{(1)}_{\mr{TV}}(\u,T)\ge C^{(1)}_{\mr{TV}}(\u,T)$ and $D^{(2)}_{\mr{TV}}(\u,T)\ge C^{(2)}_{\mr{TV}}(\u,T)$, thus  $D_{\mr{TV}}(\u,T,\epsilon) \ge C_{\mr{TV}}(\u,T,\epsilon)$.
We also assume that $D^{(1)}_{\mr{TV}}(\u,T)$ satisfies that 
\begin{align}\label{DTVlowbndeq1}
D^{(1)}_{\mr{TV}}(\u,T)&\ge \frac{2^{\frac{5}{2}} d^{\frac{3}{2}} M_2(T) \exp(4GT)}{\sqrt{C_{\|\mtx{A}\|}}},\\
\label{DTVlowbndeq2}D^{(1)}_{\mr{TV}}(\u,T)&\ge \frac{2 \sqrt{M_2(T) \exp(4GT) C_{\|\mtx{A}\|}}}{(R-\|\u\|)^{\f{3}{2}}}.
\end{align}
Based on these assumptions on $D_{\mr{TV}}(\u,T,\epsilon)$, and the assumption that $\sigma_Z\sqrt{h}\le \frac{1}{2} D_{\mr{TV}}(\u,T,\epsilon)^{-1}$, it follows that the probability that all the four events in the equations \eqref{Aknormboundeq1},\eqref{Bknormboundeq1},\eqref{lsmlowerboundpropeq}, and \eqref{lsmlsmGdiffpropeq} hold is at least $1-\epsilon$. We are going to assume that this is the case for the rest of the proof. We define
\begin{equation}\label{rhopdefeq}
\rho'(\sigma_Z,h):=\frac{(4\sigma_Z^2 h)^{\frac{1}{3}}}{\l(M_2(T)\exp(4GT) C_{\|\mtx{A}\|}\r)^{\f{1}{3}}} \text{ and }B_{\rho'}:=\{\v\in \R^d: \|\v-\Psi_T(\u)\|\le \rho'(\sigma_Z,h)\}.
\end{equation}
Based on \eqref{DTVlowbndeq1}, \eqref{DTVlowbndeq2}, and the assumption that $\sigma_Z\sqrt{h}\le \frac{1}{2} D_{\mr{TV}}(\u,T,\epsilon)^{-1}$, it follows that
\begin{equation}\label{rhopineq}
\rho'(\sigma_Z,h)\le \min\left(\frac{(4\sigma_Z^2 h)^{\frac{1}{3}}}{\l(M_2(T)\exp(4GT) C_{\|\mtx{A}\|}\r)^{\f{1}{3}}}, \frac{1}{2d M_2(T) \exp(4GT)}, (R-\|\u\|)\exp(-GT)\right).
\end{equation}
The surface of a ball of radius $R-\|\u\|$ centered at $\u$ is contained in $\BR$, and it will be transformed by $\Psi_T$ to a closed continuous manifold whose points are at least $(R-\|\u\|)\exp(-GT)$ away from $\Psi_T(\u)$ (based on \eqref{eqpathdistancebound}). This implies that the ball of radius $(R-\|\u\|)\exp(-GT)$ centered at $\Psi_T(\u)$ is contained in $\Psi_T(\BR)$, and thus by \eqref{rhopineq}, $\B_{\rho'}\subset \Psi_T(\BR)$.

By Lemma \ref{dtvboundlemma}, we have
\begin{equation}\label{dtvbndfi3termseq}
\dtv\left(\etafiG(\cdot|\Yok),\mufiG(\cdot|\Yok)\right)\le \int_{\v\in \B_{\rho'}}|\tetafiG(\v|\Yok)-\mufiG(\v|\Yok)|d\v+\tetafiG(\B_{\rho'}^c|\Yok)+\mufiG(\B_{\rho'}^c|\Yok).
\end{equation}
By Lemma \ref{lemtetafiGmufiGdiff}, \eqref{rhopineq}, and the fact that $|\exp(x)-1|\le 2|x|$ for $x\in [-1,1]$, it follows that for $\v\in \B_{\rho'}$, we have
\[\l|\frac{\tetafiG(\v|\Yok)}{\mufiG(\v|\Yok)}-1\r|\le  \frac{M_2(T) C_{\|\mtx{A}\|}  \exp(4GT) \|\v-\Psi_T(\uG)\|^3}{\sigma_Z^2h}+2M_1(T)M_2(T) d \|\v-\Psi_T(\uG)\|.\]
Therefore the first term of \eqref{dtvbndfi3termseq} can be bounded as
\begin{align*}
&\int_{\v\in \B_{\rho'}}|\tetafiG(\v|\Yok)-\mufiG(\v|\Yok)|d\v= \int_{\v\in \B_{\rho'}}\mufiG(\v|\Yok)\l|\frac{\tetafiG(\v|\Yok)}{\mufiG(\v|\Yok)}-1\r| d\v\\
&\le \int_{\v\in \R^d}\mufiG(\v|\Yok) \Bigg(\frac{M_2(T) C_{\|\mtx{A}\|}  \exp(4GT) \|\v-\Psi_T(\uG)\|^3}{\sigma_Z^2h}\\
&\hspace{3.3cm}+2M_1(T)M_2(T) d\|\v-\Psi_T(\uG)\|\Bigg)d\v.
\end{align*}
This in turn can be bounded as in Lemma \ref{gaussapproxsmTVlem1}. The terms $\tetafiG(\B_{\rho'}^c|\Yok)$ and $\mufiG(\B_{\rho'}^c|\Yok)$ can be bounded in a similar way as in Lemmas \ref{gaussapproxsmTVlem2} and \ref{gaussapproxsmTVlem3}. Therefore by Lemma \ref{lemcouplingTV} we obtain that under the assumptions of this section, there are some finite constants $D^{(1)}_{\mr{TV}}(\u,T)$ and $D^{(2)}_{\mr{TV}}(\u,T)$ such that 
\begin{equation}
\dtv\l(\mufi(\cdot|\vct{Y}_{0:k}),\mufiG(\cdot|\vct{Y}_{0:k})\r)\le D_{\mr{TV}}(\u,T,\epsilon) \sigma_Z\sqrt{h}.
\end{equation}

For the Wasserstein distance bound, the proof is based on Lemma \ref{lemcouplingW}. Note that 
by the proof of Theorem \ref{Gaussianapproxthmsm}, under the assumptions on this section, we have  \[d_{\mr{W}}\l(\musm(\cdot|\vct{Y}_{0:k}),\musmG(\cdot|\vct{Y}_{0:k})\r)\le C_{\mr{W}}(\u,T,\epsilon) \sigma_Z^2 h.\]
Therefore we only need to bound the last two terms of \eqref{dWcouplingineq}. The fact that $\musmG(\BR^c|\Yok))=o(\sigma_Z^2 h)$ can be shown similarly to the proof of Lemma \ref{gaussapproxsmTVlem3}. Finally, the last term can be bounded by applying Lemma \ref{dWboundlemma} for $\y:=\Psi_T(\uG)$. This implies that
\begin{align*}
&\dW(\etafiG(\cdot|\Yok), \mufiG(\cdot|\Yok))\le \int_{\v\in \R^d} \l|\etafiG(\v|\Yok)-\mufiG(\v|\Yok) \r|\cdot \|\v-\Psi_T(\uG)\| d\v\\
& = \int_{\v\in \R^d} \mufiG(\v|\Yok) \l|\frac{\etafiG(\v|\Yok)}{\mufiG(\v|\Yok))}-1\r|\cdot \|\v-\Psi_T(\uG)\| d\v\\
&\le  \int_{\v\in B_{\rho'}} \mufiG(\v|\Yok) \l|\frac{\etafiG(\v|\Yok)}{\mufiG(\v|\Yok))}-1\r| \cdot \|\v-\Psi_T(\uG)\| 
 d\v\\
 &+\int_{\v\in B_{\rho'}^c} \etafiG(\v|\Yok)\cdot \|\v-\Psi_T(\uG)\| d\v+\int_{\v\in B_{\rho'}^c} \mufiG(\v|\Yok)\cdot \|\v-\Psi_T(\uG)\| d\v.
\end{align*}
These terms can be bounded in a similar way as in Lemmas \ref{gaussapproxsmWlem1}, \ref{gaussapproxsmWlem2} and \ref{gaussapproxsmWlem3}, and the claim of the theorem follows.
\end{proof}

\subsection{Comparison of mean square error of MAP and posterior mean}
In the following two subsections, we are going to prove our results concerning the mean square error of the MAP estimator for the smoother, and the filter, respectively.

\subsubsection{Comparison of MAP and posterior mean for the smoother}
In this section, we are going to prove Theorem \ref{GaussianMAPthm}. 
First, we introduce some notation.
Let $\uG:=\u-\mtx{A}_k^{-1}\vct{B}_k$ denote the center of the Gaussian approximation $\musmG$ (defined when $\mtx{A}_k$ is positive definite). Let $\rho(h,\sigma_Z)$ be as in \eqref{rhodefeq}, 
$B_{\rho}:=\{\v\in \R^d: \|\v-\u\|\le \rho(h,\sigma_Z)\}$, and $B_{\rho}^c$ be the complement of $B_{\rho}$. The proof is based on several lemmas which are described as follows. All of them implicitly assume that the assumptions of Theorem \ref{GaussianMAPthm} hold.
\begin{lem}[A bound on $\|\umean-\uG\|$]\label{umeanuGdifflemma}
There are some finite constants $D_5(\u,T)$ and $D_6(\u,T)$ such that for any $0<\epsilon\le 1$, for $\sigma_Z\sqrt{h}\le \f{1}{2}\cdot C_{\mr{TV}}(\u,T,\epsilon)^{-1}$, we have
\begin{align*}\PP\Bigg(&\f{c(\u,T)}{2h}\mtx{I}_d\prec \mtx{A}_k \prec \f{C_{\|\mtx{A}\|}}{h}\mtx{I}_d\\
&\text{and }\|\umean-\uG\|\le \l(D_5(\u,T)+D_6(\u,T)\l(\log\l(\f{1}{\epsilon}\r)\r)^{2}\r)\sigma_Z^2 h\Bigg|\u\Bigg)\ge 1-\epsilon.
\end{align*}
\end{lem}
\begin{proof}
This is a direct consequence of the Wasserstein distance bound of Theorem \ref{Gaussianapproxthmsm}, 
since
\begin{align*}\|\umean-\uG\|&=\l\|\int_{\x\in \R^d}\x \cdot \musm(\x |\vct{Y}_{0:k})d\x-\int_{\y\in \R^d}\y \cdot \musmG(\y |\vct{Y}_{0:k})d\y\r\|\\
&\le d_W(\musm(\cdot |\vct{Y}_{0:k}),\musmG(\cdot |\vct{Y}_{0:k})).\qedhere
\end{align*}
\end{proof}

\begin{lem}[A bound on $\|\uMAPsm-\u\|$]\label{lemmaUMAPudiff}
For any $0<\epsilon\le 1$, we have
\begin{equation}\label{uMAPudiffbndeq}
\Pcu{\|\uMAPsm-\u\|\le \f{C_1(\u,T,\epsilon)\sigma_Z\sqrt{h}+2C_q^{(1)}\sigma_Z^2h}{c(\u,T)}}\ge 1-\epsilon.
\end{equation}
\end{lem}
\begin{proof}
From Proposition \ref{lsmlowerboundprop}, and \eqref{musmlsmeq}, it follows that for any $0<\epsilon\le 1$,
\begin{align*}
\PP\Bigg(&\log\musm(\v|\vct{Y}_{0:k})-\log\musm(\u|\vct{Y}_{0:k})\le -\f{1}{2\sigma_Z^2h}\cdot \bigg(c(\u,T) \|\v-\u\|^2\\
&-\l(C_1(\u,T,\epsilon)\sigma_Z\sqrt{h}+2C_q^{(1)}\sigma_Z^2h\r) \|\v-\u\|\bigg)\text{ for every }\v\in \BR\Bigg)\ge 1-\epsilon.
\end{align*}
Since $\uMAPsm$ is the maximizer of $\log\musm(\v|\vct{Y}_{0:k})$ on $\BR$, our claim  follows.
\end{proof}

\begin{lem}[A bound on $\|\uMAPsm-\uG\|$]\label{uMAPuGdifflem}
There are finite constants $S_{\mathrm{MAP}}^{(1)}>0$, $S_{\mathrm{MAP}}^{(2)}$,  $D_7(\u,T)$ and $D_8(\u,T)$ such that for any $0<\epsilon\le 1$, for 
$\sigma_Z \sqrt{h}< \l(S_{\mathrm{MAP}}^{(1)}+S_{\mathrm{MAP}}^{(2)}\l(\log\l(\f{1}{\epsilon}\r)\r)^{1/2}\r)^{-1}$,
we have
\begin{equation}\label{uMAPuGdiffeq}\Pcu{\mtx{A}_k\succ \mtx{0} \text{ and }\|\uMAPsm-\uG\|\le \l(D_7(\u,T)+D_8(\u,T)\l(\log\l(\f{1}{\epsilon}\r)\r)^{\f{3}{2}}\r)\sigma_Z^2 h}\ge 1-\epsilon.
\end{equation}
\end{lem}
\begin{proof}
By choosing $S_{\mathrm{MAP}}^{(1)}$ and $S_{\mathrm{MAP}}^{(2)}$ sufficiently large, we can assume that
\begin{equation}\label{sigmazsqrthbnd1eq}\sigma_Z\sqrt{h}< \min\l(\f{c(\u,T)}{2} \cdot \f{R-\|\u\|}{C_1\l(\u,T,\f{\epsilon}{3}\r)},\sqrt{ \f{(R-\|\u\| )c(\u,T)}{C_q^{(1)}}},\f{c(\u,T)}{2 C_{\mtx{A}}\l(\f{\epsilon}{3}\r)}\r).\end{equation}
From Lemma \ref{lemmaUMAPudiff}, we know that 
\begin{equation}\label{eqUMAPu1}\Pcu{\|\uMAPsm-\u\|\le \f{C_1(\u,T,\f{\epsilon}{3})\sigma_Z\sqrt{h}+2C_q^{(1)}\sigma_Z^2h}{c(\u,T)}}\ge 1-\f{\epsilon}{3}.
\end{equation}
Using \eqref{sigmazsqrthbnd1eq}, it follows that if the above event happens, then $\|\uMAPsm\|<R$,  and thus 
\begin{equation}\label{eqUMAPu2}
\grad \log(\musm(\uMAPsm|\vct{Y}_{0:k})))=\grad \lsm(\uMAPsm)-2\sigma_Z^2 \grad \log q(\uMAPsm)=0.\end{equation}
Using the fact that $\grad\lsmG(\v)=2\mtx{A}_k(\v-\uG)$, and Proposition \ref{lsmlsmGgraddiffprop}, it follows that
\begin{equation}\label{eqUMAPu3}\Pcu{\|\grad \lsm(\v)\|\ge \|2\mtx{A}_k(\v-\uG)\|-\|\v-\u\|^2 \cdot \f{C_4(\u,T)+C_5(\u,T,\f{\epsilon}{3})\sigma_Z\sqrt{h}}{h}}\ge 1-\f{\epsilon}{3}.\end{equation}
Moreover, by Lemma \ref{Akeigboundlemma}, we know that for any $0< \epsilon\le 1$, we have
\begin{equation}\label{eqUMAPu4}
\Pcu{\lambda_{\min}(\mtx{A}_k)> \f{c(\u,T)}{2h}}\ge 1-\f{\epsilon}{3} \text{ for }\sigma_Z \sqrt{h}\le \f{c(\u,T)}{2 C_{\mtx{A}}\l(\f{\epsilon}{3}\r)}.
\end{equation}
By combining the four equations \eqref{eqUMAPu1}, \eqref{eqUMAPu2}, \eqref{eqUMAPu3} and \eqref{eqUMAPu4}, it follows that with probability at least $1-\epsilon$, we have
\begin{align*}&2\sigma_Z^2 C_q^{(1)}\ge \f{c(\u,T)}{h} \|\uMAPsm-\uG\|\\
&-\l(\f{C_1(\u,T,\f{\epsilon}{3})\sigma_Z\sqrt{h}+2C_q^{(1)}\sigma_Z^2h}{c(\u,T)}\r)^2 \cdot 
\f{C_4(\u,T)+C_5(\u,T,\f{\epsilon}{3})\sigma_Z\sqrt{h}}{h},
\end{align*}
and the claim of the lemma follows by rearrangement.
\end{proof}

\begin{lem}[A lower bound on $\E\l(\l.\|\umean-\u\|^2\r|\u\r)$]\label{lowerbndlem}
There are positive constants $D_9(\u,T)$ and $D_{10}(\u,T)$ such that for $\sigma_Z\sqrt{h}\le D_{10}(\u,T)$, we have 
\begin{equation}\label{uumeaniffloweq}\E\l(\l.\|\umean-\u\|^2\r|\u\r)\ge D_9(\u,T) \cdot \sigma_Z^2 h.
\end{equation}
\end{lem}
\begin{proof}
By applying Lemma \ref{umeanuGdifflemma} for $\epsilon=0.1$, we obtain that for 
$\sigma_Z\sqrt{h}\le \f{1}{2}\cdot C_{\mr{TV}}(\u,T,0.1)^{-1}$, we have
\begin{align}\label{umeanuGdiffeps01eq}&\PP\Bigg(\lambda_{\min}(\mtx{A}_k)> \f{c(\u,T)}{2h}\text{ and } 
\|\mtx{A}_k\|<  \f{C_{\|\mtx{A}\|}}{h} \\
&\nonumber\text{ and }\|\umean-\uG\|\le \l(D_5(\u,T)+D_6(\u,T)\l(\log\l(10\r)\r)^{2}\r)\sigma_Z^2 h\Bigg|\u\Bigg)\ge 0.9.
\end{align}
If this event happens, then in particular, we have
\begin{equation}\label{uGuboundeq1}\|\uG-\u\|=\|\mtx{A}_k^{-1} \vct{B}_k\|\ge \f{h}{C_{\|\mtx{A}\|}}\cdot \|\vct{B}_k\|.\end{equation}
By the definition of $\vct{B}_k$ in \eqref{eqBkdef}, it follows that conditioned on $\u$, $\vct{B}_k$ has $d$-dimensional multivariate normal distribution with covariance matrix  $\Sigma_{\vct{B}_k}:=\sigma_Z^2\sum_{i=0}^k \J\Phi_{t_i}(\u)'\J\Phi_{t_i}(\u)$. This means that if $\vct{Z}$ is a $d$ dimensional standard normal random vector, then $(\Sigma_{\vct{B}_k})^{1/2}\cdot \vct{Z}$ has the same distribution as $\vct{B}_k$ (conditioned on $\u$). By Assumption \ref{assgauss1}, we have
$\lambda_{\min}\l(\Sigma_{\vct{B}_k}\r)\ge \sigma_Z^2 \cdot \f{c(\u,T)}{h}$, and thus $\|(\Sigma_{\vct{B}_k})^{1/2}\cdot \vct{Z}\|\ge \sigma_Z\cdot \sqrt{\f{c(\u,T)}{h}}\cdot \|\vct{Z}\|$.

It is not difficult to show that for any $d\ge 1$, $\PP\l(\|\vct{Z}\|\ge \f{\sqrt{d}}{2}\r)\ge \f{1}{4}$ (indeed, if $(Z_{(i)})_{1\le i\le d}$ are i.i.d. standard normal random variables, then for any $\lambda>0$, $\E(e^{-\lambda Z_1^2})=\sqrt{\f{1}{1+2\lambda}}$, so $\E(e^{-\lambda (\|\vct{Z}\|^2-d)})=\E(e^{-\lambda \sum_{i=1}^d(Z_{(i)}^2-1)})=(1+2\lambda)^{-d/2}\cdot e^{\lambda d}$, and the claim follows by applying Markov's inequality $\PP(\|\vct{Z}\|^2-d\le -t)\le \E(e^{-\lambda (\|\vct{Z}\|^2-d)})\cdot e^{-\lambda t}$ for $t=\f{3}{4}d$ and $\lambda=1$). Therefore, we have
\[\Pcu{ \|\vct{B}_k\|\ge \sigma_Z\cdot \sqrt{\f{c(\u,T)}{h}} \cdot \f{\sqrt{d}}{2}}\ge \f{1}{4},\]
and thus by \eqref{umeanuGdiffeps01eq} and \eqref{uGuboundeq1}, it follows that
\begin{align*}
&\PP\Bigg(\|\umean-\u\|\ge \f{\sigma_Z\sqrt{h} \cdot \sqrt{d c(\u,T)}}{2 C_{\|\mtx{A}\|}} -\l(D_5(\u,T)+D_6(\u,T)\l(\log\l(10\r)\r)^{2}\r)\sigma_Z^2 h\Bigg)\ge 0.15.
\end{align*}
By choosing $D_{10}(\u,T)$ sufficiently small, we have that for $\sigma_Z\sqrt{h}\le D_{10}(\u,T)$, 
\begin{align*}\l(D_5(\u,T)+D_6(\u,T)\l(\log\l(10\r)\r)^{2}\r)\sigma_Z^2 h\le \f{1}{2}\cdot \f{\sigma_Z\sqrt{h} \cdot \sqrt{d c(\u,T)}}{2C_{\|\mtx{A}\|}},
\end{align*}
and the result follows.
\end{proof}

\begin{lem}[A bound on the difference of $\E\l(\l.\|\umean-\u\|^2\r|\u\r)$ and $\E\l(\l.\|\uMAPsm-\u\|^2\r|\u\r)$]\label{difflem}
There are some finite constants $D_{11}(\u,T)$ and $D_{12}(\u,T)>0$ such that for $\sigma_Z\sqrt{h}\le D_{12}(\u,T)$, we have
\begin{equation}\label{uMAPumeandiffeq}
\l|\E\l(\l.\|\umean-\u\|^2\r|\u\r)-\E\l(\l.\|\uMAPsm-\u\|^2\r|\u\r)\r|\le 
D_{11}(\u,T) \cdot (\sigma_Z^2 h)^{\frac{3}{2}}.
\end{equation}
\end{lem}
\begin{proof}
We define the event $E_k$ as 
\begin{equation}\label{Ekdefeq}
E_k:=\l\{\f{c(\u,T)}{2h}\cdot\mtx{I}_d\prec \mtx{A}_k \prec  \f{C_{\|\mtx{A}\|}}{h} \cdot\mtx{I}_d, \,\|\mtx{A}_k^{-1} \vct{B}_k\|<R-\|\u\|\r\}.
\end{equation}
Under this event, we have, in particular  $\mtx{A}_k\succ \mtx{0}$ and $\|\uG\|<R$. Let $E_k^c$ denote the complement of $E_k$. Then the difference in the variances can be bounded as
\begin{align}
\nonumber&\l|\E\l(\l.\|\umean-\u\|^2\r|\u\r)-\E\l(\l.\|\uMAPsm-\u\|^2\r|\u\r)\r|\\
\nonumber&\le \Ecu{4R^2 1_{E_k^c} + 1_{E_k} \l(\l|\|\umean-\u\|^2-\|\uG-\u\|^2\r| + \l|\|\uMAPsm-\u\|^2-\|\uG-\u\|^2\r|\r)}\\
\nonumber&\le 4R^2 \Pcu{E_k^c} + \E\Big(1_{E_k} \big(\|\umean-\uG\|\l(\|\umean-\u\|+2\|\uG-\u\|\r)\\
&\nonumber+\|\uMAPsm-\uG\|\l(\|\uMAPsm-\uG\|+2\|\uG-\u\|\r)\big)\Big|\u\Big)\\
&\nonumber \le 4R^2 \Pcu{E_k^c}+ \Ecu{1_{E_k} \|\umean-\uG\|^2 }+\Ecu{1_{E_k} \|\uMAPsm-\uG\|^2 }\\
\label{diffbndeq}&+2\sqrt{\Ecu{1_{E_k} \|\uG-\u\|^2 }}\cdot \l(\sqrt{\Ecu{1_{E_k} \|\umean-\uG\|^2 }}+\sqrt{\Ecu{1_{E_k} \|\uMAPsm-\uG\|^2 }}\r),
\end{align}
where in the last step we have used the Cauchy-Schwarz inequality. The above terms can be further bounded as follows. By Lemmas \ref{Akeigboundlemma} and \ref{Bknormboundlemma} it follows that
\begin{align}
\nonumber&\Pcu{E_k^c}\\
\nonumber&=\Pcu{\lambda_{\min}(\mtx{A}_k)\le \f{c(\u,T)}{2h} \text{ or } \|\mtx{A}_k\|\le\f{C_{\|\mtx{A}\|}}{h} \text{ or }\|\mtx{A}_k^{-1}\vct{B}_k\|\ge R-\|\u\| }\\
&\label{Ekcbndeq}\le 2d\exp\l(-\f{c(\u,T)^2}{4\olT \hM_2(T)^2 d_o \sigma_Z^2 h}\r)+(d+1)\exp\l(-\f{(R-\|\u\|)^2 c(\u,T)^2}{\olT  \hM_1(T)d_o \sigma_Z^2 h}\r)\le C_{E}(\u,T)(\sigma_Z^2 h)^2,
\end{align}
for some finite constant $C_{E}(\u,T)$ independent of $h$ and $\sigma_Z$.

The term $\Ecu{1_{E_k} \|\uG-\u\|^2 }$ can be bounded as
\begin{align}\nonumber
\Ecu{1_{E_k} \|\uG-\u\|^2 }&=\Ecu{1_{E_k} \|\mtx{A}_k^{-1}\vct{B}_k\|^2 }\le \l(\f{2h}{c(\u,T)}\r)^2\cdot \Ecu{1_{E_k} \|\vct{B}_k\|^2 } \\
&\label{uGuvarianceeq}\le \l(\f{2h}{c(\u,T)}\r)^2 \hM_1(T)^2 d_o \f{\olT}{h} \sigma_Z^2=\f{4 \olT \hM_1(T)^2  d_o}{c(\u,T)^2} \cdot \sigma_Z^2 h.
\end{align}
For bounding the term $\Ecu{1_{E_k} \|\umean-\uG\|^2 }$, we define 
\begin{align*}t_{\min}&:=D_5(\u,T) \sigma_Z^2 h,\text{ and }\\
t_{\max}&:=\l(D_5(\u,T)+D_6(\u,T)\cdot \l(\f{(2\sigma_Z \sqrt{h})^{-1}-C_{\mathrm{TV}}^{(1)}(\u,T) }{C_{\mathrm{TV}}^{(2)}(\u,T)}\r)^{\f{5}{4}}\r) \sigma_Z^2h,
\end{align*}
then by Lemma \ref{umeanuGdifflemma}, it follows that for $\sigma_Z\sqrt{h}<\f{1}{2}(C_{\mathrm{TV}}^{(1)})^{-1}$,  for $t\in [t_{\min},t_{\max}]$, we have
\begin{equation}\label{umeanuGdiffteq}
\Pcu{1_{E_k} \|\umean-\uG\|\ge t}\le \exp\l(-\l(\f{t/(\sigma_Z^2h) -D_5(\u,T)}{D_6(\u,T)}\r)^{\f{2}{5}}\r).
\end{equation}
By writing 
\[\Ecu{1_{E_k} \|\umean-\uG\|^2 }=\int_{t=0}^{\infty}\PP(1_{E_k} \|\umean-\uG\|^2>t)dt=\int_{t=0}^{\infty}\PP(1_{E_k} \|\umean-\uG\|>\sqrt{t})dt,\]
and using the fact that $1_{E_k} \|\umean-\uG\|^2<4R^2$, one can show that for $\sigma_Z\sqrt{h}<\ol{S}(\u,T)$, 
\begin{equation}\label{umeanuGvarbndeq}\Ecu{1_{E_k} \|\umean-\uG\|^2 }\le \ol{C}(\u,T) (\sigma_Z^2 h)^{2},
\end{equation}
for some constants $\ol{S}(\u,T)>0$, $\ol{C}(\u,T)<\infty$, that are independent of $\sigma_Z$ and $h$.  

Finally, for bounding the term $\Ecu{1_{E_k} \|\uMAPsm-\uG\|^2 }$, we define 
\[t_{\min}':=D_7(\u,T)\sigma_Z^2 h,\text{ and }t_{\max}':=\l(D_7(\u,T)+D_8(\u,T)\l(\f{(\sigma_Z\sqrt{h})^{-1}-S_{\mathrm{MAP}}^{(1)}}{S_{\mathrm{MAP}}^{(2)}}\r)^3\r)\cdot \sigma_Z^2 h.\]
By Lemma \ref{uMAPuGdifflem}, it follows that for $\sigma_Z\sqrt{h}< \l(S_{\mathrm{MAP}}^{(1)}\r)^{-1}$, for $t\in [t_{\min}', t_{\max}']$, we have
\begin{equation}
\Pcu{\|\uMAPsm-\uG\|> t}\le \exp\l(-\l(\f{t/(\sigma_Z^2h)-D_7(\u,T)}{D_8(\u,T)}\r)^{\f{2}{3}}\r),
\end{equation}
which implies that for $\sigma_Z\sqrt{h}<S_{\mathrm{M}}(\u,T)$,
\begin{equation}\label{uMAPuGvarbndeq}\Ecu{1_{E_k} \|\umean-\uG\|^2 }\le C_{\mathrm{M}}(\u,T) (\sigma_Z^2 h)^{2},
\end{equation}
for some constants $S_{\mathrm{M}}(\u,T)>0$, $C_{\mathrm{M}}(\u,T)>0$.

The result now follows by \eqref{diffbndeq} and the bounds \eqref{Ekcbndeq}, \eqref{uGuvarianceeq}, \eqref{umeanuGvarbndeq} and \eqref{uMAPuGvarbndeq}.
\end{proof}
\begin{proof}[Proof of Theorem \ref{GaussianMAPthm}]
The lower bound on $\f{\Ecu{\|\umean-\u\|^2}}{\sigma_Z^2 h}$ follows by Lemma \ref{lowerbndlem}.
Let $E_k$ be defined as in \eqref{Ekdefeq}, then we have
\begin{align*}&\Ecu{\|\umean-\u\|^2}\\
&\le \Ecu{\|\umean-\u\|^2\cdot 1_{E_k^c}}+2 \Ecu{\|\umean-\uG\|^2\cdot 1_{E_k}}+2 \Ecu{\|\uG-\u\|^2\cdot 1_{E_k}}\\
&\le 4R^2 \Pcu{1_{E_k^c}}+2 \Ecu{\|\umean-\uG\|^2\cdot 1_{E_k}}+2 \Ecu{\|\uG-\u\|^2\cdot 1_{E_k}},
\end{align*}
so the upper bound on $\f{\Ecu{\|\umean-\u\|^2}}{\sigma_Z^2 h}$ follows by \eqref{Ekcbndeq}, \eqref{uGuvarianceeq} and \eqref{umeanuGvarbndeq}. Finally, the bound on $\l|\E\l[\|\uMAPsm-\u\|^2|\u\r]-\E\l[\|\umean-\u\|^2|\u\r]\r|$ follows directly from Lemma \ref{difflem}.
\end{proof}

\subsubsection{Comparison of push-forward MAP and posterior mean for the filter}
The main idea of proof is similar to the proof of Theorem \ref{GaussianMAPthm}. 
We are going to use the following Lemmas (variants of Lemmas \ref{umeanuGdifflemma}-\ref{difflem}).

\begin{lem}[A bound on $\|\ufimean-\Psi_T(\uG)\|$]\label{umeanuGdifflemmafi}
There are some finite constants $D_5'(\u,T)$ and $D_6'(\u,T)$ such that for any $0<\epsilon\le 1$, for $\sigma_Z\sqrt{h}\le \f{1}{2}\cdot D_{\mr{TV}}(\u,T,\epsilon)^{-1}$, we have
\begin{align*}&\PP\Bigg(\f{c(\u,T)}{2h}\mtx{I}_d\prec \mtx{A}_k \prec \f{C_{\|\mtx{A}\|}}{h}\mtx{I}_d\text{ and }
\uG\in \BR \text{ and }\\
&\|\ufimean-\Psi_T(\uG)\|\le \l(D_5'(\u,T)+D_6'(\u,T)\l(\log\l(\f{1}{\epsilon}\r)\r)^{2}\r)\sigma_Z^2 h\Bigg|\u\Bigg)\ge 1-\epsilon.
\end{align*}
\end{lem}
\begin{proof}
This is a direct consequence of the Wasserstein distance bound of Theorem \ref{Gaussianapproxthmfi}.
\end{proof}

\begin{lem}[A bound on $\|\uMAPfi-\Psi_T(\uG)\|$]\label{uMAPuGdifflemfi}
There are finite constants $S_{\mathrm{MAP}}^{(1')}>0$, $S_{\mathrm{MAP}}^{(2')}$,  $D_7'(\u,T)$ and $D_8'(\u,T)$ such that for any $0<\epsilon\le 1$, for 
$\sigma_Z \sqrt{h}< \l(S_{\mathrm{MAP}}^{(1')}+S_{\mathrm{MAP}}^{(2')}\l(\log\l(\f{1}{\epsilon}\r)\r)^{1/2}\r)^{-1}$,
we have
\begin{align}\label{uMAPuGdiffeqfi}\PP\Bigg(&\mtx{A}_k\succ \mtx{0} \text{ and } \uG\in \BR \text{ and }\\
&\nonumber\|\uMAPfi-\Psi_T(\uG)\|\le \l(D_7'(\u,T)+D_8'(\u,T)\l(\log\l(\f{1}{\epsilon}\r)\r)^{\f{3}{2}}\r)\sigma_Z^2 h\Bigg|\u\Bigg)\ge 1-\epsilon.
\end{align}
\end{lem}
\begin{proof}
The result follows by Lemma \ref{uMAPuGdifflem}, Theorem \ref{Gaussianapproxthmfi}, and \eqref{eqpathdistancebound}.
\end{proof}

\begin{lem}[A lower bound on $\E\l(\l.\|\ufimean-\u(T)\|^2\r|\u\r)$]\label{lowerbndlemfi}
There are positive constants $D_9'(\u,T)$ and $D_{10}'(\u,T)$ such that for $\sigma_Z\sqrt{h}\le D_{10}'(\u,T)$, we have 
\begin{equation}\label{uumeaniffloweq}\E\l(\l.\|\ufimean-\u(T)\|^2\r|\u\r)\ge D_9'(\u,T) \cdot \sigma_Z^2 h.
\end{equation}
\end{lem}
\begin{proof}
The proof is similar to the proof of Lemma \ref{lowerbndlem}. By applying Lemma \ref{umeanuGdifflemmafi} for $\epsilon=0.1$, we obtain that for 
$\sigma_Z\sqrt{h}\le \f{1}{2}\cdot D_{\mr{TV}}(\u,T,0.1)^{-1}$, we have
\begin{align}\label{umeanuGdiffeps01eqfi}&\PP\Bigg(\lambda_{\min}(\mtx{A}_k)> \f{c(\u,T)}{2h}\text{ and } 
\|\mtx{A}_k\|<  \f{C_{\|\mtx{A}\|}}{h} \\
&\nonumber\text{ and }\|\Psi_{-T}(\ufimean)-\uG\|\le \l(D_5'(\u,T)+D_6'(\u,T)\l(\log\l(10\r)\r)^{2}\r)\sigma_Z^2 h\Bigg|\u\Bigg)\ge 0.9.
\end{align}
If this event happens, then by \eqref{eqpathdistancebound}, we have
\[\|\Psi_T(\uG)-\u(T)\|\ge \exp(-GT)\|\uG-\u\|\ge \exp(-GT)\|\mtx{A}_k^{-1} \vct{B}_k\|\ge \exp(-GT)\f{h}{C_{\|\mtx{A}\|}}\cdot \|\vct{B}_k\|.\]
The rest of the argument is the same as in the proof of Lemma \ref{lowerbndlem}, so it is omitted.
\end{proof}

\begin{lem}[A bound on the difference of $\E\l(\l.\|\ufimean-\u\|^2\r|\u\r)$ and $\E\l(\l.\|\uMAPfi-\u\|^2\r|\u\r)$]\label{difflemfi}
There are some finite constants $D_{11}'(\u,T)$ and $D_{12}'(\u,T)>0$ such that for $\sigma_Z\sqrt{h}\le D_{12}'(\u,T)$, we have
\begin{equation}\label{uMAPumeandiffeq}
\l|\E\l(\l.\|\ufimean-\u(T)\|^2\r|\u\r)-\E\l(\l.\|\uMAPfi-\u(T)\|^2\r|\u\r)\r|\le 
D_{11}'(\u,T) \cdot (\sigma_Z^2 h)^{\frac{3}{2}}.
\end{equation}
\end{lem}
\begin{proof}
We define the event $E_k$ as in \eqref{Ekdefeq}. Under this event, $\mtx{A}_k\succ \mtx{0}$ and $\|\uG\|<R$, so $\mufiG(\cdot|\Yok)$ is defined according to \eqref{mufiGdefeq}. The proof of the claim of the lemma follows the same lines as the proof of Lemma \ref{difflem}. In particular, we obtain from \eqref{uGuvarianceeq} and \eqref{eqpathdistancebound} that
\begin{equation}\label{uGuvarianceeqfi}\Ecu{1_{E_k} \|\Psi_T(\uG)-\u(T)\|^2 }\le \f{4 \olT \hM_1(T)^2  d_o \exp(GT)}{c(\u,T)^2} \cdot \sigma_Z^2 h.
\end{equation}
Based on Lemma \ref{umeanuGdifflemmafi}, we obtain that for $\sigma_Z\sqrt{h}<\ol{S}'(\u,T)$, 
\begin{equation}\label{umeanuGvarbndeqfi}\Ecu{1_{E_k} \|\ufimean-\Psi_T(\uG)\|^2 }\le \ol{C}'(\u,T) (\sigma_Z^2 h)^{2},
\end{equation}
for some constants $\ol{S}'(\u,T)>0$, $\ol{C}'(\u,T)<\infty$, that are independent of $\sigma_Z$ and $h$.  We omit the details.
\end{proof}

\begin{proof}[Proof of Theorem \ref{GaussianMAPthmfi}]
The lower bound on $\f{\Ecu{\|\ufimean-\u(T)\|^2}}{\sigma_Z^2 h}$ follows by Lemma \ref{lowerbndlemfi}. Similarly to the case of the smoother, we have 
\begin{align*}\Ecu{\|\ufimean-\u(T)\|^2}
&\le 4R^2 \Pcu{1_{E_k^c}}+2 \Ecu{\|\ufimean-\Psi_T(\uG)\|^2\cdot 1_{E_k}}\\
&+2 \Ecu{\|\Psi_T(\uG)-\u(T)\|^2\cdot 1_{E_k}},
\end{align*}
which can be further bounded by \eqref{Ekcbndeq}, \eqref{uGuvarianceeqfi} and \eqref{umeanuGvarbndeqfi} to yield the upper bound on $\f{\Ecu{\|\ufimean-\u(T)\|^2}}{\sigma_Z^2 h}$. Finally, the bound on $\l|\E\l[\|\uMAPfi-\u(T)\|^2|\u\r]-\E\l[\|\ufimean-\u(T)\|^2|\u\r]\r|$ follows directly from Lemma \ref{difflemfi}.
\end{proof}

\subsection{Convergence of Newton's method to the MAP}
In this section, we will prove Theorem \ref{NewtonMAPthm}. The following proposition shows a classical bound on the convergence of Newton's method (this is a reformulation of Theorem 5.3 of \cite{bubeck2015convex} to our setting). For $\v\in R^d$, $r>0$, we denote the ball of radius $r$ centered at $\v$ by $B(\v,r):=\{\vct{x}\in \R^d: \|\v-\vct{x}\|\le r^*\}$)
\begin{prop}\label{propNewtonlocalconvergence}
Suppose that $\Omega\subset \R^d$ is an open set, and $g: \Omega\to \R$ is a 3 times continuously differentiable function satisfying that
\begin{enumerate}
\item $g$ has a local minimum at a point $\vct{x}^*\in \Omega$,
\item there exists a radius $r^*>0$ and constants $C_H>0, L_H<\infty$ such that $B(\vct{x}^*,r^*)\subset \Omega$, $\grad^2 g(\vct{x})\succeq C_H \cdot \mtx{I}_d$ for every $\vct{x}\in B(\vct{x}^*,r^*)$, and $\grad^2 g(\vct{x})$ is $L_H$-Lipschitz on $B(\vct{x}^*, r^*)$.
\end{enumerate}
Suppose that the starting point $x_0\in \Omega$ satisfies that $\|\vct{x}_0-\vct{x}^*\|<\min\l(r^*, 2\frac{C_H}{L_H}\r)$. Then the iterates of Newton's method defined recursively for every $i\in \N$ as
\[\vct{x}_{i+1}:=\vct{x}_i-(\grad^2 g(\vct{x}_i))^{-1}\cdot \grad g(\vct{x}_i)\]
always stay in $B(\vct{x}^*, r^*)$ (thus they are well defined), and satisfy that
\begin{equation}\label{eqNewtonbnd}
\|\vct{x}_{i}-\vct{x}^*\|\le \frac{2C_H}{L_H} \cdot \l( \frac{L_H}{2 C_H}\|\vct{x}_0-\vct{x}^*\|\r)^{\displaystyle{2^i}} \text{ for every }i\in \N.
\end{equation}
\end{prop}
\begin{proof}
We will show that $\vct{x}_{i}\in B(\vct{x}^*,r^*)$ by induction in $i$. This is true for $i=0$. Assuming that $\vct{x}_i\in B(\vct{x}^*,r^*)$, we can write the gradient $\grad g(\vct{x}_i)$ as
\begin{align*}
\grad g(\vct{x}_i)&=\int_{t=0}^{1} \grad^2 g(\vct{x}^*+t(\vct{x}_i-\vct{x}^*)) \cdot (\vct{x}_i-\vct{x}^*) dt,\text{ therefore }\\
\vct{x}_{i+1}-\vct{x}^*&=\vct{x}_i-\vct{x}^*-(\grad^2 g(\vct{x}_i))^{-1}\cdot \grad g(\vct{x}_i)\\
&=\vct{x}_i-\vct{x}^*-(\grad^2 g(\vct{x}_i))^{-1}\cdot \int_{t=0}^{1} \grad^2 g(\vct{x}^*+t(\vct{x}_i-\vct{x}^*)) \cdot (\vct{x}_i-\vct{x}^*) dt\\
&=(\grad^2 g(\vct{x}_i))^{-1}\int_{t=0}^{1} \l[ \grad^2 g(\vct{x}_i)-\grad^2 g(\vct{x}^*+t(\vct{x}_i-\vct{x}^*))\r] \cdot (\vct{x}_i-\vct{x}^*) dt.
\end{align*}
By the $L_H$-Lipschitz property of $\grad^2 g(\vct{x})$ on $B(\vct{x}^*,r^*)$, we have
\[\int_{t=0}^{1} \l\| \grad^2 g(\vct{x}_i)-\grad^2 g(\vct{x}^*+t(\vct{x}_i-\vct{x}^*))\r\| dt\le \frac{L_H}{2}\|\vct{x}_i-\vct{x}^*\|.\]
By combining this with the fact that $\|(\grad^2 g(\vct{x}_i))^{-1}\|\le \frac{1}{C_H}$, we obtain that
$\|\vct{x}_{i+1}-\vct{x}^*\|\le \frac{L_H}{2 C_H} \|\vct{x}_i-\vct{x}^*\|^2$ for every $i\in \N$, and by rearrangement, it follows that 
$\log\l(\frac{L_H}{2 C_H} \|\vct{x}_{i+1}-\vct{x}^*\|\r)\le 2 \log \l(\frac{L_H}{2 C_H}\|\vct{x}_i-\vct{x}^*\|\r)$ for every  $i\in \N$, hence the result.
\end{proof}

The following proposition gives a lower bound on the Hessian near $\u$. The proof is included in Section \ref{Secproofpreliminaryresults} of the Appendix.
\begin{prop}\label{proplsmHessianlowerbnd}
Suppose that Assumptions \ref{assgauss1} and \ref{assprior} hold for the initial point $\u$ and the prior $q$. Let \[r_H(\u,T):=\min\l(\frac{c(\u,T)}{8\olT \hM_1(T)\hM_2(T)},R-\|\u\|\r).\]
Then for any $0<\epsilon\le 1$, $\sigma_Z>0$,  $0< h\le h_{\max}(\u,T)$, we have
\begin{align}
&\label{eqHesslowerbound}\PP\Bigg(\grad^2 \log\musm(\v|\vct{Y}_{0:k}) \preceq \frac{-\frac{3}{4}c(\u,T)+C_6(\u,T,\epsilon)\sigma_Z\sqrt{h}+C_{q}^{(2)}\sigma_Z^2h}{\sigma_Z^2 h}\cdot \mtx{I}_d\\
&\nonumber\quad \text{ for every }\v\in B(\u, r_H(\u,T))\bigg| \u\Bigg)\ge 1-\epsilon,\text{ where }\\
& \nonumber C_6(\u,T,\epsilon):=33\hM_2(T)(R+1)\sqrt{\olT (2d+1)d_o}+\sqrt{2\olT \hM_2(T)d_o \log\l(\frac{1}{\epsilon}\r)}.
\end{align}
\end{prop}

The following proposition bounds the Lipschitz coefficient of the Hessian. The proof is included in Section \ref{Secproofpreliminaryresults} of the Appendix.
\begin{prop}\label{proplsmHessianLipschitz}
Suppose that Assumptions \ref{assgauss1} and \ref{assprior} hold for the initial point $\u$ and the prior $q$. Then for any $0<\epsilon\le 1$, $\sigma_Z>0$,  $0< h\le h_{\max}(\u,T)$, we have
\begin{align*}
&\PP\Bigg(\|\grad^3 \log\musm(\v|\vct{Y}_{0:k})\|\le \\
&\quad C_{q}^{(3)}+\frac{\olT}{\sigma_Z^2h}\l(
C_7(\u,T)+C_8(\u,T,\epsilon)\sigma_Z\sqrt{h}\r)\text{ for every }\v\in \BR\bigg|\u\Bigg)\ge 1-\epsilon,\text{ where }
\end{align*}
\begin{align*}
&C_7(\u,T):=3\hM_1(T)\hM_2(T)+2\hM_1(T)\hM_3(T) R,\text{ and }\\
&C_8(\u,T,\epsilon):=44(\hM_4(T)R+\hM_3(T))\sqrt{3\olT (d+1)d_o}+\sqrt{2\olT \hM_3(T)d_o \log\l(\frac{1}{\epsilon}\r)}.
\end{align*}
\end{prop}
Now we are ready to prove Theorem \ref{NewtonMAPthm}.
\begin{proof}[Proof of Theorem \ref{NewtonMAPthm}]
Let
\begin{align}\label{Smaxsmdefeq}
S_{\max}^{\mathrm{sm}}(\u,T,\epsilon):=\min\Bigg(&\f{c(\u,T)}{8C_6(\u,T,\f{\epsilon}{3})},\sqrt{\f{c(\u,T)}{8C_q^{(2)}}},\f{C_7(\u,T)}{8C_8(\u,T,\f{\epsilon}{3})},\sqrt{\f{C_7(\u,T) \olT}{C_q^{(3)}}}, \\
&\nonumber \f{c(\u,T)\min\l(r_H(\u,T), \f{c(\u,T)}{C_7(\u,T)}\r)}{8C_1(\u,T,\f{\epsilon}{3})}, \sqrt{\f{c(\u,T) \min\l(r_H(\u,T), \f{c(\u,T)}{C_7(\u,T)}\r)}{8C_q^{(1)}}}\Bigg).
\end{align}
Then by the assumption that 
$\sigma_Z\sqrt{h}\le S_{\max}^{\mathrm{sm}}(\u,T,\epsilon)$, using Propositions \ref{proplsmHessianlowerbnd}, \ref{proplsmHessianLipschitz} and inequality \eqref{uMAPudiffbndeq} we know that with probability at least $1-\epsilon$, all three of the following events hold at the same time,
\begin{enumerate}
\item $\grad^2 g^{\mr{sm}}(\v)\succeq C_g \cdot \mtx{I}_d$ for every $\v\in B(\u, r_{H}(\u,T))$ for $C_g:=\f{1}{2} \f{c(\u,T)}{\sigma_Z^2 h}$,
\item $\grad^2 g^{\mr{sm}}$ is $L$-Lipschitz in $\BR$ for $L:=\f{2C_7(\u,T)}{\sigma_Z^2 h}$,
\item $\|\u-\uMAPsm\|\le \min\l(\f{r_H(\u,T)}{4},\f{1}{4}\f{c(\u,T)}{C_7(\u,T)} \r)$.
\end{enumerate}
If these events hold, then the conditions of Proposition \ref{propNewtonlocalconvergence} are satisfied for the function $g$ with $\vct{x}^*=\uMAPsm$, $r^*:=\f{3}{4}r_H(\u,T)$, $C_H:= \f{1}{2} \f{c(\u,T)}{\sigma_Z^2 h}$ and $L_H:=\f{2C_7(\u,T)}{\sigma_Z^2 h}$. Therefore, \eqref{eqNewtonbnd} holds if $\|\vct{x}_0-\uMAPsm\|\le \min\l(\f{3}{4}r_H(\u,T),\f{1}{2}\f{c(\u,T)}{C_7(\u,T)} \r)$. By the triangle inequality, and the third event above, this is satisfied if
$\|\vct{x}_0-\u\|\le \min\l(\f{1}{2}r_H(\u,T),\f{1}{4}\f{c(\u,T)}{C_7(\u,T)} \r)$. Thus the claim of the theorem follows for 
\begin{equation}\label{DmaxCNdefeq}
D_{\max}^{\mr{sm}}(\u,T):=\min\l(\f{1}{2}r_H(\u,T),\f{1}{4}\f{c(\u,T)}{C_7(\u,T)} \r),\text{ and }N^{\mr{sm}}(\u,T):=\f{c(\u,T)}{2C_7(\u,T)}.\qedhere
\end{equation}
\end{proof}
\subsection{Initial estimator}\label{secInitialproof}
In this section, we will prove our results about the initial estimator that we have proposed in Section \ref{secInitial}.
\begin{proof}[Proof of Proposition \ref{propestbiasconc1}]
By Taylor series expansion of $\hat{\Phi}^{(l|\jmax)}(\Y_{0:\hk})$ with remainder term of order $\jmax+1$, we obtain that
\begin{align*}
&\Ecu{\hat{\Phi}^{(l|\jmax)}(\Y_{0:\hk})}-\mtx{H} \frac{d^l \u}{ dt^l}=
\sum_{i=0}^{\hk} c^{(l|\jmax|\hk)}_i \mtx{H}\u(ih)-\mtx{H} \frac{d^l \u}{ dt^l}\\
&=
\l(\sum_{j=0}^{\jmax} \frac{(\hk h)^j}{j!} \mtx{H} \frac{d^j \u}{ dt^j}\l<\c^{(l|\jmax|\hk)}, \v^{(j|\hk)}\r>\r)-\mtx{H} \frac{d^l \u}{ dt^l}+R_{l,\jmax+1},
\end{align*}
where by \eqref{uderboundeq} the remainder term $R_{l,\jmax+1}$ can be bounded using the Cauchy-Schwarz inequality as
\begin{align*}&\|R_{l,\jmax+1}\|\le C_0 \|\mtx{H}\|\sum_{i=0}^{\hk}\l|c^{(l|\jmax|\hk)}_i\r| \frac{(ih)^{\jmax+1}}{(\jmax+1)!} \cdot (\jmax+1)! \l(\Cder\r)^{\jmax+1}\\
&\le C_0 \|\mtx{H}\| \l(\Cder \hk h\r)^{\jmax+1} \l(\sum_{i=0}^{\hk}\l(\f{i}{\hk}\r)^{2(\jmax+1)}\r)^{1/2}\cdot \l\|\c^{(l|\jmax|\hk)}\r\|\\
&\le C_0 \|\mtx{H}\| \l(\Cder \hk h\r)^{\jmax+1} \l(1+\hk\int_{x=0}^{1} x^{2\jmax+2}dx \r)^{1/2}\cdot \l\|\c^{(l|\jmax|\hk)}\r\|\\
&\le C_0 \|\mtx{H}\| \l(\Cder \hk h\r)^{\jmax+1} \l(1+\frac{\hk}{2\jmax+3} \r)^{1/2}\cdot \l\|\c^{(l|\jmax|\hk)}\r\|.
\end{align*}
Due to the particular choice of the coefficients of $\c^{(l|\jmax|\hk)}$, we can see that all the terms up to order $j_{\max}$ disappear, and we are left with the remainder term that can be bounded as 
\begin{align*}
\l\|\Ecu{\hat{\Phi}^{(l|\jmax)}(\Y_{0:\hk})}-\mtx{H} \frac{d^l \u}{ dt^l}\r\|\le C_0\|\mtx{H}\| \sqrt{\frac{\hk}{\jmax+3/2}} (\Cder \hk h)^{\jmax+1} \l\|\c^{(l|\jmax|\hk)}\r\|,
\end{align*}
using the assumption that $\hk\ge 2\jmax+3$. The concentration bound now follows directly from this bias bound, \eqref{eqnormZgetbound1} and the fact that the estimator $\hat{\Phi}^{(l|\jmax)}(\Y_{0:\hk})$ has Gaussian distribution with covariance matrix $\l\|\c^{(l|\jmax|\hk)}\r\|\cdot \sigma_Z \cdot \mtx{I}_{d_o}$.
\end{proof}

\begin{proof}[Proof of Lemma \ref{Klimitlem}]
Let $\mc{P}$ denote the space of finite degree polynomials with real coefficients on $[0,1]$. For $a,b\in \mc{P}$, we let $\l<a,b\r>_{\mc{P}}=\int_{x=0}^{1}a(x)b(x)dx$. Then the elements of the matrix $\mtx{K}^{\jmax}$ can be written as 
\[K^{\jmax}_{i,j}=\frac{1}{i+j-1}=\l<x^{i-1},x^{j-1}\r>_{\mc{P}}.\]
If $\mtx{K}^{\jmax}$ would not be invertible, then its rows would be linearly dependent, that is, there would exist a non-zero vector $\vct{\alpha}\in \R^{\jmax+1}$ such that $\sum_{i=1}^{\jmax+1}K^{\jmax}_{i,j}=0$ for every $1\le j\le \jmax+1$. This would imply that
$\l<\sum_{i=1}^{\jmax+1} \alpha_i x^{i-1},x^{j-1}\r>_{\mc{P}}=0$ for every $1\le j\le \jmax+1$, and thus
\[\int_{x=0}^{1}\l(\sum_{i=1}^{\jmax+1} \alpha_i x^{i-1}\r)^2 dx=\l<\sum_{i=1}^{\jmax+1} \alpha_i x^{i-1},\sum_{i=1}^{\jmax+1} \alpha_i x^{i-1}\r>_{\mc{P}}=0.\]
However, this is not possible, since by the fundamental theorem of algebra, $\sum_{i=1}^{\jmax+1} \alpha_i x^{i-1}$ can have at most $\jmax$ roots, so it cannot be zero Lebesgue almost everywhere in $[0,1]$. Therefore $\mtx{K}^{\jmax}$ is invertible.
The result now follows from the continuity of the matrix inverse and the fact that for any $1\le i,j\le \jmax+1$, 
\begin{align*}
\lim_{\hk\to\infty}\l(\frac{\mM (\mM)'}{\hk}\r)_{i,j}&=\lim_{\hk\to \infty}\frac{1}{\hk}\sum_{m=0}^{\hk}\l(\frac{m}{\hk}\r)^{i-1}\l(\frac{m}{\hk}\r)^{j-1}\\
&=\int_{x=0}^{1}x^{i+j-2}dx=\frac{1}{i+j-1}.\qedhere
\end{align*}
\end{proof}

\begin{proof}[Proof of Theorem \ref{thminitialest}]
Let 
\[B_{\mtx{M}}^{(l|\jmax)}:=\sup_{k\in \N, k\ge 2\jmax+3} \CM.\]
Based on Lemma \ref{Klimitlem}, this is finite.
With the choice $\jmax=l$, by \eqref{kmindefeq}, we obtain that
\[
\hk_{\mr{min}}(l,l)=\f{1}{h}\cdot \l(\sigma_Z\sqrt{h}\r)^{\f{1}{l+3/2}} \cdot \Cder^{-\f{l+1}{l+3/2}} \l(\f{\sqrt{d_o \log(d_o+1) (l+3/2)} (l+1/2)}{C_0\|\mtx{H}\| }\r)^{\f{1}{l+3/2}}.
\]
By choosing $s_{\max}^{(l)}$ sufficiently small, we can ensure that for $\sigma_Z\sqrt{h}\le s_{\max}^{(l)}$, $h\hk_{\mr{min}}(l,l)\le T$, and thus by the definition \eqref{koptdefeq}, we have
\begin{align*}&\Bigg|\hk_{\mr{opt}}(l,l)\\
&-\max\l(2l+3, \f{1}{h}\cdot \l(\sigma_Z\sqrt{h}\r)^{\f{1}{l+3/2}} \cdot \Cder^{-\f{l+1}{l+3/2}} \cdot \l(\f{\sqrt{d_o \log(d_o+1) (l+3/2)} (l+1/2)}{C_0\|\mtx{H}\| }\r)^{\f{1}{l+3/2}}\r)\Bigg|< 1.
\end{align*}
By substituting this into \eqref{gfuncdefeq}, and applying some algebra, we obtain that  
\begin{align*}
&g(l,l,\hk_{\mr{opt}}(l,l))=
\frac{C_0\|\mtx{H}\|\Cder^{l+1}}{\sqrt{l+3/2}}  \cdot  (\hk_{\mr{opt}}(l,l) h)
+(\hk_{\mr{opt}}(l,l) h)^{-l-1/2}\sigma_Z\sqrt{h} \sqrt{2d_o \log(d_o+1)}\\
&\le \frac{C_0\|\mtx{H}\|\Cder^{l+1}}{\sqrt{l+3/2}}\\
&\cdot \l[ (2l+4)h+\l(\sigma_Z\sqrt{h}\r)^{\f{1}{l+3/2}} \cdot \Cder^{-\f{l+1}{l+3/2}} \cdot \l(\f{\sqrt{d_o \log(d_o+1) (l+3/2)} (l+1/2)}{C_0\|\mtx{H}\| }\r)^{\f{1}{l+3/2}}\r]\\
&+2\Cder^{\f{(l+1)(l+1/2)}{l+3/2}} \cdot \l(\f{\sqrt{d_o \log(d_o+1) (l+3/2)} (l+1/2)}{C_0\|\mtx{H}\| }\r)^{\f{-l-1/2}{l+3/2}} \l(\sigma_Z\sqrt{h}\r)^{\f{1}{l+3/2}} \sqrt{2d_o \log(d_o+1)},
\end{align*}
and the claim of the theorem now follows by substituting this into Proposition \ref{propestbiasconc1}.
\end{proof}

\section{Conclusion}
In this paper, we have proven consistency and asymptotic efficiency results for MAP estimators for smoothing and filtering a class of partially observed non-linear dynamical systems.
We have also shown that the smoothing and filtering distributions are approximately Gaussian in the low observation noise / high observation frequency regime when the length of the assimilation window is fixed. These results contribute to the statistical understanding of the widely used 4D-Var data assimilation method (\cite{le1986variational, talagrand1987variational}). The precise size of the observation noise $\sigma_Z$ and assimilation step $h$ under which the Gaussian approximation approximately holds, and the MAP estimator is close to the posterior mean is strongly dependent on the model parameters and the size of the assimilation window. However, we have found in simulations on Figure \ref{fig:halfobsdepsigmazsqrth} that even for relatively large values of $\sigma_Z$ and $h$, for large dimensions, and not very short assimilation windows, these approximations seem to be working reasonably well. Besides theoretical importance, the Gaussian approximation of the smoother can be also used to construct the prior (background) distributions for the subsequent intervals in a flow-dependent way, as we have shown in \cite{paulin20174d} for the non-linear shallow-water equations, even for realistic values of $\sigma_Z$ and $h$. These flow-dependent prior distributions can considerably improve filtering accuracy. Going beyond the approximately Gaussian case (for example when $\sigma_Z$, $h$, and $T$ are large, or the system is highly non-linear) in a computationally efficient way  is a challenging problem for future research (see \cite{bocquet2010beyond} for some examples where this situation arises).

\subsubsection*{Acknowledgements}
DP \& AJ were supported by an AcRF tier 2 grant: R-155-000-161-112. AJ is affiliated with the Risk Management Institute, the Center for Quantitative Finance and the OR \& Analytics cluster at NUS. DC was partially supported by the EPSRC grant: EP/N023781/1. We thank the anonymous referees for their valuable comments and suggestions that have substantially improved the paper.

\bibliographystyle{abbrvnat}
\bibliography{References}

\begin{thebibliography}{52}
\providecommand{\natexlab}[1]{#1}
\providecommand{\url}[1]{\texttt{#1}}
\expandafter\ifx\csname urlstyle\endcsname\relax
  \providecommand{\doi}[1]{doi: #1}\else
  \providecommand{\doi}{doi: \begingroup \urlstyle{rm}\Url}\fi

\bibitem[Bannister(2016)]{bannister2016review}
R.~Bannister.
\newblock A review of operational methods of variational and
  ensemble-variational data assimilation.
\newblock \emph{Quarterly Journal of the Royal Meteorological Society}, 2016.

\bibitem[Bardsley et~al.(2014)Bardsley, Solonen, Haario, and
  Laine]{Randomizethenoptimize}
J.~M. Bardsley, A.~Solonen, H.~Haario, and M.~Laine.
\newblock Randomize-then-optimize: a method for sampling from posterior
  distributions in nonlinear inverse problems.
\newblock \emph{SIAM J. Sci. Comput.}, 36\penalty0 (4):\penalty0 A1895--A1910,
  2014.
\newblock ISSN 1064-8275.

\bibitem[Bauer and Fike(1960)]{BauerFike}
F.~L. Bauer and C.~T. Fike.
\newblock Norms and exclusion theorems.
\newblock \emph{Numer. Math.}, 2:\penalty0 137--141, 1960.

\bibitem[Benning and Burger(2018)]{benning2018modern}
M.~Benning and M.~Burger.
\newblock Modern regularization methods for inverse problems.
\newblock \emph{arXiv preprint arXiv:1801.09922}, 2018.

\bibitem[Bhatia and Jain(2009)]{determinantderivative}
R.~Bhatia and T.~Jain.
\newblock Higher order derivatives and perturbation bounds for determinants.
\newblock \emph{Linear Algebra Appl.}, 431\penalty0 (11):\penalty0 2102--2108,
  2009.

\bibitem[Blomker et~al.(2013)Blomker, Law, Stuart, and Zygalakis]{blomker}
D.~Blomker, K.~Law, A.~Stuart, and K.~Zygalakis.
\newblock Accuracy and stability of the continuous time {3DVAR} filter for the
  {N}avier-{S}tokes equation.
\newblock \emph{Nonlinearity}, 26:\penalty0 2193--2227, 2013.

\bibitem[Bocquet et~al.(2010)Bocquet, Pires, and Wu]{bocquet2010beyond}
M.~Bocquet, C.~A. Pires, and L.~Wu.
\newblock Beyond {G}aussian statistical modeling in geophysical data
  assimilation.
\newblock \emph{Monthly Weather Review}, 138\penalty0 (8):\penalty0 2997--3023,
  2010.

\bibitem[Boucheron et~al.(2013)Boucheron, Lugosi, and Massart]{BoLuMa2013}
S.~Boucheron, G.~Lugosi, and P.~Massart.
\newblock \emph{Concentration inequalities}.
\newblock Oxford University Press, Oxford, 2013.
\newblock A nonasymptotic theory of independence, With a foreword by Michel
  Ledoux.

\bibitem[Bubeck(2015)]{bubeck2015convex}
S.~Bubeck.
\newblock Convex optimization: Algorithms and complexity.
\newblock \emph{Foundations and Trends{\textregistered} in Machine Learning},
  8\penalty0 (3-4):\penalty0 231--357, 2015.

\bibitem[Cotter et~al.(2009)Cotter, Dashti, Robinson, and
  Stuart]{Cotterfunctions}
S.~L. Cotter, M.~Dashti, J.~C. Robinson, and A.~M. Stuart.
\newblock Bayesian inverse problems for functions and applications to fluid
  mechanics.
\newblock \emph{Inverse Problems}, 25\penalty0 (11):\penalty0 115008, 43, 2009.
\newblock ISSN 0266-5611.

\bibitem[Cotter et~al.(2012)Cotter, Dashti, and Stuart]{Cotterrandwalk}
S.~L. Cotter, M.~Dashti, and A.~M. Stuart.
\newblock Variational data assimilation using targetted random walks.
\newblock \emph{Internat. J. Numer. Methods Fluids}, 68\penalty0 (4):\penalty0
  403--421, 2012.
\newblock ISSN 0271-2091.

\bibitem[Cotter et~al.(2013)Cotter, Roberts, Stuart, and
  White]{MCMCforfunctions}
S.~L. Cotter, G.~O. Roberts, A.~M. Stuart, and D.~White.
\newblock M{CMC} methods for functions: modifying old algorithms to make them
  faster.
\newblock \emph{Statist. Sci.}, 28\penalty0 (3):\penalty0 424--446, 2013.
\newblock ISSN 0883-4237.

\bibitem[Crisan and Rozovskii(2011)]{dan}
D.~Crisan and B.~Rozovskii.
\newblock \emph{The Oxford Handbook of Nonlinear Filtering}.
\newblock OUP, Oxford, 2011.

\bibitem[Cui et~al.(2016)Cui, Law, and Marzouk]{cui2016dimension}
T.~Cui, K.~J. Law, and Y.~M. Marzouk.
\newblock Dimension-independent likelihood-informed mcmc.
\newblock \emph{Journal of Computational Physics}, 304:\penalty0 109--137,
  2016.

\bibitem[Dashti and Stuart(2017)]{dashti2017bayesian}
M.~Dashti and A.~M. Stuart.
\newblock The {B}ayesian approach to inverse problems.
\newblock \emph{Handbook of Uncertainty Quantification}, pages 311--428, 2017.

\bibitem[Dashti et~al.(2013)Dashti, Law, Stuart, and Voss]{DashtiLawStuart}
M.~Dashti, K.~J.~H. Law, A.~M. Stuart, and J.~Voss.
\newblock M{AP} estimators and their consistency in {B}ayesian nonparametric
  inverse problems.
\newblock \emph{Inverse Problems}, 29\penalty0 (9):\penalty0 095017, 27, 2013.

\bibitem[Dimet and Shutyaev(2005)]{dimet2005deterministic}
F.-X.~L. Dimet and V.~Shutyaev.
\newblock On deterministic error analysis in variational data assimilation.
\newblock \emph{Nonlinear Processes in Geophysics}, 12\penalty0 (4):\penalty0
  481--490, 2005.

\bibitem[Dumer(2007)]{Coveringspheres}
I.~Dumer.
\newblock Covering spheres with spheres.
\newblock \emph{Discrete Comput. Geom.}, 38\penalty0 (4):\penalty0 665--679,
  2007.

\bibitem[Dunlop and Stuart(2016)]{dunlop2016map}
M.~M. Dunlop and A.~M. Stuart.
\newblock Map estimators for piecewise continuous inversion.
\newblock \emph{Inverse Problems}, 32\penalty0 (10):\penalty0 105003, 2016.

\bibitem[Fornberg(1988)]{Finitediffformulas}
B.~Fornberg.
\newblock Generation of finite difference formulas on arbitrarily spaced grids.
\newblock \emph{Math. Comp.}, 51\penalty0 (184):\penalty0 699--706, 1988.

\bibitem[Gin\'e and Nickl(2016)]{Nicklinfinitedimstatmodels}
E.~Gin\'e and R.~Nickl.
\newblock \emph{Mathematical foundations of infinite-dimensional statistical
  models}.
\newblock Cambridge Series in Statistical and Probabilistic Mathematics, [40].
  Cambridge University Press, New York, 2016.
\newblock ISBN 978-1-107-04316-9.

\bibitem[Hayden et~al.(2011)Hayden, Olson, and Titi]{DiscreteHaydenOlsonTiti}
K.~Hayden, E.~Olson, and E.~S. Titi.
\newblock Discrete data assimilation in the {L}orenz and 2{D} {N}avier-{S}tokes
  equations.
\newblock \emph{Phys. D}, 240\penalty0 (18):\penalty0 1416--1425, 2011.
\newblock ISSN 0167-2789.

\bibitem[Helin and Burger(2015)]{HelinBurgerMAP}
T.~Helin and M.~Burger.
\newblock Maximum {\it a posteriori} probability estimates in
  infinite-dimensional {B}ayesian inverse problems.
\newblock \emph{Inverse Problems}, 31\penalty0 (8):\penalty0 085009, 22, 2015.
\newblock ISSN 0266-5611.

\bibitem[Kalnay(2003)]{kalnay}
E.~Kalnay.
\newblock \emph{Atmospheric Modeling, Data Assimilation and Predictability}.
\newblock Cambridge University Press, Cambridge, 2003.

\bibitem[Kekkonen et~al.(2016)Kekkonen, Lassas, and
  Siltanen]{PosteriorconsistencyHanneLassasSamuli}
H.~Kekkonen, M.~Lassas, and S.~Siltanen.
\newblock Posterior consistency and convergence rates for {B}ayesian inversion
  with hypoelliptic operators.
\newblock \emph{Inverse Problems}, 32\penalty0 (8):\penalty0 085005, 31, 2016.
\newblock ISSN 0266-5611.

\bibitem[Lasserre(2010)]{lasserremoments}
J.~B. Lasserre.
\newblock \emph{Moments, positive polynomials and their applications}, volume~1
  of \emph{Imperial College Press Optimization Series}.
\newblock Imperial College Press, London, 2010.

\bibitem[Law et~al.(2015)Law, Stuart, and Zygalakis]{Dataassimilation}
K.~Law, A.~Stuart, and K.~Zygalakis.
\newblock \emph{Data assimilation}, volume~62 of \emph{Texts in Applied
  Mathematics}.
\newblock Springer, Cham, 2015.
\newblock A mathematical introduction.

\bibitem[Law et~al.(2016)Law, Sanz-Alonso, Shukla, and Stuart]{law2016filter}
K.~Law, D.~Sanz-Alonso, A.~Shukla, and A.~Stuart.
\newblock Filter accuracy for the lorenz 96 model: Fixed versus adaptive
  observation operators.
\newblock \emph{Physica D: Nonlinear Phenomena}, 325:\penalty0 1--13, 2016.

\bibitem[Le~Dimet and Talagrand(1986)]{le1986variational}
F.-X. Le~Dimet and O.~Talagrand.
\newblock Variational algorithms for analysis and assimilation of
  meteorological observations: theoretical aspects.
\newblock \emph{Tellus A: Dynamic Meteorology and Oceanography}, 38\penalty0
  (2):\penalty0 97--110, 1986.

\bibitem[Le~Dimet et~al.(2002)Le~Dimet, Navon, and Daescu]{le2002second}
F.-X. Le~Dimet, I.~M. Navon, and D.~N. Daescu.
\newblock Second-order information in data assimilation.
\newblock \emph{Monthly Weather Review}, 130\penalty0 (3):\penalty0 629--648,
  2002.

\bibitem[Lorenz(1996)]{lorenz1996predictability}
E.~N. Lorenz.
\newblock Predictability: A problem partly solved.
\newblock In \emph{Proc. Seminar on predictability}, volume~1, 1996.

\bibitem[Majda and Harlim(2012)]{MajdaHarlim}
A.~J. Majda and J.~Harlim.
\newblock \emph{Filtering complex turbulent systems}.
\newblock Cambridge University Press, Cambridge, 2012.
\newblock ISBN 978-1-107-01666-8.

\bibitem[Majda et~al.(2010)Majda, Harlim, and
  Gershgorin]{MajdaHarlimGershgorin}
A.~J. Majda, J.~Harlim, and B.~Gershgorin.
\newblock Mathematical strategies for filtering turbulent dynamical systems.
\newblock \emph{Discrete Contin. Dyn. Syst.}, 27\penalty0 (2):\penalty0
  441--486, 2010.
\newblock ISSN 1078-0947.

\bibitem[Monard et~al.(2017)Monard, Nickl, and Paternain]{monard2017efficient}
F.~Monard, R.~Nickl, and G.~P. Paternain.
\newblock Efficient nonparametric {B}ayesian inference for {X}-{R}ay
  transforms.
\newblock \emph{arXiv preprint arXiv:1708.06332}, 2017.

\bibitem[Navon(2009)]{navon2009data}
I.~M. Navon.
\newblock Data assimilation for numerical weather prediction: a review.
\newblock In \emph{Data assimilation for atmospheric, oceanic and hydrologic
  applications}, pages 21--65. Springer, 2009.

\bibitem[Nickl(2017)]{nickl2017schrodinger}
R.~Nickl.
\newblock {B}ernstein-von {M}ises theorems for statistical inverse problems i:
  {S}chr$\backslash$" odinger equation.
\newblock \emph{arXiv preprint arXiv:1707.01764}, 2017.

\bibitem[Paulin et~al.(2017)Paulin, Jasra, Beskos, and Crisan]{paulin20174d}
D.~Paulin, A.~Jasra, A.~Beskos, and D.~Crisan.
\newblock A 4{D}-{V}ar method with flow-dependent background covariances for
  the shallow-water equations.
\newblock \emph{arXiv preprint arXiv:1710.11529}, 2017.

\bibitem[Paulin et~al.(2018)Paulin, Jasra, Crisan, and
  Beskos]{ConcentrationProperties}
D.~Paulin, A.~Jasra, D.~Crisan, and A.~Beskos.
\newblock On concentration properties of partially observed chaotic systems.
\newblock \emph{Advances in Applied Probability, to appear}, 2018.

\bibitem[Pillai et~al.(2014)Pillai, Stuart, and Thi\'ery]{nateshstuartthiery}
N.~S. Pillai, A.~M. Stuart, and A.~H. Thi\'ery.
\newblock Noisy gradient flow from a random walk in {H}ilbert space.
\newblock \emph{Stoch. Partial Differ. Equ. Anal. Comput.}, 2\penalty0
  (2):\penalty0 196--232, 2014.
\newblock ISSN 2194-0401.

\bibitem[Prajna et~al.(2002)Prajna, Papachristodoulou, and Parrilo]{sostools}
S.~Prajna, A.~Papachristodoulou, and P.~A. Parrilo.
\newblock Introducing sostools: A general purpose sum of squares programming
  solver.
\newblock In \emph{Decision and Control, 2002, Proceedings of the 41st IEEE
  Conference on}, volume~1, pages 741--746. IEEE, 2002.

\bibitem[Roberts and Rosenthal(2004)]{Robertsgeneral}
G.~O. Roberts and J.~S. Rosenthal.
\newblock General state space {M}arkov chains and {MCMC} algorithms.
\newblock \emph{Probab. Surv.}, 1:\penalty0 20--71, 2004.

\bibitem[Sanz-Alonso and Stuart(2015)]{AlonsoStuartLongtime}
D.~Sanz-Alonso and A.~M. Stuart.
\newblock Long-time asymptotics of the filtering distribution for partially
  observed chaotic dynamical systems.
\newblock \emph{SIAM/ASA J. Uncertain. Quantif.}, 3\penalty0 (1):\penalty0
  1200--1220, 2015.

\bibitem[Stuart and Humphries(1996)]{Stuartdynamicalsystems}
A.~M. Stuart and A.~R. Humphries.
\newblock \emph{Dynamical systems and numerical analysis}, volume~2 of
  \emph{Cambridge Monographs on Applied and Computational Mathematics}.
\newblock Cambridge University Press, Cambridge, 1996.

\bibitem[Talagrand and Courtier(1987)]{talagrand1987variational}
O.~Talagrand and P.~Courtier.
\newblock Variational assimilation of meteorological observations with the
  adjoint vorticity equation. i: Theory.
\newblock \emph{Quarterly Journal of the Royal Meteorological Society},
  113\penalty0 (478):\penalty0 1311--1328, 1987.

\bibitem[Temam(1997)]{Temam}
R.~Temam.
\newblock \emph{Infinite-dimensional dynamical systems in mechanics and
  physics}, volume~68 of \emph{Applied Mathematical Sciences}.
\newblock Springer-Verlag, New York, second edition, 1997.

\bibitem[Tropp(2015)]{tropp2015introduction}
J.~A. Tropp.
\newblock An introduction to matrix concentration inequalities.
\newblock \emph{Foundations and Trends in Machine Learning}, 8\penalty0
  (1-2):\penalty0 1--230, 2015.

\bibitem[T{\"u}t{\"u}nc{\"u} et~al.(2003)T{\"u}t{\"u}nc{\"u}, Toh, and
  Todd]{tutuncu2003solving}
R.~H. T{\"u}t{\"u}nc{\"u}, K.-C. Toh, and M.~J. Todd.
\newblock Solving semidefinite-quadratic-linear programs using sdpt3.
\newblock \emph{Mathematical programming}, 95\penalty0 (2):\penalty0 189--217,
  2003.

\bibitem[Villani(2009)]{Villanioptimaloldnew}
C.~Villani.
\newblock \emph{Optimal transport}, volume 338 of \emph{Grundlehren der
  Mathematischen Wissenschaften [Fundamental Principles of Mathematical
  Sciences]}.
\newblock Springer-Verlag, Berlin, 2009.
\newblock Old and new.

\bibitem[Vollmer(2013)]{Vollmerpostconsistency}
S.~J. Vollmer.
\newblock Posterior consistency for {B}ayesian inverse problems through
  stability and regression results.
\newblock \emph{Inverse Problems}, 29\penalty0 (12):\penalty0 125011, 32, 2013.
\newblock ISSN 0266-5611.

\bibitem[Vollmer(2015)]{VollmerdimindepMCMC}
S.~J. Vollmer.
\newblock Dimension-independent {MCMC} sampling for inverse problems with
  non-{G}aussian priors.
\newblock \emph{SIAM/ASA J. Uncertain. Quantif.}, 3\penalty0 (1):\penalty0
  535--561, 2015.
\newblock ISSN 2166-2525.

\bibitem[Wang et~al.(2018)Wang, Bui-Thanh, and Ghattas]{RandomizeMAP}
K.~Wang, T.~Bui-Thanh, and O.~Ghattas.
\newblock A {R}andomized {M}aximum {A} {P}osteriori {M}ethod for {P}osterior
  {S}ampling of {H}igh {D}imensional {N}onlinear {B}ayesian {I}nverse
  {P}roblems.
\newblock \emph{SIAM J. Sci. Comput.}, 40\penalty0 (1):\penalty0 A142--A171,
  2018.
\newblock ISSN 1064-8275.

\bibitem[Yao et~al.(2016)Yao, Hu, and Li]{TVGaussian}
Z.~Yao, Z.~Hu, and J.~Li.
\newblock A {TV}-{G}aussian prior for infinite-dimensional {B}ayesian inverse
  problems and its numerical implementations.
\newblock \emph{Inverse Problems}, 32\penalty0 (7):\penalty0 075006, 19, 2016.
\newblock ISSN 0266-5611.

\end{thebibliography}
\appendix
\section{Appendix}
\subsection{Proof of preliminary results}\label{Secproofpreliminaryresults}
\begin{proof}[Proof of Lemma \ref{Dinormbndlemma}]
To prove \eqref{uderboundeq}, it suffices to first verify \eqref{uderboundeq} and for $i=0$ and $i=1$, and then use induction and the recursion formula \eqref{udereq} for $i\ge 2$. For \eqref{ugradboundeqgen}, by taking the $k$th derivative of \eqref{udereq}, we obtain that
\begin{align}\label{udereqkder}\J^k_{\v}\l(\D^i \v\r)=-\mtx{A}\cdot \J^k_{\v} \l(\D^{i-1} \v\r)-\sum_{j=0}^{i-1} {i-1 \choose j}\sum_{l=0}^k  {k \choose l}\mtx{B}\l(\J^{l}_{\v}\l(\D^j \v\r),\J^{k-l}_{\v}\l(\D^{i-1-j} \v\r)\r).
\end{align}
For $k=1$, \eqref{ugradboundeqgen} can be verified by first checking it for $i=0$ and $i=1$, and then using mathematical induction and \eqref{udereq} for $i\ge 2$. Suppose that \eqref{ugradboundeqgen} holds for $k=1\ldots, k'-1$, then by mathematical induction and \eqref{udereqkder}, we only need to show that
\begin{align*}\l(C_{\mtx{J}}^{(k')} \r)^i \cdot i!&\ge \|\mtx{A}\|\l(C_{\mtx{J}}^{(k')}  \r)^{i-1} (i-1)!
\\&+\sum_{j=0}^{i-1}{i-1 \choose j}\sum_{l=1}^{k'-1}    {k' \choose l}  j! (i-1-j)!\|\mtx{B}\| \l(C_{\mtx{J}}^{(l)} \r)^j
\l(C_{\mtx{J}}^{(k'-l)} \r)^{i-1-j}\\
&+\sum_{j=0}^{i-1}{i-1 \choose j} j! (i-1-j)! \cdot 2\|\mtx{B}\| C_0\l(\Cder\r)^j
\l(C_{\mtx{J}}^{(k')} \r)^{i-1-j},
\end{align*}
which is straightforward to check since $\sum_{l=0}^{k'}{k' \choose l}=2^{k'}$ and ${i-1 \choose j} j! (i-1-j)!=(i-1)!$.
\end{proof}

\begin{proof}[Proof of Lemma \ref{partderbndlem}]
Notice that the derivative  $\frac{d}{ d t}(\J^k \Psi_{t}(\vct{v}))$ 
can be rewritten by exchanging the order of derivation (which can be justified by the Taylor series expansion and the bounds \eqref{ugradboundeqgen}) as
\begin{align}\nonumber\frac{d}{ d t}(\J^k \Psi_{t}(\vct{v}))&=\J^k_{\v} \l(-\mtx{A} \Psi_{t}(\vct{v}) -\mtx{B}(\Psi_{t}(\vct{v}),\Psi_{t}(\vct{v}))+\vct{f}\r)\\
\label{dtJkPhitveq}&=-\mtx{A} \J^k\Psi_{t}(\vct{v}) - \sum_{l=0}^k {k \choose l}\mtx{B}(\J^l \Psi_{t}(\vct{v}),\J^{k-l}\Psi_{t}(\vct{v})).
\end{align}
For $k=1$, the above equation implies that
\[\frac{d}{ d t}\|\J \Psi_{t}(\vct{v})\|\le (\|\mtx{A}\|+2\|\mtx{B}\|R)\|\J \Psi_{t}(\vct{v})\|, \]
thus using the fact that $\|\J \Psi_0(\vct{v})\|=\|\J_{\v} (\vct{v})\|=1$, by Gr\"{o}nwall's lemma, we have \[\|\J \Psi_{t}(\vct{v})\|\le \exp((\|\mtx{A}\|+2\|\mtx{B}\|R)t)\text{ for any }t\ge 0.\]
Now we are going to show that
\begin{equation}\label{Jknormbndeq}\|\J^k \Psi_{t}(\vct{v})\|\le \exp\l(D_{\J}^{(k)}t\r)\text{ for any }\v\in \BR, t\ge 0, k\in \Z_+.
\end{equation}
Indeed, this was shown for $k=1$ above, and for $k\ge 2$, from \eqref{dtJkPhitveq}, it follows that
\begin{align*}
\frac{d}{ d t}\|\J^k \Psi_{t}(\vct{v})\|\le \l(\|\mtx{A}\|+2\|\mtx{B}\|R\r) \|\J^k\Psi_{t}(\vct{v})\|+ \|\mtx{B}\| \sum_{l=1}^{k-1} {k \choose l}\|\J^l \Psi_{t}(\vct{v})\|\|\J^{k-l} \Psi_{t}(\vct{v})\|,
\end{align*}
thus \eqref{Jknormbndeq} can be proven by mathematical induction and Gr\"{o}nwall's lemma. This implies in particular our first claim, \eqref{MkTdef}. Our second claim, \eqref{tMkTdef} follows by the fact that $\Phi_{t}(\vct{v})=\mtx{H}\Psi_{t}(\vct{v})$ is a linear transformation.
\end{proof}

\begin{proof}[Proof of Lemma \ref{Bknormboundlemma}]
Let $(\J \Phi_{t_i}(\u))_{\cdot, j}$ denote the $j$th column of the Jacobian, and 
let $\tilde{\vZ}_i^j:=\vZ_i^j/\sigma_Z$ (which is a standard normal random variable).
Then we can write $\vct{B}_k$ as 
\[\vct{B}_k=\sum_{i=0}^k\sum_{j=1}^{d_o} \sigma_Z\cdot (\J \Phi_{t_i}(\u))_{\cdot, j} \cdot \tilde{\vZ}_i^j.\]
This is of the same form as in equation (4.1.2) of Theorem 4.1.1 of \citet*{tropp2015introduction}. 
Since $\|\J \Phi_{t_i}(\u)\|\le \hM_1(T)$, one can see that we also have $\|(\J \Phi_{t_i}(\u))_{\cdot, j}\|\le \hM_1(T)$, and thus the variance statistics $v(\mtx{Z})$ of Theorem 4.1.1 can be bounded as $v(\mtx{Z})\le \sigma_Z^2 \hM_1(T)^2  (k+1)d_o$. The result now follows from equation (4.1.6) of \citet*{tropp2015introduction} (with $d_1=1$ and $d_2=d$).
\end{proof}
\begin{proof}[Proof of Lemma \ref{Akeigboundlemma}]
First, note that from Assumption \ref{assgauss1} and looking at the Taylor expansion near $\u$ it follows that the first term in the definition of $\mtx{A}_k$ satisfies that
\[\lambda_{\min}\l(\sum_{i=0}^k \J \Phi_{t_i}(\u)' \J \Phi_{t_i}(\u)\r)\ge \frac{c(\u,T)}{h}.\]
We can rewrite the second term in the definition \eqref{eqAkdef} as
\begin{equation}\label{eqAk1bnd}\J^2 \Phi_{t_i}(\u)[\cdot,\cdot,\vZ_i]=\sum_{j=1}^{d_o}H \Phi_{t_i}^j(\u)\cdot \sigma_Z\cdot  \tilde{\vZ}_i^j,\end{equation}
where $H \Phi_{t_i}^j(\u)$ denotes the Hessian of the function $\Phi_{t_i}^j$ at point $\u$.
This is of the same form as in equation (4.1.2) of Theorem 4.1.1 of \citet*{tropp2015introduction}. 
Since $\|\J^2 \Phi_{t_i}(\u)\|\le \hM_2(T)$, one can see that we also have $\|H \Phi_{t_i}^j(\u)\|\le \hM_2(T)$, and thus the variance statistics $v(\mtx{Z})$ of Theorem 4.1.1 can be bounded as $v(\mtx{Z})\le \sigma_Z^2 \hM_2(T)^2  (k+1)d_o$, and thus for any $t\ge 0$, we have
\begin{equation}\label{eqAk2bnd}\PP\l(\l\|\sum_{i=0}^k \J^2 \Phi_{t_i}(\u)[\cdot,\cdot,\vZ_i]\r\|\ge t\r)\le 2d\exp\l(-\frac{t^2}{\sigma_Z^2 \hM_2(T)^2  (k+1)d_o}\r).\end{equation}
By the Bauer-Fike theorem (see \cite{BauerFike}), we have 
\[\lambda_{\min}(\mtx{A}_k)\ge \lambda_{\min}\l(\sum_{i=0}^k \J \Phi_{t_i}(\u)' \J \Phi_{t_i}(\u)\r)-\l\|\sum_{i=0}^k \J^2 \Phi_{t_i}(\u)[\cdot,\cdot,\vZ_i]\r\|,\]
and 
\[\|\mtx{A}_k\|\le \hM_1(T)^2\cdot \frac{\olT}{h}+\l\|\sum_{i=0}^k \J^2 \Phi_{t_i}(\u)[\cdot,\cdot,\vZ_i]\r\|, \]
so \eqref{eqlambdaminAk1bnd} follows by the bounds \eqref{eqAk1bnd} and \eqref{eqAk2bnd}, and \eqref{eqlambdaminAk2bnd} follows  by rearrangement.
\end{proof}

\begin{proof}[Proof of Proposition \ref{Propcheckassumptionsgauss}]
The proof is quite similar to the proof of Proposition 4.1 of \citet*{ConcentrationProperties}. With a slight modification of that argument, we will show that for any $\delta\in [0,\tilde{h})$,
\begin{equation}\label{Phiihplusdeltaeq}\sum_{i=0}^{j} \|\Phi_{i\tilde{h}+\delta}(\v)-\Phi_{i\tilde{h}+\delta}(\u)\|^2\ge c'(\u,\tilde{h})\|\v-\u\|^2,\end{equation}
for some constant $c'(\u,\tilde{h})>0$ independent of $\delta$, that is monotone increasing in $\tilde{h}$ for $0<\tilde{h}\le \tilde{h}_{\max}$. This result allows us to decouple the summation in \eqref{eqassgauss1} into sets of size $j+1$ as follows. Let $h_0:=\min\l(\frac{T}{j+1},\tilde{h}_{\max}\r)$, and set $\tilde{h}:=\lfloor h_0/h\rfloor \cdot h$. Then for $h\le h_0$, we have $\lfloor h_0/h\rfloor > \frac{h_0}{2h}$ and $\tilde{h}> h_0/2$, so using \eqref{Phiihplusdeltaeq}, we obtain that
\[\sum_{i=0}^{k} \|\Phi_{t_i}(\v)-\Phi_{t_i}(\u)\|^2\ge
\sum_{l=0}^{\lfloor h_0/h\rfloor-1} \sum_{i=0}^{j} \l\|\Phi_{i\tilde{h}+lh}(\v)-\Phi_{i\tilde{h}+lh}(\u)\r\|^2\ge  \frac{h_0}{2h}\cdot c'(\u,h_0/2) \|\v-\u\|^2.\]
Thus Assumption \ref{assgauss1} holds with $h_{\min}(\u,T):=h_0/2$ and $c(\u,T):=c'(\u,h_0/2)\cdot h_0/2$.

To complete the proof, we will now show \eqref{Phiihplusdeltaeq}. Using inequality \eqref{uderboundeq}, we can see that the Taylor expansion
\[\Phi_t(\v)=\sum_{i=0}^{\infty} \frac{\mtx{H} \D^i \v \cdot t^i}{i!}\]
is valid for times $0\le t< \Cder^{-1}$.  Based on this expansion, assuming that $i\tilde{h}+\delta<\Cder^{-1}$, $\mtx{H} \D^i \v$ can be approximated by a finite difference formula depending on the values of $\Phi_\delta(\v), \Phi_{\tilde{h}+\delta}(\v), \ldots, \Phi_{i\tilde{h}+\delta}$, with error of $O(\tilde{h})$. This finite difference formula will be denoted as
\begin{equation}\label{hatPhideltadef}
\hat{\Phi}^{(i,\delta)}(\v):=\frac{\sum_{l=0}^{i}a_l^{(i,\delta)} \Phi_{l\tilde{h}+\delta}(\v)}{\tilde{h}^{i}}.
\end{equation}
The coefficients $a_l^{(i,\delta)}$ are explicitly defined in \citet*{Finitediffformulas}, and they only depend on $i$ and the ratio $\delta/\tilde{h}$. Based on the definition of these coefficients on page 700 of \citet*{Finitediffformulas}, we can see that
\begin{equation}
\label{oladefeq}
\ol{a}:=\sup_{\tilde{h}>0, \delta\in [0, \tilde{h})}\max_{0\le i\le j, 0\le l\le i} \l|a_l^{(i,\delta)}\r|<\infty,
\end{equation}
i.e. they can be bounded by a finite constant independently of $\delta$. By Taylor's expansion of the terms $\Phi_{l\tilde{h}+\delta}(\v)$ around time point 0, for $l\tilde{h}+\delta< \Cder^{-1}$, we have
\begin{align*}
\Phi_{l\tilde{h}+\delta}(\v)&=\sum_{m=0}^{\infty} \mtx{H} \D^m \v \cdot \frac{(l\tilde{h}+\delta)^m}{m!},\text{ and thus }\\
\hat{\Phi}^{(i,\delta)}(\v)&=\frac{1}{\tilde{h}^i}\cdot \sum_{l=0}^{i}a_l^{(i,\delta)} \sum_{m=0}^{\infty} \mtx{H} \D^m \v\cdot \frac{(l\tilde{h}+\delta)^m}{m!}=\frac{1}{\tilde{h}^i}\sum_{m=0}^{\infty}\tilde{h}^m b_m^{(i,\delta)} \mtx{H} \D^m \v,
\end{align*}
with $b_{m}^{(i,\delta)}:=\frac{1}{m!}\cdot \sum_{l=0}^{i} a_l^{(i,\delta)}(l+\delta/\tilde{h})^m$. Due to the particular choice of the constants $a_l^{(i,\delta)}$, we have $b_{m}^{(i,\delta)}=0$ for $0\le m<i$ and $b_{m}^{(i,\delta)}=1$ for $m=i$. Based on this, we can write the difference between the approximation \eqref{hatPhideltadef} and the derivative explicitly as
\[\hat{\Phi}^{(i,\delta)}(\v)-\mtx{H} \D^i \v=\tilde{h}  \l(\sum_{m=i+1}^{\infty}\tilde{h}^{m-i-1} \cdot b_{m}^{(i,\delta)}   \cdot \mtx{H} \D^m \v \r).\]
Let us denote $\tilde{\Phi}^{(i,\delta)}(\v,\tilde{h}):=\sum_{m=i+1}^{\infty}\tilde{h}^{m-i-1} \cdot b_{m}^{(i,\delta)}   \cdot \mtx{H} \D^m \v$.
Using inequality \eqref{ugradboundeqgen}, and the bound $|b_{m}^{(i,\delta)}|\le \frac{\ol{a}\cdot (i+1)^{m+1}}{m!}$, we have that for $0\le i\le j$, $\tilde{h}\le \frac{1}{2(j+1)\CJ^{(1)}}$,
\begin{align*}
\|\mtx{J}_{\v}\tilde{\Phi}^{(i,\delta)}(\v,\tilde{h})\|&\le \|\mtx{H}\|\cdot\sum_{m=i+1}^{\infty} \tilde{h}^{m-i-1} |b_{m}^{(i,\delta)}| \l(\CJ^{(1)}\r)^{m} m!\\
&\le \|\mtx{H}\|\cdot \frac{\ol{a}(i+1)}{\tilde{h}^{i+1}}\sum_{m=i+1}^{\infty}\l((i+1)\tilde{h}\CJ^{(1)}\r)^m\le
2\|\mtx{H}\| \ol{a}(i+1) \l((i+1)\CJ^{(1)}\r)^{i+1}.
\end{align*}
Denote $C_{\mathrm{Lip}}:=2\|\mtx{H}\|\ol{a} \cdot \max_{0\le i\le j} \l((i+1) \cdot \l((i+1)\CJ^{(1)}\r)^{i+1}\r)$, then we know that for every $0\le i\le j$, $\tilde{h}<\frac{1}{2 (j+1) \CJ^{(1)}}$, the functions $\tilde{\Phi}^{(i)}(\v,\tilde{h})$ are $C_{\mathrm{Lip}}$ - Lipschitz in $\v$ with respect to the $\|\cdot \|$ norm, and thus for every $\v\in \BR$,
\begin{equation}\label{eqdiffbnddelta}
\l\|\l[\mtx{H}\D^i \v - \mtx{H} \D^i \u\r]-\l[\hat{\Phi}^{(i,\delta)}(\v)-\hat{\Phi}^{(i,\delta)}(\u)\r]\r\|\le \tilde{h} C_{\mathrm{Lip}} \|\v-\u\|.
\end{equation}
By Assumption \ref{assder}, and the boundedness of $\BR$, it follows that there is constant $C_{\D}(\u,T)>0$ such that for every $\v\in \BR$,
\begin{equation}\label{eqHDilowerbnd}
\sum_{i=0}^{j}\l\|\mtx{H}\D^i \v - \mtx{H} \D^i \u\r\|^2\ge C_{\D}(\u,T)\|\v-\u\|^2.
\end{equation}
From equations \eqref{eqdiffbnddelta}, \eqref{eqHDilowerbnd}, and the boundedness of $\BR$, it follows that there is a constant $C_{j}(\u,T,\tilde{h})>0$ that is non-decreasing in $\tilde{h}$ such that for every $\v\in \BR$,
\begin{equation}
\sum_{i=0}^{j}\l\|\hat{\Phi}^{(i,\delta)}(\v)-\hat{\Phi}^{(i,\delta)}(\u)\r\|^2\ge C_{j}(\u,T,\tilde{h})\|\v-\u\|^2.
\end{equation}
By the definitions \eqref{hatPhideltadef} and \eqref{oladefeq}, it follows that
\begin{equation}
\ol{a}\cdot \max\l(\frac{1}{\tilde{h}^{j}},1\r)\cdot \sum_{i=0}^{j}\l\|\Phi_{l\tilde{h}+\delta}(\v)-\Phi_{l\tilde{h}+\delta}(\u)\r\|^2\ge \sum_{i=0}^{j}\l\|\hat{\Phi}^{(i,\delta)}(\v)-\hat{\Phi}^{(i,\delta)}(\u)\r\|^2,
\end{equation}
and thus \eqref{Phiihplusdeltaeq} follows by rearrangement.
\end{proof}

The following lemma bounds the number of balls of radius $\delta$ required to cover a $d$-dimensional unit ball. It will be used in the proofs of Propositions \ref{lsmlowerboundprop}, \ref{lsmlsmGdiffprop} and \ref{lsmlsmGgraddiffprop}.
\begin{lem}\label{coverspherelemma}
For any $d\ge 1$, $0<\delta\le 1$, a $d$-dimensional unit ball can be fully covered by the union of $c(d)\cdot (\frac{1}{\delta})^d$ balls of radius $\delta$, where
\[c(1):=2, \quad c(2):=6, \quad \text{ and }\quad c(d):=d\log(d)\cdot \l(\frac{1}{2}+\frac{2\log(\log(d))}{\log(d)}+\frac{5}{\log(d)}\r) \text{ for }d\ge 3.\]
\end{lem}
\begin{proof}
For $d\ge 3$, this follows from Theorem 1 of \citet*{Coveringspheres}. For $d=1$, the result follows from the fact that $\frac{1}{\delta}+1$ intervals suffice, and $\frac{1}{\delta}+1\le \frac{2}{\delta}$. For $d=2$, we know that the circles of radius $\delta$ contains a square of edge length $\sqrt{2}\delta$, and in order to cover a square of edge length 2 containing the unit ball, it suffices to use 
$\lceil 2/(\sqrt{2}\delta)\rceil^2\le \l(\frac{\sqrt{2}}{\delta}+1\r)^2< \frac{6}{\delta^2}$ squares of edge length $\sqrt{2}\delta$, thus the result follows.
\end{proof}

Our next lemma shows some concentration inequalities that will be used in the proof of our propositions. It is a reformulation of Corollary 13.2 and Theorem 5.8 of \cite{BoLuMa2013} to our setting.
\begin{lem}\label{supempprocconclemma}
For every $l\in \N$, define the sets
\begin{equation}\label{setTdefeq}
\T_l:=\{(r,\s_1,\ldots,\s_l)\in [0,2R]\times \B_1^{l}: \u+r\s_1\in \BR\}, \quad \olTset_l:=\BR\times \B_1^{l}.
\end{equation}
For any two elements $(r,\s_1,\ldots,\s_l), (r,\s_1',\ldots,\s_l')\in \T_l$, we define the distance 
\begin{equation}\label{dldefeq}
d_l((r,\s_1,\ldots,\s_l),(r,\s_1',\ldots,\s_l')):=\frac{|r-r'|}{2R}+\sum_{i=0}^l \|\s_i-\s_i'\|.
\end{equation}
Similarly, for any two elements $(\v,\s_1,\ldots,\s_l), (\v',\s_1',\ldots,\s_l')\in \olTset_l$, we define
\begin{equation}\label{oldldefeq}
\ol{d}_l((\v,\s_1,\ldots,\s_l),(\v',\s_1',\ldots,\s_l')):=\frac{\|\v-\v'\|}{R}+\sum_{i=0}^l \|\s_i-\s_i'\|.
\end{equation}
Suppose that $\vZ_0,\ldots, \vZ_k$ are i.i.d. $d_o$ dimensional standard normal random vectors, and
$\phi_0,\ldots,\phi_k: \T_l\to \R^{d_o}$ are functions that are $L$-Lipschitz with respect to the distance $d_l$ on $\T_l$, and satisfy that $\|\phi_i(r,\s_1,\ldots,\s_l)\|\le M$ for any $0\le i\le k$, $(r,\s_1,\ldots,\s_l)\in \T_l$ (the constants $M$ and $L$ can depend on $l$). Then
$W_l:=\sup_{(r,\s_1,\ldots,\s_l)\in \T_l} \sum_{i=0}^k \l<\phi_i(r,\s_1,\ldots,\s_l),\vZ_i\r>$ satisfies that for any $0<\epsilon\le 1$,
\begin{align}\label{Wlconceq}
&\PP(W_l\ge C^{(l)}(\u,k,\epsilon))\le \epsilon \text{ for }\\
\nonumber&C^{(l)}(\u,k,\epsilon):=11(l+1)L\sqrt{(k+1)(ld+1)d_o}+\sqrt{2(k+1)Md_o \log\l(\frac{1}{\epsilon}\r)}.
\end{align}
Similarly, if $\ol{\phi}_0,\ldots,\ol{\phi}_k: \olTset_l\to \R^{d_o}$ are  $L$-Lipschitz with respect to $\ol{d}_l$, and satisfy that $\|\ol{\phi}_i(\v,\s_1,\ldots,\s_l)\|\le M$ for any $0\le i\le k$, $(\v,\s_1,\ldots,\s_l)\in \olTset_l$, then the quantity
$\ol{W}_l:=\sup_{(\v,\s_1,\ldots,\s_l)\in \olTset_l} \sum_{i=0}^k \l<\ol{\phi}_i(\v,\s_1,\ldots,\s_l),\vZ_i\r>$ satisfies that for any $0<\epsilon\le 1$,
\begin{align}\label{olWlconceq}
&\PP(\ol{W}_l\ge \ol{C}^{(l)}(\u,k,\epsilon))\le \epsilon \text{ for }\\
\nonumber&\ol{C}^{(l)}(\u,k,\epsilon):=11(l+1)L\sqrt{(k+1)l(d+1)d_o}+\sqrt{2(k+1)Md_o \log\l(\frac{1}{\epsilon}\r)}.
\end{align}
\end{lem}
\begin{proof}
For $(r,\s_1,\ldots,\s_l)\in \T_l$, let us denote 
\[W_l'(r,\s_1,\ldots,\s_l):=\l(\sum_{i=0}^{k}\l<\phi_i(r,\s_1,\ldots,\s_l), \vZ_i\r>\r)/(L\sqrt{(k+1)d_o}).\]
Then we have $W_l'(0,\vct{0},\ldots,\vct{0})=0$, and by the $L$-Lipschitz assumption on $\phi_i$, we can see that $W_l'(r,\s_1,\ldots,\s_l)-W_l'(0,\vct{0},\ldots,\vct{0})$ is a one dimensional Gaussian random variable whose variance is bounded by $\l(d_l((r,\s_1,\ldots,\s_l),(0,\vct{0}))\r)^2$. Therefore its moment generating function can be bounded as
\[\log \E(e^{\lambda (W_l'(r,\s_1,\ldots,\s_l)-W_l'(0,\vct{0},\ldots,\vct{0}))})\le \frac{\lambda^2\l[d_l((r,\s_1,\ldots,\s_l),(0,\vct{0},\ldots,\vct{0}))\r]^2}{2}.\]
This means that Dudley's entropy integral expectation bound (Corollary 13.2 of \citet*{BoLuMa2013}) is applicable here. To apply that result, we first need to upper bound the packing number $N(\delta,\T_l)$, which is the maximum number of points that can be selected in $\T_l$ such that all of them are further away from each other than $\delta$ in $d_l$ distance. It is easy to show that $N(\delta,\T_l)\le N'(\delta/2,\T_l)$, where $N'(\delta/2,\T_l)$ is the number of spheres of radius $\frac{\delta}{2}$ in $d_l$ distance needed to cover $\T_l$.

Since $\T_l\subset [0,2R]\times\Bone^l$, it follows that $N'(\delta/2,\T_l)\le N'(\delta/2,[0,2R]\times\Bone^l)$. Moreover, due to the product nature of the space $[0,2R]\times\Bone^l$ and the definition of the distance $d_l$, if we first cover $[0,2R]$ with intervals of length $R \delta/(l+1)$ (in 1 dimensional Euclidean distance), and then cover each sphere $\Bone$ in $\Bone^l$ with spheres of radius $\delta/(2l+2)$ (in $d$ dimensional Euclidean distance), then the product of any such interval and spheres will be contained in a sphere of $d_l$-radius less than or equal to $\delta/2$, and the union of all such spheres will cover $[0,2R]\times\Bone^l$. Therefore using Lemma \ref{coverspherelemma}, we obtain that for $0<\delta\le 2l+2$,
\[N(\delta,\T_l)\le N'(\delta/2,\T_l)\le \frac{4}{\delta/(l+1)}\cdot \l(\frac{c(d)}{(\delta/(2l+2))^d}\r)^l,\]
and using the fact that $\log(c(d))\le d$ for any $d\in \Z_+$, we have
\begin{align*}
H(\delta,\T):=\log(N(\delta,\T))\le (ld+1)+(ld+1)\log(2(l+1)/\delta).
\end{align*}
Using this, and the fact that the maximum $d_l$ distance of any point in $\T_l$ from $(0,\vct{0},\ldots,\vct{0})$ is bounded by $l+1$, by Corollary 13.2 of \citet*{BoLuMa2013}, we obtain that
\begin{align*}
&\E\sup_{(r,\s_1,\ldots,\s_l)\in \T_l} W_l'(r,\s_1,\ldots,\s_l)\le 12\int_{\delta=0}^{(l+1)/2} \sqrt{H(\delta,\T)} d \delta\\
&\le 12\sqrt{ld+1}\int_{\delta=0}^{(l+1)/2} \sqrt{1+\log(2(l+1)/\delta)} d \delta\\
&= 12(l+1)\sqrt{ld+1} \int_{x=0}^{1/2}\sqrt{1+\log(2/x)}dx\le 11(l+1)\sqrt{ld+1}.
\end{align*}
By the definition of $W_l'$, this implies that $\E(W_l)\le 11(l+1)\sqrt{ld+1}\cdot L\sqrt{k+1}$. Moreover, it is easy to check that the conditions of Theorem 5.8 of \citet*{BoLuMa2013} hold for $W_l$, and thus for any $t\ge 0$,
\begin{equation}
\label{W1conceq}
\PP(W_l\ge \E(W_l)+t)\le \exp\l(-\frac{t^2}{2\sigma_{W_l}^2}\r),
\end{equation}
with
\[\sigma_{W_l}^2:=\sup_{(r,\s_1,\ldots,\s_l)\in \T_l} \E\l[\l(\sum_{i=0}^{k}\l<\phi_i(r,\s_1,\ldots,\s_l), \vZ_i\r>\r)^2\r] \le \sum_{i=0}^{k} M \E\l[\|\vZ_i\|^2\r]=(k+1)M d_o \sigma_Z^2,\]
and \eqref{Wlconceq} follows. The proof of \eqref{olWlconceq} is similar.
\end{proof}

\begin{proof}[Proof of Proposition \ref{lsmlowerboundprop}]
Using Assumption \ref{assgauss1}, we know that 
\begin{equation}\sum_{i=0}^{k} \|\Phi_{t_i}(\v)-\Phi_{t_i}(\u)\|^2\ge \frac{c(\u,T)}{h} \|\v-\u\|^2,\end{equation}
thus it suffices to lower bound the terms $2\sum_{i=0}^{k}\l<\Phi_{t_i}(\v)-\Phi_{t_i}(\u), \vZ_i\r>$. 

Let $\phi_i(r,\s):=\frac{\Phi_{t_i}(\u)-\Phi_{t_i}(\u+r\s)}{r}$ for $(r,\s)\in \T_1, r>0$ ($\T_1$ was defined as in Lemma \ref{supempprocconclemma}). We continuously extend it to $r=0$ as
\[\phi_i(0,\s):=\lim_{r\to 0}\frac{\Phi_{t_i}(\u)-\Phi_{t_i}(\u+r\s)}{r}=\J\Phi_{t_i}(\u)\s.\]
Based on Lemma \ref{supempprocconclemma}, the lower bound of the random part can be obtained based on the upper bound on the quantity $W_1:=\sup_{(r,\s)\in \T} \sum_{i=0}^{k}\l<\phi_i(r,\s), \vZ_i\r>$, since
\begin{equation}
2\sum_{i=0}^{k}\l<\Phi_{t_i}(\v)-\Phi_{t_i}(\u), \vZ_i\r>\ge -2 W_1\|\v-\u\|.
\end{equation}
Now we are going to obtain bounds on the constants $L$ and $M$ of Lemma \ref{supempprocconclemma}. We have
\begin{align*}
\|\J_{\s} \phi_i(r,\s)\|&=\|\J \Phi_{t_i}(\u+r\s)\|\le \hM_1(T),\text{ and }\\
\l\|\frac{\partial}{\partial_r}\phi_i(r,\s)\r\|&=\l\|
\frac{-\J\Phi_{t_i}(\u+r\s)\cdot \s r+(\Phi_{t_i}(\u+r\s)-\Phi_{t_i}(\u))}{r^2}\r\|\le \frac{\hM_2(T)}{2},
\end{align*}
thus the $\phi_i(r,\s)$ is $L$-Lipschitz with respect to the $d_1$ distance for $L:=\hM_2(T)R+\hM_1(T)$. Moreover, from the definition of $\phi_i(r,\s)$, by \eqref{eqRkp1normbnd}, it follows that
$\|\phi_i(r,\s)\|\le M$ for $M:=\hM_1(T)$. The claim of the proposition now follows by Lemma \ref{supempprocconclemma}.
\end{proof}

\begin{proof}[Proof of Proposition \ref{lsmlsmGdiffprop}]
From the definitions, we have
\begin{align*}
\lsm(\v)&=\sum_{i=0}^{k} \|\Phi_{t_i}(\v)-\Phi_{t_i}(\u)\|^2 +2\sum_{i=0}^{k} \l<\Phi_{t_i}(\v)-\Phi_{t_i}(\u), \vZ_i\r>, \text{ and }\\
\lsmG(\v)&=\sum_{i=0}^k \|\J\Phi_{t_i}(\u)\cdot (\v-\u)\|^2+\sum_{i=0}^k\J^2\Phi_{t_i}(\u)[\v-\u,\v-\u,\vZ_i]\\
&+2\sum_{i=0}^k \l<\J\Phi_{t_i}(\u)\cdot (\v-\u),\vZ_i\r>,
\end{align*}
thus 
\begin{align}
\nonumber&|\lsm(\v)-\lsmG(\v)|\le \sum_{i=0}^k \l|\|\Phi_{t_i}(\v)-\Phi_{t_i}(\u)\|^2-\|\J\Phi_{t_i}(\u)\cdot (\v-\u)\|^2\r|\\
&\label{lsmlsmGdiffbndeq1}+2\l|\sum_{i=0}^k\l<\Phi_{t_i}(\v)-\Phi_{t_i}(\u)-\J\Phi_{t_i}(\u)\cdot (\v-\u)-\frac{1}{2}\J^2\Phi_{t_i}(\u)[\v-\u,\v-\u,\cdot], \vZ_i\r>\r|.
\end{align}
The first term in the right hand side of the above inequality can be upper bounded as
\begin{align}
\nonumber&\sum_{i=0}^k \l|\|\Phi_{t_i}(\v)-\Phi_{t_i}(\u)\|^2-\|\J\Phi_{t_i}(\u)\cdot (\v-\u)\|^2\r|\\
\nonumber&\le \sum_{i=0}^k \|\Phi_{t_i}(\v)-\Phi_{t_i}(\u)-\J\Phi_{t_i}(\u)\cdot (\v-\u)\|\cdot \l(\|\Phi_{t_i}(\v)-\Phi_{t_i}(\u)\|+\|\J\Phi_{t_i}(\u)\cdot (\v-\u)\|\r)\\
\label{lsmlsmGdifffirstermbndeq}&\le (k+1)\frac{1}{2}\hM_2(T)\|\v-\u\|^2\cdot 2 \hM_1(T) \|\v-\u\|\le \frac{\olT \hM_1(T)\hM_2(T)}{h}\|\v-\u\|^3,
\end{align}
where we have used the multivariate Taylor's expansion bound \eqref{eqRkp1normbnd}.

For the second term in \eqref{lsmlsmGdiffbndeq1},  for $(r,\s)\in \T_1$ (defined as in Lemma \ref{supempprocconclemma}), $r>0$, let
\[\phi_i(r,\s):=\l(\Phi_{t_i}(\u+r\s)-\Phi_{t_i}(\u)-\J\Phi_{t_i}(\u)\cdot \s r-\frac{1}{2}\J^2\Phi_{t_i}(\u)[r\s,r\s,\cdot]\r)/r^3.\]
For $r=0$, this can be continuously extended as
\[\phi_i(0,\s):=\lim_{r\to 0}\phi_i(r,\s)=\frac{1}{6}\J^3\Phi_{t_i}(\u)[\s,\s,\s,\cdot].\]
Similarly to Lemma \ref{supempprocconclemma}, we define
$W_1:=\sup_{(r,\s)\in \T_1}\sum_{i=0}^k\l<\phi_i(r,\s),\vZ_i\r>$, and \\$W_1':=\sup_{(r,\s)\in \T_1}\sum_{i=0}^k\l<-\phi_i(r,\s),\vZ_i\r>$, then the second term in \eqref{lsmlsmGdiffbndeq1} can be bounded as
\begin{align*}
&2\l|\sum_{i=0}^k\l<\Phi_{t_i}(\v)-\Phi_{t_i}(\u)-\J\Phi_{t_i}(\u)\cdot (\v-\u)-\frac{1}{2}\J^2\Phi_{t_i}(\u)[\v-\u,\v-\u,\cdot], \vZ_i\r>\r|\\
&\le 2\max(W_1,W_1') \|\v-\u\|^3.
\end{align*}
Based on \eqref{eqRkp1normbnd}, the partial derivatives of $\phi_i(r,\s)$ satisfy that
\begin{align*}
&\l\|\J_{\s} \phi_i(r,\s)\r\|=\frac{\l\|r\J\Phi_{t_i}(\u+r\s)-r\J\Phi_{t_i}(\u)-r\J^2\Phi_{t_i}(\u)[r\s,\s,\cdot] \r\|}{r^3}\le \frac{1}{2}\hM_3(T), \text{ and }\\
&\l\|\frac{\partial }{\partial_{r}} \phi_i(r,\s)\r\|=\Big\|\l(\J\Phi_{t_i}(\u+r\s)\cdot \s-\J\Phi_{t_i}(\u)\cdot \s-\J^2\Phi_{t_i}(\u)[r\s,\s,\cdot] \r)r^3\\
&-\l( \Phi_{t_i}(\u+r\s)-\Phi_{t_i}(\u)-\J\Phi_{t_i}(\u)\cdot \s r-\frac{1}{2}\J^2\Phi_{t_i}(\u)[r\s,r\s,\cdot]\r)3r^2\Big\|/r^6\\
&=-3\l\|\Phi_{t_i}(\u+r\s)\Phi_{t_i}(\u)-\frac{1}{3}\J\Phi_{t_i}(\u+r\s)\cdot \s r-\frac{2}{3}\J\Phi_{t_i}(\u)\cdot \s r-\frac{1}{6}\J^2\Phi_{t_i}(\u)[r\s,r\s,\cdot]\r\|/r^4\\
&\le \frac{1}{4}\hM_4(T),
\end{align*}
therefore $\phi_i$ is $L$-Lipschitz with respect to the distance $d_1$ for $L:=\frac{1}{2}\hM_3(T)+\frac{1}{2}\hM_4(T)R$, and we have $\|\phi_i(r,\s)\|\le M$ for $M:=\frac{1}{6}\hM_3(T)$. The claim of the proposition now follows by applying Lemma \ref{supempprocconclemma} to $W_1$ and $W_1'$ separately (for $\epsilon/2$ instead of $\epsilon$), and then using the union bound.
\end{proof}

\begin{proof}[Proof of Proposition \ref{lsmlsmGgraddiffprop}]
From the definitions, we have
\begin{align*}
\grad\lsm(\v)&=2\sum_{i=0}^{k} \J\Phi_{t_i}(\v)' \cdot (\Phi_{t_i}(\v)-\Phi_{t_i}(\u)) +
2\sum_{i=0}^{k} \J\Phi_{t_i}(\u)' \cdot \vZ_i, \text{ and }\\
\grad\lsmG(\v)&=2\sum_{i=0}^k \J\Phi_{t_i}(\u)'\J\Phi_{t_i}(\u)\cdot (\v-\u)+2\sum_{i=0}^k\J^2\Phi_{t_i}(\u)[\cdot,\v-\u,\vZ_i]+2\sum_{i=0}^k \J\Phi_{t_i}(\u)' \cdot \vZ_i,
\end{align*}
thus 
\begin{align}
\nonumber&\|\grad\lsm(\v)-\grad\lsmG(\v)\|\le
2\l\|\sum_{i=0}^k \l(\J\Phi_{t_i}(\v)'- \J\Phi_{t_i}(\u)'\r)\l(\Phi_{t_i}(\v)-\Phi_{t_i}(\u)\r)\r\|\\&\nonumber+2\l\|\sum_{i=0}^k  \J\Phi_{t_i}(\u)'\l(\Phi_{t_i}(\v)-\Phi_{t_i}(\u)-\J\Phi_{t_i}(\u)\cdot (\v-\u)\r)\r\|\\
&\label{lsmlsmGdiffgradbndeq}+2\l\|\sum_{i=0}^k \l(\J\Phi_{t_i}(\v)'- \J\Phi_{t_i}(\u)'-\l(\J^2\Phi_{t_i}(\u)[\v-\u,\cdot,\cdot]\r)' \r)\cdot \vZ_i \r\|.
\end{align}
Using the multivariate Taylor's expansion bound \eqref{eqRkp1normbnd}, the first two in the right hand side of the above inequality can be upper bounded as
\begin{align}
\nonumber&2\l\|\sum_{i=0}^k \l(\J\Phi_{t_i}(\v)'- \J\Phi_{t_i}(\u)'\r)\l(\Phi_{t_i}(\v)-\Phi_{t_i}(\u)\r)\r\|\\&\nonumber+2\l\|\sum_{i=0}^k  \J\Phi_{t_i}(\u)'\l(\Phi_{t_i}(\v)-\Phi_{t_i}(\u)-\J\Phi_{t_i}(\u)\cdot (\v-\u)\r)\r\|\\
\label{lsmlsmGdiffgradfirstermbndeq}&\le 4 \frac{\olT}{h} \hM_1(T)\hM_2(T) \|\v-\u\|^2.
\end{align}
For the last term in \eqref{lsmlsmGdiffgradbndeq}, for $(r,\s_1,\s_2)\in \T_2$ (defined as in Lemma \ref{supempprocconclemma}), $r>0$, we let
\[\phi_i(r,\s_1,\s_2):=
\l(\J\Phi_{t_i}(\u+r\s_1)[\s_2,\cdot]- \J\Phi_{t_i}(\u)[\s_2,\cdot]-r\J^2\Phi_{t_i}(\u)[\s_1,\s_2,\cdot]\r)/r^2.\]
We extend this continuously to $r=0$ as
\[\phi_i(0,\s_1,\s_2):=\lim_{r\to 0}\phi_i(r,\s_1,\s_2)=\frac{1}{2}\J^3\Phi_{t_i}(\u)[\s_1,\s_1,\s_2,\cdot].\]
We define $W_2$ as in Lemma \ref{supempprocconclemma} as
$W_2:=\sup_{(r,\s_1,\s_2)\in \T_2}\sum_{i=0}^k\l<\phi_i(r,\s_1,\s_2),\vZ_i\r>$,
then the last term of \eqref{lsmlsmGdiffgradbndeq} can be bounded as
\begin{equation}
2\l\|\sum_{i=0}^k \l(\J\Phi_{t_i}(\v)'- \J\Phi_{t_i}(\u)'-\l(\J^2\Phi_{t_i}(\u)[\v-\u,\cdot,\cdot]\r)' \r)\cdot \vZ_i \r\|\le 2 W_2 \|\v-\u\|^2.
\end{equation}
By \eqref{eqRkp1normbnd}, for any $(r,\s_1,\s_2)\in \T_2$, the partial derivatives of $\phi_i$ satisfy that
\begin{align*}
&\|\J_{\s_1} \phi_i(r,\s_1,\s_2)\|\le \|r\J^2\Phi_{t_i}(\u+r\s_1)- r\J^2\Phi_{t_i}(\u)\|/r^2\le \hM_3(T),\\
&\|\J_{\s_2} \phi_i(r,\s_1,\s_2)\|=\|\J\Phi_{t_i}(\u+r\s_1)- \J\Phi_{t_i}(\u)-\J^2\Phi_{t_i}(\u)[r\s_1,\cdot,\cdot]\|/r^2\le \frac{1}{2}\hM_3(T),\\
&\|\J_{r} \phi_i(r,\s_1,\s_2)\|\le 
\l\|\frac{\partial}{\partial_{r}} \l(\frac{\J\Phi_{t_i}(\u+r\s_1)- \J\Phi_{t_i}(\u)-\J^2\Phi_{t_i}(\u)[r\s_1,\cdot,\cdot]}{r^2}\r)\r\|\\
&=\bigg\|r^2(\J^2\Phi_{t_i}(\u+r\s_1)[\s_1,\cdot,\cdot]-\J^2\Phi_{t_i}(\u)[\s_1,\cdot,\cdot])\\
&-2r(\J\Phi_{t_i}(\u+r\s_1)- \J\Phi_{t_i}(\u)-\J^2\Phi_{t_i}(\u)[r\s_1,\cdot,\cdot]) \bigg\|/r^4\\
&=\bigg\|-2\bigg[\J\Phi_{t_i}(\u+r\s_1)- \J\Phi_{t_i}(\u)-\J^2\Phi_{t_i}(\u)[r\s_1,\cdot,\cdot]\\
&-\frac{1}{2}(\J^2\Phi_{t_i}(\u+r\s_1)[\s_1,\cdot,\cdot]-\J^2\Phi_{t_i}(\u)[\s_1,\cdot,\cdot])\bigg] \bigg\|/r^3\le \hM_4(T).
\end{align*}
Based on these bounds, we can see that $\phi_i$ is $L$-Lipschitz with respect to the $d_2$ distance for $L:=2R \hM_4(T)+\hM_3(T)$, and it satisfies that $\|\phi_i(r,\s_1,\s_2)\|\le M$ for $M:=\frac{1}{2}\hM_3(T)$. The claim of the proposition now follows by Lemma \ref{supempprocconclemma}.
\end{proof}

\begin{proof}[Proof of Proposition \ref{proplsmHessianlowerbnd}]
By \eqref{lsmdefeq} and \eqref{musmlsmeq}, we have
\begin{align*}
&\grad^2 \log\musm(\v|\vct{Y}_{0:k}) = \grad^2 \log q(\v) - \frac{1}{\sigma_Z^2}\grad^2 \lsm(v)=\grad^2 \log q(\v) -\frac{1}{2\sigma_Z^2} \\
&\cdot\l(2\sum_{i=0}^k\J \Phi_{t_i}(\v)'\cdot \J \Phi_{t_i}(\v) +2\sum_{i=0}^k\J^2 \Phi_{t_i}(\v)[\cdot,\cdot, \Phi_{t_i}(\v)-\Phi_{t_i}(\u)]+2\sum_{i=0}^k \J^2 \Phi_{t_i}(\v)[\cdot,\cdot,\vZ_i]\r).
\end{align*}
We first study the deterministic terms. Notice that
\[\l\|\sum_{i=0}^k\J \Phi_{t_i}(\v)'\cdot \J \Phi_{t_i}(\v)-\sum_{i=0}^k\J \Phi_{t_i}(\u)'\cdot \J \Phi_{t_i}(\u)\r\|\le \frac{2\hM_1(T)\hM_2(T) \olT \|\v-\u\|}{h}.\]
By Assumption \ref{assgauss1}, it follows that $\sum_{i=0}^k\J \Phi_{t_i}(\u)'\cdot \J \Phi_{t_i}(\u)\succeq \frac{c(\u,T)}{h}\cdot \mtx{I}_d$, thus for every $\v\in B(\u, r_H(\u,T))$, we have
\begin{align*}
&\grad^2 \log q(\v) -\frac{1}{2\sigma_Z^2} \l(2\sum_{i=0}^k\J \Phi_{t_i}(\v)'\cdot \J \Phi_{t_i}(\v) +2\sum_{i=0}^k\J^2 \Phi_{t_i}(\v)[\cdot,\cdot, \Phi_{t_i}(\v)-\Phi_{t_i}(\u)]\r)\\
&\preceq \l( C_{q}^{(2)} - \frac{3}{4}\cdot \frac{c(\u,T)}{\sigma_Z^2 h} \r)\cdot \mtx{I}_d.
\end{align*}
For the random terms, we first define $\ol{\phi}_i: \olTset_2\to \R$ ($\olTset_2$ was defined in Lemma \ref{supempprocconclemma}) for $0\le i\le k$ as 
\[\ol{\phi}_i(\v,\s_1,\s_2):=\J^2 \Phi_{t_i}(\v)[\s_1, \s_2,\cdot].\]
Based on these, we let $\ol{W}_2:=\sup_{(\v,\s_1,\s_2)\in \olTset_2}\sum_{i=0}^k\l<\ol{\phi}_i(\v,\s_1,\s_2),\vZ_i\r>$, then the random terms can be bounded as
\[\l\|-\frac{1}{2\sigma_Z^2}\cdot 2\sum_{i=0}^k\J^2 \Phi_{t_i}(\v)[\cdot,\cdot, \vZ_i]\r\|\le \frac{\ol{W}_2}{\sigma_Z^2}.\]
By its definition, it is easy to see that for every $0\le i\le k$, $\ol{\phi}_i$ is $L$-Lipschitz with respect to the $\ol{d}_2$ distance for $L:=\hM_2(T)+\hM_3(T)R$, and that $\|\ol{\phi}_i(\v,\s_1,\s_2)\|\le M$ for $M:=\hM_2(T)$ for every $(\v,\s_1,\s_2)\in \olTset_2$. The claim of the proposition now follows from Lemma \ref{supempprocconclemma}.
\end{proof}

\begin{proof}[Proof of Proposition \ref{proplsmHessianLipschitz}]
By \eqref{lsmdefeq} and \eqref{musmlsmeq}, we have
\begin{align*}
&\grad^3 \log\musm(\v|\vct{Y}_{0:k}) =\grad^3 \log(q(\v))-\frac{1}{2\sigma_Z^2}\grad^3\lsm(\v)=\grad^3 \log(q(\v)) -\frac{1}{2\sigma_Z^2}\\
&\cdot \l(6\sum_{i=0}^k \J^2 \Phi_{t_i}(\v)'\cdot \J \Phi_{t_i}(\v)+ 2\sum_{i=0}^k \J^3 \Phi_{t_i}(\v)[\cdot, \cdot, \Phi_{t_i}(\v)-\Phi_{t_i}(\u)]+2\sum_{i=0}^k\J^3 \Phi_{t_i}(\v)[\cdot,\cdot, \cdot, \vZ_i]\r).
\end{align*}
Based on the assumption on $q$, and \eqref{eqRkp1normbnd}, the deterministic terms can be bounded as
\begin{align*}
&\l\|\grad^3 \log(q(\v)) -\frac{1}{2\sigma_Z^2}\l(6\sum_{i=0}^k \J^2 \Phi_{t_i}(\v)'\cdot \J \Phi_{t_i}(\v)+ 2\sum_{i=0}^k \J^3 \Phi_{t_i}(\v)[\cdot, \cdot, \Phi_{t_i}(\v)-\Phi_{t_i}(\u)]\r)\r\|\\
&\le C_{q}^{(3)} + \frac{\olT}{\sigma_Z^2 h}\l(3\hM_1(T)\hM_2(T)+2\hM_1(T)\hM_3(T)R\r).
\end{align*}
For the random terms, we first define $\ol{\phi}_i: \olTset_3\to \R$ ($\olTset_3$ was defined in Lemma \ref{supempprocconclemma}) for $0\le i\le k$ as 
\[\ol{\phi}_i(\v,\s_1,\s_2,\s_3):=
\J^3 \Phi_{t_i}(\v)[\s_1, \s_2, \s_3,\cdot].\]
Based on this, we let $\ol{W}_3:=\sup_{(\v,\s_1,\s_2,\s_3)\in \olTset_3}\sum_{i=0}^k\l<\ol{\phi}_i(\v,\s_1,\s_2,\s_3),\vZ_i\r>$, then one can see that the random terms can be bounded as
\[\l\|-\frac{1}{2\sigma_Z^2}\cdot 2\sum_{i=0}^k\J^3 \Phi_{t_i}(\v)[\cdot,\cdot, \cdot, \vZ_i]\r\|\le \frac{\ol{W}_3}{\sigma_Z^2}.\]
By its definition, it is easy to see that for every $0\le i\le k$, $\ol{\phi}_i$ is $L$-Lipschitz with respect to the $\ol{d}_3$ distance for $L:=\hM_3(T)+\hM_4(T)R$, and that $\|\ol{\phi}_i(\v,\s_1,\s_2,\s_3)\|\le M$ for $M:=\hM_3(T)$ for every $(\v,\s_1,\s_2,\s_3)\in \olTset_3$. The claim of the proposition now follows from Lemma \ref{supempprocconclemma}.
\end{proof}

\subsection{Initial estimator when some of the components are zero}\label{secinitialestcompzero}
In Section \ref{secChoiceofF}, we have proposed a function $F$ that allows us to express the un-observed coordinates of $\u$ from the observed coordinates and their derivatives (in the two observation scenarios described in Section \ref{secapplications}). By substituting appropriate estimators of the derivatives, we obtained an initial estimator based on Theorem \ref{initialestthm}. Unfortunately, this function $F$ was not defined when some of coordinates of $\u$ are 0. In this section we propose a modified version of this estimator that overcomes this difficulty.

We start by a lemma allowing us to run the ODE \eqref{diffeqgeneralform} backwards in time (for a while).
\begin{lem}
Suppose that $\v\in \BR$, and the trapping ball assumption \eqref{eqtrappingball} holds. Then for any $0\le t<\Cder^{-1}$, the series 
\begin{equation}\label{psimintdefeq}
\Psi_{-t}(\v):=\sum_{i=0}^{\infty}\D^i \v \cdot \frac{(-t)^i}{i!}
\end{equation}
is convergent, well defined, and satisfies that $\Psi_t(\Psi_{-t}(\v))=\v$ and that for any $i_{\max}\in \N$,
\begin{equation}\label{psimintapproxdefeq}
\l\|\Psi_{-t}(\v)-\sum_{i=0}^{i_{\max}}\D^i \v \cdot \frac{(-t)^i}{i!}\r\|\le \frac{C_0 \Cder^{i_{\max}+1}}{1-\Cder t}.
\end{equation}
\end{lem}
\begin{proof}
The result follows from the bounds \eqref{uderboundeq}, and the definition of equation \eqref{diffeqgeneralform}.
\end{proof}
Based on this lemma, given the observations $\Y_{0:k}$, we propose the following initial estimator. First, select some intermediate indices $0=i_1<i_2<\ldots <i_m<k$ satisfying that $i_m\cdot h<\Cder^{-1}$. For each index $i_r$, we compute the derivative estimates $\hat{\Phi}^{(l)}$ of $\D^{l} (\u(t_{i_r}))$, and then use the function $F$ described in Section \ref{secChoiceofF} to obtain initial estimators $\widehat{\u(t_{i_r})}$ of $\u(t_{i_r})$ for $0\le r\le m$. After this, we project these estimators to $\BR$ (see \eqref{PBRdefeq}), and run them backwards by $t_{i_r}$ time units via the approximation \eqref{psimintapproxdefeq}, and project them back to $\BR$, that is, for some sufficiently large $i_{\max}\in \N$, we let
\begin{equation}\hat{\u}^{r}:=P_{\BR}\l(\sum_{i=0}^{i_{\max}}\D^i \l(P_{\BR}\l(\widehat{\u(t_{i_r})}\r)\r) \cdot \frac{(-t_{i_r})^i}{i!}\r) \text{ for }0\le r\le m.
\end{equation}
The final initial estimator $\hat{\u}$ is then chosen as the one among $(\hat{\u}^{r})_{0\le r\le m}$ that has the largest the a-posteriori probability $\musm(\hat{\u}^{r}|\Y_{0:k})$ (see \eqref{eqmusmgaussian}). Based on \eqref{initialestthm}, and some algebra, one can show that this estimator will satisfy the conditions required for the convergence of Newton's method (Theorem \ref{NewtonMAPthm}) if $\sigma_Z\sqrt{h}$ and $h$ are sufficiently small, and $i_{\max}$ is sufficiently large, as long as at least one of the vectors $\l(\u(t_{i_r})\r)_{0\le r\le m}$ has no zero coefficients. Moreover, by a continuity argument, it is possible to show that Assumption \ref{assder} holds as long as there is a $t\in [0,T]$ such that none of the coefficients of $\u(t)$ are 0.

\end{document}